\documentclass[12pt]{article}
\usepackage[margin=1.05in]{geometry}
\usepackage{amsmath,amssymb,amsthm,mathtools}
\usepackage{graphicx}
\usepackage{longtable,array}
\usepackage[colorlinks=true,linkcolor=blue,citecolor=blue]{hyperref}
\theoremstyle{plain}
\newtheorem{theorem}{Theorem}
\newtheorem{proposition}{Proposition}
\newtheorem{lemma}{Lemma}
\newtheorem{corollary}{Corollary}
\theoremstyle{plain}
\newtheorem{definition}{Definition}
\newtheoremstyle{remstyle}{\topsep}{\topsep}{}{}{\bfseries}{.}{.5em}{}
\theoremstyle{remstyle}
\newtheorem{remark}{Remark}
\newtheorem{condition}{Condition}

\newcommand{\Lap}{\triangle_k}
\newcommand{\Lapm}{\triangle_{k,\varepsilon_1}}
\newcommand{\Lapy}{\triangle_{y}}
\newcommand{\Lape}{\triangle_{\eta}}
\newcommand{\Hk}{\mathcal H_k}
\newcommand{\Hkp}{\mathcal H_k'}
\newcommand{\ip}[2]{\langle #1,#2\rangle_k}
\newcommand{\ipL}[2]{\langle #1,#2\rangle}
\newcommand{\Rea}{\operatorname{Re}}
\newcommand{\Ima}{\operatorname{Im}}
\newcommand{\dist}{\operatorname{dist}}
\newcommand{\spn}{\operatorname{span}}
\newcommand{\ord}{\operatorname{ord}}
\newcommand{\tr}{\operatorname{tr}}
\newcommand{\lp}[1]{\lambda^{\mathrm p}_{#1}}
\newcommand{\ep}[1]{e^{\mathrm p}_{#1}}
\newcommand{\tep}[1]{\tilde e^{\mathrm p}_{#1}}
\newcommand{\dd}{\,\mathrm d}

\title{Fredholm Backstepping with Heymann Pre-Feedback\\
for Linearized Navier--Stokes 2D Channel}
\author{Miroslav Krstic\thanks{Department of Mechanical and Aerospace Engineering,
University of California, San Diego, La Jolla, CA 92093-0411, USA
(\texttt{mkrstic@ucsd.edu}).}}
\date{}

\begin{document}
\maketitle

\begin{abstract}
We stabilize the linearized 2-D Poiseuille channel flow at arbitrary Reynolds numbers,
and assign the decay of every wavenumber to any user-prescribed rate. The actuation is
as in Vazquez--Krstic (IEEE TAC, 2007), VK'07 below: the two velocity components at
only one wall. This is the first Fredholm backstepping design for a Navier--Stokes
system. In VK'07 the wall-normal input renders the plant spatially causal, after which
the tangential input closes a Volterra backstepping loop. Here the normal input carries
a finite-rank feedback of arbitrarily small gain, whose sole role, in the sense of
Heymann's lemma (TAC, 1968), is to give the plant simple spectrum and make every mode
controllable from the tangential velocity. The tangential input then applies a
stabilizing feedback, obtained through an invertible Fredholm transform to a Stokes
target PDE, shifted left by an arbitrary decay. No simplicity of eigenvalues of the
Orr--Sommerfeld operator is assumed and the exclusions of plant coefficient values that
have accompanied Fredholm backstepping are removed. We state the comparison with the
Volterra design without embellishment: the Fredholm architecture is powerful, but its
gain is spectral where VK'07 is closed-form, the simpler of the two, and one from which
the Fredholm approach takes its model, its actuation, and its target.
What is demonstrated here is the reach of the Fredholm method on its most demanding
plant to date. The feedback is tested at Reynolds number $30000$, with the dominant
part of the closed-loop spectrum assigned so that it amounts to slowing the flow about
twentyfold.
\end{abstract}

\newpage
\tableofcontents
\newpage

\section{Introduction}\label{sec:intro}

\subsection{Problem}
The plane Poiseuille flow in a 2-D channel, linearized about its parabolic profile, is a
benchmark for boundary feedback stabilization of the Navier--Stokes equations: it is
linearly unstable for large Reynolds numbers, its geometry is the simplest that carries
the instability, and its control from the wall raises, in miniature, the difficulties of
turbulence suppression. Vazquez and Krstic \cite{VK07} gave the first explicit feedback
law stabilizing this flow at arbitrary Reynolds number without discretization. Their
actuation is the two velocity components at the upper wall, $u(x,1,t)=U_c$ (streamwise)
and $V(x,1,t)=V_c$ (wall-normal), and their design is sequential: the normal input $V_c$
is a dynamic feedback --- per wavenumber, a first-order filter --- chosen so that the
transformed plant becomes spatially causal in $y$ (``strict-feedback''), and the
tangential input $U_c$ then closes a backstepping loop of Volterra type, mapping the plant
to a heat equation.

In the same channel, \cite{VTC08} adds damping to the Volterra target and tracks
trajectories between parabolic profiles, with exponential decay of the error in $L^2$,
$H^1$ and $H^2$.

Rigorous boundary stabilization of this flow has also been obtained by spectral
methods. Barbu \cite{Barbu2007} controls the wall-normal velocity at both walls,
projects onto the finitely many unstable modes, and stabilizes that block by a
finite-dimensional feedback under a rank condition on the modal coefficients,
admitting geometric and algebraic multiplicity up to two. Munteanu reduces the
actuation to a single component: the wall-normal velocity at one wall
\cite{Munteanu2012}, and the tangential velocity at one wall \cite{Munteanu2011},
placing the poles of the unstable block by scalar feedback. These designs are
finite-dimensional and spectral, they act on the unstable modes only, and their
validity rests on properties of the Orr--Sommerfeld spectrum --- semi-simplicity of
the selected eigenvalues, and a rank condition that degenerates when two of them
coincide --- assumed rather than established for the channel. In bounded
two-dimensional domains a prescribed decay rate is reached by the same splitting, with
the feedback on the modes to the right of $-\omega$ obtained from a Riccati equation
\cite{Raymond2006} and the remaining modes decaying under viscosity.

This paper solves the same stabilization problem, with the same model and the same two
actuators, by a design in which the two inputs exchange their burdens:
\begin{itemize}
\item the normal input carries only the \emph{static, finite-rank} law
$V_c(t,x)=\varepsilon_1\int_0^1 e^{-\theta_1(1-y)}V(t,x,y)\dd y$, a single weighted
average of the normal velocity field --- full-state feedback returning a scalar per
wavenumber --- with a weight decaying exponentially away from the actuated wall;
\item the tangential input carries a static finite-rank term of the same type plus the
stabilizing feedback, which is constructed from an invertible transformation of
\emph{Fredholm} type --- a bounded, boundedly invertible map $T$ with $T$ minus an
isomorphism Hilbert--Schmidt --- onto a shifted Stokes target system.
\end{itemize}
The role of the static laws is that of the pre-feedback in Heymann's lemma
\cite{Heymann68}: they do not stabilize anything; they make every mode
controllable from the tangential wall velocity, for \emph{every} wavenumber $k$ and
Reynolds number $Re$, by splitting whatever multiple eigenvalues and Jordan chains the
Orr--Sommerfeld operator may possess. The Fredholm
step then assigns the entire closed-loop spectrum.

The methodological template is the two-input Fredholm design for the
Kuramoto--Sivashinsky equation of \cite{BKS26}, which in turn builds on the Fredholm
backstepping line initiated by Coron and L\"u \cite{CoronLu2015} and abstracted by Hayat and
Loko \cite{HayatLoko26}. Two features of the channel make the present problem
substantially harder than Kuramoto--Sivashinsky and are the mathematical content of this
paper. First, the plant operator is \emph{not self-adjoint}: the mean flow enters as a
bounded but skew-dominated perturbation of the Stokes operator, eigenvalues are complex,
and eigenvalue degeneracies may carry nontrivial Jordan structure --- documented for the
Orr--Sommerfeld problem at isolated parameter values \cite{Shanthini89} --- so the
splitting analysis cannot follow eigenfunctions and must be built on Riesz projections,
resolvent identities, and a multiplicity calculus of Weinstein--Aronszajn type. Second,
the state space is not $L^2$ but the $H^1$ energy space of the normal velocity, and the
pre-feedback changes the space itself, not only the domain: it replaces the homogeneous
boundary condition at the actuated wall by a weighted average of the normal velocity
across the channel, and $\triangle_k=\partial_{yy}-4\pi^2k^2$ must be inverted under that
condition.

Measured against the Fredholm backstepping literature, with
\cite{CoronLu2015,CoronHXZ22,HayatLoko26} as the yardstick, this is the most demanding instance of
the method executed to date: the first Navier--Stokes plant, the first non-normal plant
with a mass operator in its evolution, and, together with \cite{BKS26}, the first
designs in this line that carry no hypotheses on the plant's spectrum --- the
plant-side exclusion sets present since \cite{Cerpa2010,CoronLu2015} are moved onto the
design parameters, where they exclude nothing of consequence.

\subsection{Contributions}
\begin{enumerate}
\item \emph{Fredholm backstepping reaches the Navier--Stokes equations.} This is not
merely an extension from 1D to 2D. It is an extension from plants whose spectrum is
given to a plant whose spectrum is not: the Orr--Sommerfeld operator is non-normal, its
eigenvalues are complex and may carry Jordan chains of any admissible length, and its
time derivative appears under a second operator that must be inverted before the design
begins, so the Riesz basis on which Fredholm backstepping rests is constructed here
rather than inherited. What carries the method across is Heymann's lemma,
applied at the boundary of a partial differential equation: a wall feedback of
vanishing gain makes every mode controllable from the other wall velocity component,
and the Fredholm transformation then assigns the whole spectrum. Stabilizing this flow
is not itself new --- it is done in \cite{VK07} by Volterra backstepping and in
\cite{Barbu2007,Munteanu2011,Munteanu2012} by spectral feedback; what is new is that
the Fredholm method now reaches it.
\item \emph{Stabilization of the channel without any spectral hypothesis.} The
spectral designs for this flow reduce stabilization to a finite-dimensional
controllability condition on the modes they select
\cite{Barbu2007,Munteanu2011,Munteanu2012}, which for a single wall velocity component
requires those modes to be distinct; every previous Fredholm design excludes plants
outright \cite{Cerpa2010,CoronLu2015,HayatLoko26}. Whether the Orr--Sommerfeld
operator meets such conditions at all $(k,Re)$ is not established, and degeneracies
are documented at isolated parameter values \cite{Shanthini89}. Here nothing is
required of it: the two wall velocity components are used to convert whatever
multiplicities and Jordan chains the plant carries into simple spectrum with every
mode controllable from the tangential component, and every exclusion falls on
parameters the designer chooses before seeing the flow.
\item \emph{Why the design uses two wall velocity components, and two stages.} The
temporal Orr--Sommerfeld spectrum of this flow is not semisimple: \cite{Shanthini89}
finds infinitely many parameter combinations at which two modes coalesce into a single
mode of order two. Earlier Fredholm designs assume that away; here the second wall
velocity component converts whatever Jordan chains the plant carries into simple
spectrum with every mode controllable from the tangential one, and a single rank-one
stage cannot do that for every chain configuration, so there are two.
\end{enumerate}

The stabilization is proved wavenumber by wavenumber and assembled into the physical
channel with the pressure reconstructed; the rate is arbitrary, and the band of
actuated wavenumbers widens with it.
The controller and the main theorem are stated in Section \ref{sec:main1}, before any
analysis, in the form the PDE backstepping reader expects: the feedback through the
wall traces of a gain kernel, the kernel posed --- formally --- by a kernel PDE on the
square, and a stability estimate in the $L^2$ energy of the velocity pair.
What the paper demonstrates is two-sided, and both sides are reported as findings. It
demonstrates the reach of Fredholm backstepping: the channel is stabilized without
imposing spatial causality on the plant and with no spectral hypotheses.
And it demonstrates, by executing the alternative, what the Volterra design of
\cite{VK07} secured: a closed-form synthesis, and higher-norm
stability, which the present spectral construction does not match. The Volterra design
is the simpler one, and it is the design from which this paper takes its model, its
actuation, and its target. Section \ref{sec:status} gives the grounds for that
comparison.

\subsection{An Analogy-Based Roadmap for the Finite-Dimensional Control Reader}
\label{sec:fd-roadmap}

This subsection translates the architecture of the paper into the language of
finite-dimensional linear control, so that the design can be read by a control theorist
who does not work in PDEs. Read Navier--Stokes as a very large descriptor LTI system
\(E\dot x=Ax+B_1u_1+B_2u_2\): the horizontal and vertical velocity components are
linked algebraically through the pressure, which carries no time derivative and is not
a state. A first transformation, of state and inputs together, decomposes this system
into independent subsystems, designed one at a time and reassembled at the end.

The plant is first shown to be controllable from the two inputs. A static
pre-feedback, in the manner of Heymann's lemma \cite{Heymann68}, then makes the
eigenvalues distinct and every mode controllable from one designated input. This
pre-feedback does not stabilize; it turns a two-input plant with repeated eigenvalues
and Jordan blocks into a one-input plant. It is applied in two stages because one
stage moves only one Jordan chain.

Fredholm backstepping is a similarity transformation \(T\) onto a target:
\(T A_{\rm cl}T^{-1}=S-qI\). The target \(S\) is the dissipative Stokes dynamics, which
in the variables used later is heat dynamics for the streamwise derivative of
vorticity, and \(q\) is user-assigned. With 150,000 states in the simulation of
Section \ref{sec:numerics}, no one chooses that many pole locations; here none are
chosen. The entire target spectrum is shifted left by the user-assignable \(q\).

Two features have no matrix counterpart. The plant is highly non-normal, so its
eigenvectors can be badly conditioned: in the numerical example the feedback gain is of order ten while the
similarity that produces it has condition number about \(10^7\). And a matrix has
finitely many modes, whereas here the infinite tail must be shown harmless:
eigenvectors still give usable coordinates, feedback coefficients remain summable, the
similarity is bounded and invertible, the uncontrolled subsystems are uniformly stable,
and the reconstructed pressure remains square-integrable.

Table~\ref{tab:fd-big} gives the design and analysis concepts at the coarsest level.
A second roadmap, Table~\ref{tab:fd-detail}, is placed immediately before the proofs,
for the reader who wants to enter them.

\begin{table}[t]
\centering\small
\caption{Big-picture roadmap in finite-dimensional control language.}
\label{tab:fd-big}
\begin{tabular}{p{0.20\textwidth}p{0.55\textwidth}p{0.18\textwidth}}
\hline
\textbf{Paper operation} & \textbf{Finite-dimensional control reading} & \textbf{Where established}\\
\hline
Plant representation & Start from a descriptor system \(E\dot x=Ax+B_1u_1+B_2u_2\), not standard \(\dot x=Ax+Bu\). A coordinate/input transformation decomposes it into independent subsystems. & Section~\ref{sec:model}; Definition~\ref{def:ops}\\
Two inputs & The wall-normal and the tangential wall velocity. Together they control the plant. & Lemmas~\ref{lem:trace}, \ref{lem:inj}; Theorem~\ref{thm:2ctrl}\\
Why one input can fail & At a repeated eigenvalue with a two-dimensional eigenspace, a scalar input necessarily misses a left-eigenvector direction. & Proposition~\ref{prop:oneinput}\\
Heymann pre-feedback & Static feedback through both inputs makes the eigenvalues distinct and every mode controllable from one of the two inputs. It does not stabilize. & Theorem~\ref{thm:spectral}; Propositions~\ref{lem:step1}, \ref{lem:step2}\\
Fredholm backstepping & Construct an invertible state similarity \(T\) taking the controlled plant to a target. & Theorems~\ref{thm:riesz}, \ref{thm:main}\\
Target and arbitrary rate & Choose \(S-qI\), where \(S=S^\top<0\). Increasing \(q\) shifts the entire target spectrum left; this removes the need for individual pole placement. & Definition~\ref{def:target}; Theorem~\ref{thm:main}\\
Reassembly & Reassemble the designed subsystems, leave sufficiently fast ones uncontrolled, return to the original variables. & Proposition~\ref{lem:mean}; Lemmas~\ref{lem:sector}, \ref{lem:pgraph}; Theorem~\ref{thm:phys}\\
Global conclusion & Well-posed closed loop decays at prescribed rate. & Theorem~\ref{thm:first}\\
\hline
\end{tabular}
\end{table}

\section{Linearized channel flow and its two wall inputs}\label{sec:model}

This section is self-contained; it reproduces the model, the linearization, and the
Orr--Sommerfeld reduction of \cite[Secs.~II, V]{VK07}, which we adopt verbatim, together
with the wall actuation.

Consider a 2-D incompressible channel flow in
$\Omega=(-\infty,\infty)\times[0,1]$, with streamwise velocity $U$, wall-normal velocity
$V$, pressure $P$, governed by the Navier--Stokes equations
\begin{equation}\label{eq:NS}
U_t=\tfrac1{Re}(U_{xx}+U_{yy})-UU_x-VU_y-P_x,\qquad
V_t=\tfrac1{Re}(V_{xx}+V_{yy})-UV_x-VV_y-P_y,
\end{equation}
and the continuity equation $U_x+V_y=0$, where $Re$ is the Reynolds number. The
equilibrium is the parabolic Poiseuille profile
\begin{equation}\label{eq:poiseuille}
U^e=4y(1-y),\qquad V^e=0,\qquad P^e=P_0-\tfrac{8}{Re}\,x ,
\end{equation}
which is linearly unstable for large $Re$ \cite{SchmidHenningson}. Writing the
fluctuations $u=U-U^e$, $V$, $p=P-P^e$ and linearizing about \eqref{eq:poiseuille} gives
the linearized equations
\begin{align}
u_t&=\tfrac1{Re}(u_{xx}+u_{yy})+4y(y-1)u_x+4(2y-1)V-p_x,\label{eq:linu}\\
V_t&=\tfrac1{Re}(V_{xx}+V_{yy})+4y(y-1)V_x-p_y,\label{eq:linV}\\
0&=u_x+V_y,\label{eq:cont}
\end{align}
with boundary conditions
\begin{equation}\label{eq:bc}
u(x,0)=V(x,0)=0,\qquad u(x,1)=U_c(x),\qquad V(x,1)=V_c(x),
\end{equation}
and the pressure boundary conditions obtained by evaluating \eqref{eq:linV} at the walls,
as in \cite[(15)--(16)]{VK07}. The two actuation variables $U_c$ (tangential) and $V_c$
(normal) live on the upper wall $y=1$; no actuation occurs inside the channel or at the
lower wall.

Taking the Laplacian of \eqref{eq:linV} and eliminating the pressure by the Poisson
equation that $p$ satisfies yields the autonomous Orr--Sommerfeld equation for the normal
velocity,
\begin{equation}\label{eq:OSxy}
\triangle V_t=\tfrac1{Re}\triangle^2V+4y(y-1)\triangle V_x-8V_x ,
\end{equation}
with boundary conditions, from \eqref{eq:bc} and continuity \eqref{eq:cont} evaluated at
the walls,
\begin{equation}\label{eq:OSbcxy}
V(x,0)=V_y(x,0)=0,\qquad V(x,1)=V_c(x),\qquad V_y(x,1)=-(U_c)_x(x).
\end{equation}
Fourier transform in $x$, $f(k,y)=\int_{\mathbb R}f(x,y)e^{-2\pi ikx}\dd x$, so
$\partial_x\mapsto2\pi ki$; the $L^2$ norm is preserved (Parseval), and different
wavenumbers decouple by spatial invariance. Write
\begin{equation}\label{eq:Lapdef}
\Lap\coloneqq\partial_{yy}-4\pi^2k^2 .
\end{equation}
For $k\ne0$, \eqref{eq:OSxy}--\eqref{eq:OSbcxy} become, per wavenumber,
\begin{gather}
\Lap V_t=\mathcal L_kV,\qquad
\mathcal L_k\coloneqq\tfrac1{Re}\Lap^2+8\pi ki\,y(y-1)\Lap-16\pi ki ,\label{eq:OSk}\\
V(k,0)=V_y(k,0)=0,\qquad V(k,1)=V_c(k),\qquad V_y(k,1)=-2\pi ki\,U_c(k).\label{eq:OSkbc}
\end{gather}
The streamwise velocity is recovered from continuity: $u=V_y/(-2\pi ki)$. Thus at each
$k\ne0$ the plant is the single scalar equation \eqref{eq:OSk} with the \emph{two} scalar
inputs
\begin{equation}\label{eq:channels}
\text{normal wall velocity: the trace }V(k,1),\qquad
\text{tangential wall velocity: the trace }V_y(k,1),
\end{equation}
both free at the actuated wall. At the lower wall $u=V=0$: the no-slip/no-penetration
conditions of fluid mechanics, which in the Orr--Sommerfeld normal-velocity variable
become $V=V_y=0$; we call them the no-slip wall conditions throughout. This is the
actuation of
a fourth-order operator at one end: two Cauchy data are imposed at $y=0$, two are
available for feedback at $y=1$.

\medskip
\noindent\textbf{The design.} Fix real design parameters
\begin{equation}\label{eq:params}
c\in\mathbb R\ \text{(ratio of the two components)},\quad
\theta_1,\theta_2>0\ \text{(gain shapes)},\quad
\varepsilon_1,\varepsilon_2\in\mathbb R\ \text{(strengths)},
\end{equation}
and the gain profiles $g_i(y)\coloneqq e^{-\theta_i(1-y)}$, decaying away from the
actuated wall. The pair $(\theta_2,\varepsilon_2)$ parametrizes the second stage; it
enters the design through the pre-compensated operator of Theorem
\ref{thm:spectral}, whose simple spectrum, Riesz basis and reachability depend on it,
and not through the kernel system \eqref{eq:kpde}. Throughout, $\ipL{f}{h}\coloneqq\int_0^1f(y)\overline{h(y)}\dd y$ denotes the
$L^2(0,1)$ inner product, and $\ip{\cdot}{\cdot}$ the energy inner product
\eqref{eq:Hk}; since $g_1,g_2$ are real, $\ipL{V}{g_i}=\int_0^1V g_i\dd y$ is linear in
$V$. The feedback laws are
\begin{equation}\label{eq:prefeed}
V_c(k)=\varepsilon_1\,\ipL{V}{g_1},\qquad
U_c(k)=\frac{i}{2\pi k}\Big[c\,\varepsilon_1\,\ipL{V}{g_1}+\varepsilon_2\,\ipL{V}{g_2}+F_k[V]\Big],
\end{equation}
where $F_k$ is the Fredholm feedback functional constructed in
Section~\ref{sec:design}. Equivalently, by \eqref{eq:OSkbc}, the boundary conditions in
closed loop read
\begin{equation}\label{eq:clbc}
V(1)=\varepsilon_1\ipL{V}{g_1},\qquad
V_y(1)=c\,\varepsilon_1\ipL{V}{g_1}+\varepsilon_2\ipL{V}{g_2}+F_k[V].
\end{equation}
The static parts of \eqref{eq:prefeed} are the Heymann pre-feedback; $F_k$ carries the
stabilization. In physical space the pre-feedback is
$V_c(t,x)=\varepsilon_1\int_0^1e^{-\theta_1(1-y)}V(t,x,y)\dd y$ (restricted to the
controlled wavenumber band of Section~\ref{sec:phys}); its kernel is real, so the
reality of the physical control is automatic (Lemma~\ref{lem:real}).

\medskip
\noindent\textbf{Problem.} For each $k\ne0$, $Re>0$, construct $F_k$ and admissible
parameters \eqref{eq:params} such that the closed loop
\eqref{eq:OSk}, \eqref{eq:clbc} is exponentially stable at a prescribed decay rate in the
$L^2$ norm of the velocity pair $(u,V)$, through an invertible transformation to a stable
target system --- with every restriction placed on \eqref{eq:params} and none on
$(k,Re)$.

\section{Controller and main result}\label{sec:main1}

\subsection{Controller and stabilization theorem}\label{sec:main-thm}

The channel is taken periodic of period $L$ in $x$,
$\Omega_L\coloneqq(\mathbb R/L\mathbb Z)\times(0,1)$, with Fourier components
$V(k,y,t)=\int_0^LV(t,x,y)e^{-2\pi ikx}\dd x$, $k\in\mathbb Z/L$, to which Section
\ref{sec:model} applies verbatim per wavenumber. Given a decay rate $\omega$, the feedback acts on the finitely many wavenumbers of the
band $\mathcal B$ of Theorem \ref{thm:first}: those whose open-loop decay is not already at
least $\omega$, the modes outside decaying at $2\pi^2k^2/Re\ge\omega$ under viscosity
alone. The band selects the slow wavenumbers, not the unstable ones; the long waves
are stable but slow, and stabilizing them at the rate $\omega$ is what makes the
field estimate rapid. At the mean mode $k=0$ incompressibility and $V(0)=0$ force $V\equiv0$, so that mode
is invisible in the normal velocity and is stabilized by tangential actuation acting
on $u$. Its plant is the one-dimensional heat equation $u_t=u_{yy}/Re$ with Dirichlet
actuation at the upper wall: no incompressibility coupling survives there, hence none
of the spatial nonlocality that forces the Fredholm construction at $k\ne0$, and the
mean mode is stabilized at an arbitrary rate by Volterra backstepping with the
classical kernel \cite[Ch.~4]{SGK2009} appearing in \eqref{eq:Ucphys}. The paper
therefore carries two kernels of different provenance: the Fredholm kernels
$\mathcal K^{(k)}$, constructed spectrally, for $k\ne0$, and this closed-form Volterra
kernel for the mean mode. Wavenumbers do not appear in the main theorem, which
is stated in the physical variables; one wavenumber at a time is treated only inside
the design, where the superscript in $\mathcal K^{(k)}$ marks the
dependence (dropped when $k$ is fixed).

Three objects are used in this section ahead of their systematic development. The inner
product $\ip{V}{W}=\ipL{V_y}{W_y}+4\pi^2k^2\ipL{V}{W}$ on $H^1(0,1)$, with norm
$\|\cdot\|_k$: by continuity $u=V_y/(-2\pi ki)$, so
$\|u\|^2+\|V\|^2=\|V\|_k^2/(4\pi^2k^2)$ at each wavenumber, and $\|\cdot\|_k$ is the
$L^2$ energy of the velocity pair (Section \ref{sec:setting}). The Stokes operator
$S_k$ of the channel --- the plant \eqref{eq:OSk} with the mean flow removed, under
no-slip conditions at both walls --- with real, simple eigenvalues $\lambda_1>\lambda_2>\cdots\to-\infty$
and eigenfunctions $e_j$, $\ip{e_i}{e_j}=\delta_{ij}$ (Section \ref{sec:stokes}). And
the advection coefficients of \eqref{eq:OSk},
$a(y)\coloneqq8\pi ki\,y(y-1)$, $b\coloneqq-16\pi ki$.

The design transforms the closed-loop state, per wavenumber, by a Fredholm integral
transformation
the Fredholm transformation
\begin{equation}\label{eq:transf}
\alpha(k,y,t)=V(k,y,t)-\int_0^1\mathcal K^{(k)}(y,\eta)\,V(k,\eta,t)\dd\eta ,
\end{equation}
with kernel $\mathcal K^{(k)}$ to be found, onto the target system
\begin{equation}\label{eq:targ}
\Lap\alpha_t=\frac1{Re}\Lap^2\alpha-q\,\Lap\alpha,
\qquad \alpha(0)=\alpha_y(0)=\alpha(1)=\alpha_y(1)=0:
\end{equation}
the Stokes system shifted left by a constant $q$ of the designer's choice, the
wavenumber argument of $\alpha$ and the superscript of $\mathcal K$ being dropped
whenever $k$ is fixed. With $q=0$ the target is the unshifted Stokes system --- the
heat-type equation obtained by removing the mean flow, the target of \cite{VK07}. The
shift moves the target spectrum to $\{\lambda_n-q\}_{n\ge1}$; the mean flow is removed
entirely by the transformation, and $q\ge\omega$ makes the target decay at the
prescribed rate: the shift is what provides ``rapid'' in the present construction.

The design parameters $c,\theta_2,\varepsilon_2$ of \eqref{eq:prefeed} do not enter the
physical form of the feedback; they enter the construction of its gain, through
Theorem \ref{thm:spectral}. For $k\ne0$ the
modal controller is carried by the wall traces of the kernel, and at $k=0$ by the
Volterra kernel $k_0$:
\begin{gather}
V_c(k)=\varepsilon_1\int_0^1e^{-\theta_1(1-\eta)}V(k,\eta,t)\dd\eta,\qquad
U_c(k)=\frac{i}{2\pi k}\int_0^1\mathcal K^{(k)}_y(1,\eta)\,V(k,\eta,t)\dd\eta ,
\label{eq:kfeed}\\
U_c(0)=\int_0^1k_0(1,\eta)\,u(0,\eta,t)\dd\eta,\qquad
k_0(y,\eta)\coloneqq-c_0Re\,\eta\,
\frac{I_1\big(\sqrt{c_0Re(y^2-\eta^2)}\big)}{\sqrt{c_0Re(y^2-\eta^2)}},\quad c_0\ge\omega,
\label{eq:kfeed0}
\end{gather}
with $I_1$ the modified Bessel function of the first kind.

\begin{theorem}[Rapid stabilization of linearized channel flow]\label{thm:first}
Let $Re>0$, $L>0$ and a decay rate $\omega>0$ be given. There exist a gain shape
$\theta_1>0$ outside a discrete subset of $(0,\infty)$, a strength
$\varepsilon_1\ne0$ of arbitrarily small modulus, a mean-mode damping $c_0\ge\omega$,
a shift $q\ge\omega$ in \eqref{eq:targ}, and a real gain $\mathcal G$ such that the
feedback
\begin{align}
V_c(t,x)&=\varepsilon_1\int_0^L\!\!\int_0^1\chi(x-\xi)\,e^{-\theta_1(1-y)}\,V(t,\xi,y)\dd y\dd\xi,
\qquad \chi(x)\coloneqq\frac2L\sum_{k\in\mathcal B,\,k>0}\cos(2\pi kx),
\label{eq:Vcphys}\\
U_c(t,x)&=\int_0^L\!\!\int_0^1\mathcal G(x-\xi,\eta)\,V(t,\xi,\eta)\dd\eta\dd\xi
+\frac1L\int_0^L\!\!\int_0^1k_0(1,\eta)\,u(t,\xi,\eta)\dd\eta\dd\xi ,
\label{eq:Ucphys}
\end{align}
whose gain assembles the modal laws \eqref{eq:kfeed},
\begin{equation}\label{eq:synth}
\mathcal G(x,\eta)=\frac1L\sum_{k\in\mathcal B\setminus\{0\}}e^{2\pi ikx}\,
\frac{i}{2\pi k}\,\mathcal K^{(k)}_y(1,\eta),
\end{equation}
with $\mathcal K=\mathcal K^{(k)}$ transcribed, formally, by the kernel system for the
pair $(\mathcal K,\mathcal M)$, $\mathcal M$ tied to $\mathcal K$ by the mass operator,
on $(0,1)^2$:
\begin{subequations}\label{eq:kpde}
\begin{align}
&\frac1{Re}\Lape^2\mathcal M+\Lape\big(a(\eta)\mathcal M\big)+q\,\Lape\mathcal M+b\mathcal M
-\frac1{Re}\Lapy^2\mathcal K\notag\\
&\qquad
+\frac1{Re}\Lape\mathcal M(y,1)\,\mathcal K_y(1,\eta)
-\frac{\varepsilon_1}{Re}\partial_\eta\Lape\mathcal M(y,1)\,e^{-\theta_1(1-\eta)}\notag\\
&\hspace{10em}=\;\big(a(y)+q\big)\,\Lape\delta(y-\eta)+b\,\delta(y-\eta),\label{eq:kpde-a}\\[2pt]
&\Lapy\mathcal K(y,\eta)=\Lape\mathcal M(y,\eta),\label{eq:kpde-b}\\
&\mathcal K(0,\eta)=0,\qquad \mathcal K_y(0,\eta)=0,\qquad \mathcal K(1,\eta)=\varepsilon_1e^{-\theta_1(1-\eta)},\label{eq:kpde-c}\\
&\mathcal M(y,0)=\mathcal M_\eta(y,0)=\mathcal M(y,1)=\mathcal M_\eta(y,1)=0,\label{eq:kpde-d}\\
&\int_0^1\cosh(2\pi k\eta)\,\Lapy\mathcal K(y,\eta)\dd\eta=
\int_0^1\sinh(2\pi k\eta)\,\Lapy\mathcal K(y,\eta)\dd\eta=0 .\label{eq:kpde-e}
\end{align}
\end{subequations}
applied to \eqref{eq:linu}--\eqref{eq:bc} on $\Omega_L$, generates an analytic
semigroup on $X$; for every initial field in $X$ the closed-loop solution exists and is
unique in the class $C([0,\infty);X)\cap C^1((0,\infty);X)$ of $X$-valued functions
solving the closed loop, is real for real data, has pressure
$p(\cdot,t)\in L^2(\Omega_L)$ for $t>0$, and satisfies
\begin{equation}\label{eq:decay-first}
\|u(\cdot,t)\|^2_{L^2(\Omega_L)}+\|V(\cdot,t)\|^2_{L^2(\Omega_L)}
\le C\,e^{-2\omega t}\,\big(\|u(\cdot,0)\|^2_{L^2}+\|V(\cdot,0)\|^2_{L^2}\big),
\qquad t\ge0,
\end{equation}
where $\Lapy,\Lape$ denote $\partial^2-4\pi^2k^2$ in the indicated variable, $\delta$
is the Dirac distribution on the diagonal, $a,b$ are the advection coefficients of
\eqref{eq:OSk}, the unknowns are $(\mathcal K,\mathcal M)$ and the tangential gain is the wall slope
$\mathcal K_y(1,\cdot)$, the $\eta$-integrals in \eqref{eq:Ucphys}, \eqref{eq:kfeed} and
\eqref{eq:synth} denote the action of bounded linear functionals on the $k$-th Fourier
mode of $V$,
$\mathcal B\coloneqq\{k\in\mathbb Z/L:|k|<K\}$ with
$K\coloneqq\max\{8Re,\pi^{-1}\sqrt{\omega Re/2}\}$, $k_0$ is the mean-mode kernel of
\eqref{eq:kfeed0},
\begin{multline}\label{eq:Xspace}
X\coloneqq\Big\{(u,V)\in L^2(\Omega_L)^2:\ u_x+V_y=0;\ \ V(k,\cdot)(0)=0\ \ \forall k;\\
V(k,\cdot)(1)=\varepsilon_1\!\int_0^1\!e^{-\theta_1(1-\eta)}V(k,\eta)\dd\eta\ \ \forall k\in\mathcal B\setminus\{0\};\ \
V(k,\cdot)(1)=0\ \ \forall k\notin\mathcal B\Big\}
\end{multline}
with the norm $\|u\|_{L^2}^2+\|V\|_{L^2}^2$, and $C\ge1$ is an overshoot coefficient
depending on the design parameters.
\end{theorem}

The theorem is proved at the end of Section \ref{sec:phys}, where the modal designs
are assembled. The system \eqref{eq:kpde} is the formal transcription of the operator
identity established in Section \ref{sec:cl}: its interior equation matches the plant and
target operators through the transformation, \eqref{eq:kpde-c} enforces
$\alpha(0)=\alpha_y(0)=\alpha(1)=0$, and $\alpha_y(1)=0$ is the tangential feedback
itself.

\begin{remark}\label{rem:diag}
The pair $(\mathcal K,\mathcal M)$ is one transformation seen in two variables:
$\mathcal K$ is the kernel of the integral part $I-T$ of the backstepping
transformation $\alpha=TV$ acting on the normal velocity, and $\mathcal M$ is the
kernel of the corresponding transformation of the vorticity
$\omega=-\Lap V/(2\pi ki)$. Writing $\mathbb K$ and $\mathbb M$ for the integral
operators with these kernels, \eqref{eq:kpde-b} with the clamps \eqref{eq:kpde-d} is
the kernel transcription of the conjugacy
$(I-\mathbb M)\Lap=\Lap(I-\mathbb K)$.
\end{remark}

\begin{remark}[Comparison with the Volterra design of \cite{VK07}]\label{rem:vort}
Both designs transform the streamwise derivative of a physical quantity into the same
damped reaction--diffusion equation, band-limited to the actuated wavenumbers by
$\chi(x)=\frac2L\sum_{k\in\mathcal B,\,k>0}\cos(2\pi kx)$. Here the quantity is the
streamwise derivative of the vorticity $\omega=u_y-V_x$,
\begin{gather}
\widetilde\omega_x=\int_0^L\chi(x-\zeta)\left[\omega_x(t,\zeta,y)
-\int_0^1\!\!\int_0^L\mathcal M(\zeta-\xi,y,s)\,\omega_x(t,\xi,s)\dd\xi\dd s\right]\dd\zeta,
\label{eq:vorttransf}\\
\big(\widetilde\omega_x\big)_t=\frac1{Re}\Delta\,\widetilde\omega_x
-q\,\widetilde\omega_x,\qquad
\int_0^L\!\!\int_0^1\mathcal S_j(x-\xi,\eta)\,\widetilde\omega_x(t,\xi,\eta)\dd\eta\dd\xi=0,
\quad j=0,1 .\label{eq:vorttarget}
\end{gather}
In \eqref{eq:vorttransf} the $s$-integration runs to $1$: the transformation is
Fredholm. The design of \cite{VK07}, with damping added, is the same construction on
the streamwise derivative of the streamwise velocity,
\begin{gather}
\widetilde u_x=\int_0^L\chi(x-\zeta)\left[u_x(t,\zeta,y)
-\int_0^y\!\!\int_0^L\mathcal K_{\mathrm V}(\zeta-\xi,y,s)\,u_x(t,\xi,s)\dd\xi\dd s\right]\dd\zeta,\label{eq:vktransf}\\
\big(\widetilde u_x\big)_t=\frac1{Re}\Delta\,\widetilde u_x-q\,\widetilde u_x,\qquad
\widetilde u_x(t,x,0)=\widetilde u_x(t,x,1)=0,\label{eq:vktarget}
\end{gather}
There the $s$-integration runs only to $y$: the transformation is Volterra, with
$\mathcal K_{\mathrm V}$ assembled over the band from the Volterra kernels of
\cite{VK07}. Both targets decay at a rate not less than $q$. The functions entering
\eqref{eq:vorttransf}, \eqref{eq:vorttarget} are
\begin{gather}
\mathcal M(x,y,s)=\frac1L\sum_{k\in\mathcal B,\ k\ne0}e^{2\pi ikx}
\big(\partial_y^2-4\pi^2k^2\big)\!\!\int_0^1\mathcal K^{(k)}(y,\eta)\,G_k(\eta,s)\dd\eta,\\
G_k(\eta,s)=-\frac{\sinh\big(2\pi|k|\min(\eta,s)\big)\,
\sinh\big(2\pi|k|(1-\max(\eta,s))\big)}{2\pi|k|\,\sinh(2\pi|k|)},\\
\mathcal S_0(x,\eta)=\frac1L\sum_{k\in\mathcal B,\ k\ne0}e^{2\pi ikx}\sinh\big(2\pi k(1-\eta)\big),
\qquad
\mathcal S_1(x,\eta)=\frac1L\sum_{k\in\mathcal B,\ k\ne0}e^{2\pi ikx}\sinh(2\pi k\eta).
\end{gather}
\end{remark}

\subsection{Why two wall inputs, and what they achieve}\label{sec:main-why}

Theorem \ref{thm:first} is the result of the paper; this subsection states the two
results that account for its design: why the stabilizer uses both wall velocity
components, and what the static pre-feedback on the normal wall velocity accomplishes.
Both are used in Sections \ref{sec:part-spectral}--\ref{sec:part-fredholm} and proved
there: Theorem \ref{thm:spectral} in
Section \ref{sec:part-spectral}, Theorem \ref{thm:2ctrl} and Proposition
\ref{prop:oneinput} in Section \ref{sec:necessity}.

At a geometrically double plant eigenvalue, one component of the wall velocity leaves
a mode that no feedback through it can move:

\begin{proposition}[A double eigenspace is immovable from one wall velocity component]\label{prop:oneinput}
Suppose $z_*\in\sigma(A_k)$ has geometric multiplicity $2$. Then for every linear
functional $F$ (defined on the relevant domain), the tangential-only closed loop
\[
\Lap V_t=\mathcal L_kV,\qquad V(0)=V_y(0)=V(1)=0,\qquad V_y(1)=F[V]
\]
admits the solution $e^{z_*t}\phi$ for some $\phi\ne0$: $z_*$ remains a closed-loop
eigenvalue. The same holds for the normal-only closed loop $V_y(1)=0$, $V(1)=G[V]$.
Consequently, if $\Rea z_*\ge0$, no linear state feedback through a single component
of the wall velocity stabilizes the plant. Moreover $z_*$ is uncontrollable from either
component alone: there is $\zeta\ne0$ in $\ker(A_k^*-\bar z_*)$ annihilated by the
corresponding boundary functional \eqref{eq:alphabeta}, so $z_*$ remains in the
spectrum of every closed loop built on that one component, dynamic controllers
included.
\end{proposition}

The hypothesis of Proposition \ref{prop:oneinput} is met where two eigenvalues of
distinct symmetry coincide. The Poiseuille profile is symmetric about $y=1/2$, so the
Orr--Sommerfeld operator commutes with $y\mapsto1-y$ and its eigenfunctions split into
even and odd families; within one family a double eigenspace is
non-generic, but a coincidence between the two families is two real conditions on the
two real parameters $(k,Re)$. The degeneracies documented for this operator are of the
other kind: the coalescences of \cite{Shanthini89} produce a mode of order two, of
geometric multiplicity one, and pair symmetric with symmetric and antisymmetric with
antisymmetric. Whether the hypothesis of Proposition \ref{prop:oneinput} is ever met
for this plant is open. The design that follows
assumes nothing about the multiplicities, so the question does not arise for it.

The two components together give approximate controllability at every wavenumber and
Reynolds number:

\begin{theorem}[Two-input approximate controllability]\label{thm:2ctrl}
Fix $k\ne0$, $Re>0$, $T>0$. For controls
$V_c,\tau\in C_c^\infty(0,T)$ let $V(\cdot)$ solve \eqref{eq:OSk} from the zero initial
state, with wall data $V(t)(0)=V_y(t)(0)=0$, $V(t)(1)=V_c(t)$, $V_y(t)(1)=\tau(t)$.
Then the reachable set
$\{V(T)\}\subset\Hk$ is dense in $\Hk$.
\end{theorem}

A static pre-feedback through the normal wall velocity sharpens this into
controllability from the tangential wall velocity alone, mode by mode, and supplies the
spectral structure the Fredholm construction uses.

The state space of the plant at wavenumber $k$ is $\Hk\coloneqq\{V\in H^1(0,1):V(0)=V(1)=0\}$ with the norm $\|\cdot\|_k$, and
the pre-compensated system is the plant \eqref{eq:OSk} under the static pre-feedback
\eqref{eq:prefeed}, on the state space
\begin{equation}\label{eq:Hkp-early}
\Hkp\coloneqq\Big\{V\in H^1(0,1):\ V(0)=0,\ \ V(1)=\varepsilon_1\ipL{V}{g_1}\Big\},
\qquad g_i(y)=e^{-\theta_i(1-y)} ,
\end{equation}
with generator $A_k'$. The tangential frequency response $\Psi^{(2)}(z)$ is the state
produced by a unit tangential wall velocity at complex frequency $z$, and $b'_n$ are
its residue coordinates in the eigenbasis of $A_k'$.

\begin{theorem}[Static pre-feedback makes every plant tangentially controllable]\label{thm:spectral}
Fix $k\ne0$ and $Re>0$. There exist a finite set $\mathcal C\subset\mathbb R$, for each
$c\notin\mathcal C$ a discrete set $\Theta_1\subset(0,\infty)$, for each
$\theta_1\notin\Theta_1$ a threshold $\bar\varepsilon_1>0$, for each
$0<|\varepsilon_1|\le\bar\varepsilon_1$ a discrete set $\Theta_2\subset(0,\infty)$, and
for each $\theta_2\notin\Theta_2$ a threshold $\bar\varepsilon_2>0$, such that for every
$0<|\varepsilon_2|\le\bar\varepsilon_2$ the generator $A_k'$ of the pre-compensated
system \eqref{eq:Hkp-early} has the following properties.
\begin{enumerate}
\item[(i)] \emph{Simple spectrum.} $\sigma(A_k')$ consists of algebraically simple
eigenvalues $\{\mu'_n\}_{n\ge1}$, whatever the multiplicities and Jordan structure of
the plant.
\item[(ii)] \emph{Riesz basis.} The eigenvectors $\{\chi'_n\}$ form a Riesz basis of
$\Hkp$.
\item[(iii)] \emph{Controllability of every mode.} For every $n$, the map
$z\mapsto\Psi^{(2)}(z)$ has a simple pole at $\mu'_n$ with residue $b'_n\chi'_n$, and
$b'_n\ne0$.
\end{enumerate}
All conditions fall on the design parameters $(c,\theta_1,\varepsilon_1,\theta_2,\varepsilon_2)$;
none fall on the plant.
\end{theorem}

What follows constructs the Fredholm transformation $T$ --- and hence $\mathcal K$,
in the pairing sense of Remark \ref{rem:diag} --- and proves Theorem
\ref{thm:first}. The paper then has four parts. Section \ref{sec:part-spectral}
proves Theorem \ref{thm:spectral}: the static pre-feedback gives every plant ---
Jordan structure included --- simple spectrum, a Riesz basis of eigenvectors, and
controllability of every mode from the tangential wall velocity; its opening paragraph maps the internal steps, and a reader
prepared to grant Theorem \ref{thm:spectral} can pass over it. Section
\ref{sec:necessity} situates the design: the two wall inputs jointly control every
plant, while one alone cannot. Section \ref{sec:part-fredholm} builds the Fredholm
transformation onto the shifted Stokes target and proves closed-loop stability.
Section \ref{sec:phys} assembles the closed loop in the physical variables and
completes the proof of Theorem \ref{thm:first}; Section \ref{sec:status} records what
remains open. Some proofs are collected in Appendix \ref{app:proofs}.

\subsection{Technical building blocks, as they would arise in finite
dimension}\label{sec:fd-blocks}

Table \ref{tab:fd-detail} is a guide for the finite-dimensional control reader through
Sections \ref{sec:part-spectral}--\ref{sec:phys} and Appendix \ref{app:proofs}.

\begin{table}[p]
\centering\small
\caption{Technical roadmap: matrix mechanisms and the extra closure required by infinitely many modes.}
\label{tab:fd-detail}
\begin{tabular}{>{\raggedright\arraybackslash}p{0.205\textwidth}
                  >{\raggedright\arraybackslash}p{0.545\textwidth}
                  >{\raggedright\arraybackslash}p{0.18\textwidth}}
\hline
\textbf{Matrix/control mechanism} &
\textbf{What the paper proves} &
\textbf{Where established}\\
\hline
\multicolumn{3}{l}{\textbf{I. Mechanisms with direct finite-dimensional counterparts}}\\[-2pt]
\hline
Similarity built from the frequency response &
The similarity is between the closed loop and the target, \(T(A+bF)T^{-1}=S-qI\), so
\(T\) and the feedback row \(F\) are determined together.  The columns of \(T^{-1}\)
are \((sI-A)^{-1}b\) evaluated at the target eigenvalues, the \(n\)th of them divided
by its own \(n\)th component \(\beta_n\); the entries of \(F\) in target coordinates
are the \(1/\beta_n\). &
Lemma~\ref{lem:expand2};
Proposition~\ref{prop:qc2};
Lemma~\ref{lem:gamman}\\
Return difference of a rank-two feedback &
Static feedback through the two inputs is a rank-two update, so its closed-loop
eigenvalues are the zeros of a \(2\times2\) return difference.  That determinant factors
into the two successive scalar stages, and a winding-number count tracks every
eigenvalue gained or lost. &
Lemma~\ref{lem:Xi};
Propositions~\ref{prop:resAp}, \ref{prop:factor};
Lemma~\ref{lem:counting}\\
Splitting a Jordan chain &
One rank-one feedback fans a chain of length \(\kappa\) into \(\kappa\) distinct
eigenvalues, a second splits the chain the first leaves behind, and the gains that fail
are isolated. &
Lemma~\ref{lem:disc};
Propositions~\ref{lem:step1}, \ref{lem:step2};
Lemma~\ref{lem:t2}\\
PBH through residues &
The input coefficients are the analogs of \(z^*b_1\) and \(z^*b_2\); for a Jordan block
the Laurent coefficients of \((sI-A)^{-1}b\) say which levels of the chain an input
reaches. &
Lemmas~\ref{lem:trace}, \ref{lem:resform},
\ref{lem:chains}, \ref{lem:inj};
Proposition~\ref{prop:bmc}\\
Normal reference plus non-normal perturbation &
Write \(A=S+B\) with \(S\) normal and explicitly diagonalized and \(B\) the non-normal
part.  Well-separated modes stay simple and close to their reference values. &
Propositions~\ref{prop:stokes-all}, \ref{prop:kato}\\[2pt]
\hline
\multicolumn{3}{l}{\textbf{II. Pathologies that appear only with infinitely many modes}}\\[-2pt]
\hline
Do the eigenvectors still give coordinates? &
For a finite square eigenvector matrix this ends at nonsingularity.  With infinitely
many modes the analog requires more: the coordinates must remain stable and the columns
must span, so that \(T\) is bounded with bounded inverse. &
Lemma~\ref{lem:bm};
Corollary~\ref{cor:decomp};
Proposition~\ref{prop:complete2};
Theorem~\ref{thm:riesz}\\
Can the tail defeat a finite-rank design? &
The estimates on the remote modes must be uniform in the mode number, or the
finite-rank argument fails at infinity. &
Lemma~\ref{lem:tails};
Propositions~\ref{prop:high}, \ref{prop:spectral-quant}\\
Can the feedback row be unbounded? &
A matrix feedback row is bounded automatically.  Here its infinitely many coefficients
must be square-summable. &
Lemma~\ref{lem:gamman}\\
Can the uncontrolled subsystems spoil the rate? &
Only finitely many are controlled.  The rest must be shown, uniformly, to decay faster
than the requested rate. &
Lemma~\ref{lem:sector}\\
\hline
\end{tabular}
\end{table}

\section{Spectral controllability by static pre-feedback}\label{sec:part-spectral}

This part proves that the static pre-feedback of Section \ref{sec:model}, acting
through both wall velocity components,
turns every plant, whatever its Jordan structure, into an operator with simple
spectrum, a Riesz basis of eigenvectors, and controllability of every mode from the
tangential wall velocity. The result is Theorem \ref{thm:spectral}, stated in Section
\ref{sec:main1}; the Fredholm construction uses only what its statement and the
asymptotics of Proposition \ref{prop:spectral-quant} provide, namely the eigenvalues
$\mu'_n$, the basis $\{\chi'_n\}$, and $b'_n\asymp n/Re$. Within this part,
Sections \ref{sec:setting}--\ref{sec:liftings} develop the spectra of the Stokes and
plant operators and the boundary functionals of the two wall velocity components,
Section \ref{sec:precomp} constructs the pre-compensated operator and the determinant
whose zeros are its eigenvalues, and Sections \ref{sec:split}--\ref{sec:spectral}
prove that every one of those eigenvalues is simple and that no others exist. Proofs
of the supporting results are collected in Appendix \ref{app:proofs}.

\subsection{Energy space, plant operators, and wall functionals}\label{sec:setting}

Fix $k\ne0$. Let $\|\cdot\|$ denote the $L^2(0,1)$ norm. Set
\begin{equation}\label{eq:Hk}
\Hk\coloneqq H^1_0(0,1),\qquad
\ip{V}{W}\coloneqq\ipL{V_y}{W_y}+4\pi^2k^2\ipL{V}{W},
\end{equation}
and use the formula \eqref{eq:Hk} for $\ip{\cdot}{\cdot}$ and the norm $\|\cdot\|_k$ on
all of $H^1(0,1)$. For any $V\in H^1$ and $W\in H^1_0$, one integration by parts gives
$\ip{V}{W}=-\ipL{\Lap V}{W}$, the right side read as the $H^{-1}\times H^1_0$ pairing.
The map $\Lap:H^1_0\to H^{-1}\coloneqq(H^1_0)^*$ is an isometry for $\|\cdot\|_k$ and its
dual norm; $\Lap^{-1}$ denotes its inverse, which restricts to the Dirichlet inverse
$L^2\to H^2\cap H^1_0$. Since $u=V_y/(-2\pi ki)$, purely algebraically,
\begin{equation}\label{eq:uV}
\|u\|^2+\|V\|^2=\frac1{4\pi^2k^2}\,\|V\|_k^2\qquad(V\in H^1,\ V(0)=0),
\end{equation}
so the $\|\cdot\|_k$ norm of the normal velocity \emph{is} the $L^2$ energy of the
velocity pair, with or without boundary conditions at $y=1$. Also
$\|V\|\le\|V\|_k/\sqrt{\pi^2+4\pi^2k^2}$ for $V\in H^1_0$, and
$\|V\|\le C_k\|V\|_k$ for $V\in H^1$ with $V(0)=0$ (Poincar\'e from the no-slip wall).

\begin{lemma}[Elliptic estimate]\label{lem:ell}
There is $C=C(k)$ such that for all $V\in H^3(0,1)$,
\begin{equation}\label{eq:ell}
\|V\|_{H^3}\le C\big(\|\Lap^2V\|_{H^{-1}}+\|V\|_{H^1}\big).
\end{equation}
\end{lemma}

\begin{definition}[Plant operators]\label{def:ops}
Let $\mathcal D\coloneqq H^3(0,1)\cap H^2_0(0,1)$. On $\mathcal D$ define
\begin{equation}\label{eq:ops}
S_kV\coloneqq\tfrac1{Re}\Lap^{-1}\Lap^2V,\qquad
B_kV\coloneqq\Lap^{-1}\big(8\pi ki\,y(y-1)\Lap V-16\pi ki\,V\big),\qquad
A_k\coloneqq S_k+B_k .
\end{equation}
\end{definition}

\begin{proposition}[Plant operators]\label{prop:ops}
\begin{enumerate}
\item[(i)] $S_k$ with domain $\mathcal D$ is self-adjoint and negative definite in
$\Hk$ with compact resolvent, and $H^4\cap H^2_0$ is a core. Its eigenvalues are real
and satisfy $\lambda<-4\pi^2k^2/Re$, its eigenfunctions are in $C^\infty[0,1]$, and the
normalized eigenfunctions $\{e_n\}_{n\ge1}$ form an orthonormal basis of $\Hk$.
\item[(ii)] $B_k$ extends to a bounded operator on $\Hk$. Hence $A_k$ with domain
$\mathcal D$ is closed, generates an analytic semigroup on $\Hk$, has compact
resolvent, and $\sigma(A_k)\subset\{\lambda:\Rea\lambda\le\|B_k\|\}$. The same holds
for $A_k^*=S_k+B_k^*$ on $\mathcal D$. Root vectors of $A_k$ and $A_k^*$ lie in
$H^4\cap H^2_0$.
\end{enumerate}
\end{proposition}

The following identity identifies the boundary functionals through which the two wall
velocity components \eqref{eq:channels} act; it is used throughout the paper.

\begin{lemma}[Trace identity]\label{lem:trace}
Let $W\in H^3(0,1)$ with $W(0)=W_y(0)=0$, and let $\zeta\in H^4\cap H^2_0$. Then, with
$\ipL{\mathcal L_kW}{\zeta}$ read as the $H^{-1}\times H^1_0$ pairing,
\begin{equation}\label{eq:trace}
\ipL{\mathcal L_kW}{\zeta}
=\frac1{Re}\,W_y(1)\,\overline{\zeta_{yy}(1)}
-\frac1{Re}\,W(1)\,\overline{\zeta_{yyy}(1)}
+\ipL{W}{\mathcal L_k^*\zeta},
\end{equation}
where
\begin{equation}\label{eq:Lkstar}
\mathcal L_k^*\zeta\coloneqq\tfrac1{Re}\Lap^2\zeta-8\pi ki\,\Lap\big(y(y-1)\zeta\big)+16\pi ki\,\zeta .
\end{equation}
Consequently the adjoint boundary functionals of the tangential and normal wall velocities are
\begin{equation}\label{eq:alphabeta}
\alpha[\zeta]\coloneqq\tfrac1{Re}\,\overline{\zeta_{yy}(1)}\qquad\text{and}\qquad
\beta[\zeta]\coloneqq\tfrac1{Re}\,\overline{\zeta_{yyy}(1)} ,
\end{equation}
and $A_k^*=\Lap^{-1}\mathcal L_k^*$ on $H^4\cap H^2_0$, a core for $A_k^*$.
\end{lemma}

\begin{lemma}[Pairing against clamped functions]\label{lem:pairconv}
Let $\zeta\in H^2_0(0,1)$ and $W\in H^1(0,1)$. Then
$\ip{W}{\zeta}=-\ipL{W}{\Lap\zeta}$.
\end{lemma}

\subsection{Stokes spectrum}\label{sec:stokes}

\begin{proposition}[Stokes eigenproblem]\label{prop:stokes-all}
Let $k\ne0$.
\begin{enumerate}
\item[(i)] \emph{Characteristic equation.} $\lambda<-4\pi^2k^2/Re$ is an eigenvalue of
$S_k$ if and only if $\sigma\coloneqq\sqrt{-4\pi^2k^2-Re\,\lambda}>0$ satisfies
\begin{equation}\label{eq:char}
4\pi k\big(1-\cosh(2\pi k)\cos\sigma\big)=\frac{\sigma^2-4\pi^2k^2}{\sigma}\,\sinh(2\pi k)\,\sin\sigma ,
\end{equation}
the eigenfunction being, up to a constant,
\begin{equation}\label{eq:efn}
e=a\,f+b\,g,\qquad
f(y)=\cosh(2\pi ky)-\cos(\sigma y),\qquad
g(y)=\sinh(2\pi ky)-\frac{2\pi k}{\sigma}\sin(\sigma y),
\end{equation}
with $a\,f(1)+b\,g(1)=0$.
\item[(ii)] \emph{Enumeration, asymptotics and gaps.} Equation \eqref{eq:char} has no
root in $(0,\pi)$ and exactly one root $\sigma_n$ in each interval
$\big(n\pi,(n+1)\pi\big)$, $n\ge1$, and it has no others; the eigenvalues
$\lambda_n\coloneqq-(\sigma_n^2+4\pi^2k^2)/Re$ are therefore enumerated in decreasing
order, $\lambda_1>\lambda_2>\cdots\to-\infty$. There is $n_0=n_0(k)$ such that for
$n\ge n_0$ the root $\sigma_n$ is simple and
\begin{equation}\label{eq:sigman}
\sigma_n=(n+1)\pi+\epsilon_n,\qquad
\epsilon_n=\frac{4\pi k\big((-1)^{n+1}-\cosh2\pi k\big)}{(n+1)\pi\sinh2\pi k}+O(n^{-2}) .
\end{equation}
Moreover,
\begin{equation}\label{eq:gap}
d_n\coloneqq\dist\big(\lambda_n,\sigma(S_k)\setminus\{\lambda_n\}\big)\ge\frac{c_0\,n}{Re}\quad(n\ge n_0),\qquad
\sum_{n\ge n_0}d_n^{-2}<\infty ,
\end{equation}
and for $n\ge n_0$ and all $j\ne n$,
\begin{equation}\label{eq:gapjm}
|\lambda_j-\lambda_n|\ge\frac{c_0'}{Re}\,|j^2-n^2| ,
\end{equation}
with $c_0,c_0'>0$ depending on $k$ only, and $|\lambda_j|\le C(k)\,j^2/Re$ for all $j$.
\item[(iii)] \emph{Eigenfunction and trace estimates.} Normalize $e_n$ in
\eqref{eq:efn} by $\|e_n\|_k=1$ and $a>0$. For $n\ge n_0$,
\begin{equation}\label{eq:efest}
a_n=\frac{\sqrt2}{\sigma_n}\big(1+O(n^{-1})\big),\qquad
\|e_n\|=O(n^{-1}),\qquad
e_{n,yy}(1)=(-1)^{n+1}\sqrt2\,\sigma_n\big(1+O(n^{-1})\big),
\end{equation}
\begin{equation}\label{eq:efest3}
e_{n,yyy}(1)=(-1)^{n+1}\sqrt2\,\sigma_n\,\frac{2\pi k}{\sinh2\pi k}\Big(\cosh2\pi k-(-1)^{n+1}\Big)\big(1+O(n^{-1})\big),
\end{equation}
and $\ipL{h}{e_n}=O(n^{-1})$ for every fixed $h\in C^1([0,1])$.
\item[(iv)] \emph{Trace nondegeneracy and simplicity.} No eigenfunction of $S_k$
satisfies $e_{yy}(1)=0$, and every eigenvalue of $S_k$ is simple.
\end{enumerate}
\end{proposition}

\subsection{Plant spectrum}\label{sec:plant}

Two facts about the resolvent are used repeatedly: $\sigma(A_k)\subset\{z\in\mathbb C:\ \dist(z,\sigma(S_k))\le\|B_k\|\}$, and if $\dist(z,\sigma(S_k))\ge2\|B_k\|$ then $\|(z-A_k)^{-1}\|\le2/\dist(z,\sigma(S_k))$. Indeed, if $\dist(z,\sigma(S_k))>\|B_k\|$ then $\|(z-S_k)^{-1}\|=\dist(z,\sigma(S_k))^{-1}<\|B_k\|^{-1}$ by self-adjointness, so $z-A_k=(I-B_k(z-S_k)^{-1})(z-S_k)$ is invertible; the Neumann series gives the bound.

For an isolated eigenvalue $z_0$ of a closed operator $A$ with compact resolvent,
$\nu(z_0;A)$ denotes its algebraic multiplicity, and a Jordan chain of $A$ at $z_0$ is
a finite sequence $\zeta_0,\dots,\zeta_{\kappa-1}$ of nonzero vectors with
\begin{equation}\label{eq:jordan}
A\zeta_0=z_0\zeta_0,\qquad A\zeta_p=z_0\zeta_p+\zeta_{p-1}\quad(1\le p\le\kappa-1);
\end{equation}
$\kappa$ is the length of the chain, $\nu(z_0;A)$ is the sum of the lengths of the
chains in a basis of the root subspace, and $z_0$ is semisimple exactly when every
chain has length one.

The spectrum is treated in two parts: the finitely many eigenvalues of $A_k$ in a
bounded region, called low, and the remaining ones, called high, which lie one in each
disk $\{|\lambda-\lambda_m|<d_m/2\}$, $m\ge n_1$, of Proposition \ref{prop:kato}
below.

\begin{proposition}[Perturbed spectrum]\label{prop:kato}
Let $n_1\ge n_0$ be such that $d_n>4\|B_k\|$ for $n\ge n_1$. For each $n\ge n_1$ the
operator $A_k$ has exactly one eigenvalue $\lp n$ in $\{|\lambda-\lambda_n|<d_n/2\}$, it
is algebraically simple, and its normalized eigenvector $\ep n$ and the eigenvector
$\tep n$ of $A_k^*$ at $\overline{\lp n}$, normalized by $\ip{\ep n}{\tep n}=1$, satisfy
\begin{equation}\label{eq:kato}
\|\ep n-e_n\|_k\le\frac{C}{d_n},\qquad \|\tep n-e_n\|_k\le\frac{C}{d_n},\qquad
|\lp n-\lambda_n|\le\|B_k\| ,
\end{equation}
with $C$ depending on $\|B_k\|$ only. The remaining spectrum of $A_k$ is contained in
$\{|z|\le r_{\mathrm{low}}\}$, $r_{\mathrm{low}}\coloneqq\tfrac12(|\lambda_{n_1-1}|+|\lambda_{n_1}|)$,
and consists of finitely many eigenvalues of total algebraic multiplicity $n_1-1$,
carried by the Riesz projection $P_{\mathrm{low}}$ over
$\Gamma_{\mathrm{low}}=\{|z|=r_{\mathrm{low}}\}$. Moreover, on the circles
$\Gamma_n=\{|z-\lambda_n|=d_n/2\}$, $\|(z-A_k)^{-1}\|\le4/d_n$.
\end{proposition}

\begin{proof}
For $\lambda\in\Gamma_n$, $\|(\lambda-S_k)^{-1}\|=2/d_n$ by self-adjointness, so
$\|B_k(\lambda-S_k)^{-1}\|\le2\|B_k\|/d_n<\tfrac12$ and
$(\lambda-A_k)^{-1}=(\lambda-S_k)^{-1}\sum_{j\ge0}[B_k(\lambda-S_k)^{-1}]^j$, giving
$\Gamma_n\subset\rho(A_k)$ and the resolvent bound. The Riesz projections
$P_n=\frac1{2\pi i}\oint_{\Gamma_n}(\lambda-A_k)^{-1}\dd\lambda$ and
$P_n^0=\ip{\cdot}{e_n}e_n$ satisfy $\|P_n-P_n^0\|\le\frac{4\|B_k\|}{d_n}<1$, so
$\operatorname{rank}P_n=1$, and the single eigenvalue inside is simple. Set
$\ep n\coloneqq P_ne_n/\|P_ne_n\|_k$; then $\|\ep n-e_n\|_k\le2\|P_n-P_n^0\|$. Repeat for
$A_k^*=S_k+B_k^*$; since $\ip{P_ne_n}{e_n}=1+O(1/d_n)$, the biorthogonal normalization
changes $\tep n$ by a factor $1+O(1/d_n)$. Eigenvalue bound: by the resolvent bound above,
$\lp n$ lies within $\|B_k\|$ of some $\lambda_j$; since $|\lp n-\lambda_n|<d_n/2$ and
$d_n/2>2\|B_k\|$, that $\lambda_j$ is $\lambda_n$. For the low block, take
$P_{\mathrm{low}}$ over the vertical line $\Rea z=h\coloneqq\tfrac12(\lambda_{n_1-1}+\lambda_{n_1})$
closed to the right, where the resolvent difference is $O(|z|^{-2})$ and the closing arc
contributes nothing. With $g\coloneqq\lambda_{n_1-1}-\lambda_{n_1}\ge d_{n_1}>4\|B_k\|$,
$\dist(h+it,\sigma(S_k))\ge\big((g/2)^2+t^2\big)^{1/2}\eqqcolon\delta(t)\ge2\|B_k\|$, so
$\|(z-A_k)^{-1}\|\le2/\delta(t)$ and
$\|(z-A_k)^{-1}-(z-S_k)^{-1}\|=\|(z-A_k)^{-1}B_k(z-S_k)^{-1}\|\le2\|B_k\|/\delta(t)^2$;
integrating, $\|P_{\mathrm{low}}-P^0_{\mathrm{low}}\|\le2\|B_k\|/g<1$ with
$P^0_{\mathrm{low}}=\sum_{n<n_1}P^0_n$, which gives
$\operatorname{rank}P_{\mathrm{low}}=n_1-1$. Every $z$ with $|z|>r_{\mathrm{low}}$
outside $\bigcup_{n\ge n_1}\{|z-\lambda_n|<d_n/2\}$ has
$\dist(z,\sigma(S_k))\ge d_{n_1}/2>\|B_k\|$, hence lies in $\rho(A_k)$.
\end{proof}

\begin{lemma}[Bari--Markus]\label{lem:bm}
Let $\{Q_n\}_{n\ge1}$ be mutually disjoint projections on a Hilbert space,
$Q_mQ_n=\delta_{mn}Q_n$, and $\{Q_n^0\}$ mutually orthogonal projections with
$\sum_nQ_n^0=I$, $\operatorname{rank}Q_n^0<\infty$, $\|Q_n-Q_n^0\|<1$ for every $n$, and
$\sum_n\operatorname{rank}Q_n^0\,\|Q_n-Q_n^0\|^2<\infty$. Then $I+K$, defined by $Ke_{n,i}\coloneqq(Q_n-Q_n^0)e_{n,i}$ on orthonormal bases
$\{e_{n,i}\}_i$ of $\operatorname{ran}Q_n^0$, is bounded, boundedly invertible, and
Hilbert--Schmidt away from the identity, and it carries $\{e_{n,i}\}$ to
$\{Q_ne_{n,i}\}$; consequently $\{Q_ne_{n,i}\}_{n,i}$ is a Riesz basis and
$\sum_nQ_nx=x$ for every $x$, with unconditional convergence.
\end{lemma}

\begin{corollary}[Unconditional spectral decomposition of the plant]\label{cor:decomp}
Without any hypothesis on the low spectrum,
\begin{equation}\label{eq:decomp}
P_{\mathrm{low}}V+\sum_{n\ge n_1}P_nV=V\qquad\text{for every }V\in\Hk,
\end{equation}
the series converging unconditionally in $\Hk$,
and the family $\{v_1,\dots,v_{n_1-1}\}\cup\{\ep n\}_{n\ge n_1}$ is a Riesz basis of
$\Hk$, where $\{v_j\}$ is the image under $P_{\mathrm{low}}$ of an orthonormal basis of
$\operatorname{ran}P^0_{\mathrm{low}}$. The coordinate functionals of the high modes are
$V\mapsto\ip{V}{\tep n}$.
\end{corollary}

\begin{proof}
Apply Lemma \ref{lem:bm} with $Q_{\mathrm{low}}=P_{\mathrm{low}}$,
$Q^0_{\mathrm{low}}=P^0_{\mathrm{low}}$ and $Q_n=P_n$, $Q_n^0=P^0_n$ ($n\ge n_1$): the
projections are mutually disjoint as Riesz projections of disjoint spectral sets,
$\|P_n-P^0_n\|\le4\|B_k\|/d_n<1$, and
$\sum_{n\ge n_1}\|P_n-P^0_n\|^2\le16\|B_k\|^2\sum_nd_n^{-2}<\infty$ by \eqref{eq:gap}.
The lemma's conclusion already contains completeness --- the family is the image of an
orthonormal basis under a bounded invertible map --- so \eqref{eq:decomp} follows. For $n\ge n_1$, $P_n=\ip{\cdot}{\tep n}\ep n$ by rank one and
biorthogonality.
\end{proof}

No simplicity of the finitely many low eigenvalues of $A_k$ is assumed anywhere in this
paper; \eqref{eq:decomp} is the only completeness statement used, and it is
unconditional.

\subsection{Frequency responses to wall velocity and their residue coefficients}\label{sec:liftings}

Throughout, $\ell$ is defined by
\begin{equation}\label{eq:ell2}
\ell(y)\coloneqq\frac{\sinh(2\pi ky)}{\sinh(2\pi k)},\qquad \Lap\ell=0,\quad\ell(0)=0,\quad\ell(1)=1 .
\end{equation}

\begin{definition}[Frequency responses to wall velocity]\label{def:Psi}
For $z\in\rho(A_k)$ let $\Psi(z),\rho(z)\in H^4(0,1)$, the states produced by a unit
tangential, respectively normal, wall velocity at complex frequency $z$, be the unique
solutions of
\begin{align}
\mathcal L_k\Psi=z\Lap\Psi,&\qquad \Psi(0)=\Psi_y(0)=\Psi(1)=0,\quad \Psi_y(1)=1,\label{eq:Psi}\\
\mathcal L_k\rho=z\Lap\rho,&\qquad \rho(0)=\rho_y(0)=\rho_y(1)=0,\quad \rho(1)=1,\label{eq:rho}
\end{align}
and for the ratio $c\in\mathbb R$ of \eqref{eq:params} set $\rho_c(z)\coloneqq\rho(z)+c\,\Psi(z)$,
so that $\rho_c(1)=1$, $\rho_{c,y}(1)=c$.
\end{definition}

These problems have unique solutions in $H^4(0,1)$: Uniqueness: the difference of two solutions lies in $H^4\cap H^2_0\subset\mathcal D$ and in $\ker(z-A_k)=\{0\}$. Existence: with $W_\Psi\coloneqq-y^2(1-y)$, respectively $W_\rho\coloneqq y^2(3-2y)$, which carry the four boundary data, seek $\Psi=W_\Psi+w$ with $w\in\mathcal D$: the requirement is $(z-A_k)w=\Lap^{-1}(\mathcal L_k-z\Lap)W_\Psi$, solvable since $z\in\rho(A_k)$, and the eigen-ODE bootstrap of Proposition \ref{prop:ops}(ii) gives $w\in H^4$; likewise for $\rho$. Analyticity in $z$ follows from Lemma \ref{lem:resform}.

\begin{lemma}[Resolvent formulas]\label{lem:resform}
For $z,z_0\in\rho(A_k)$,
\begin{align}
\Psi(z)&=\Psi(z_0)-(z-z_0)(z-A_k)^{-1}\Psi(z_0),\label{eq:resPsi}\\
\rho_c(z)&=\rho_c(z_0)-(z-z_0)(z-A_k)^{-1}\big(\rho_c(z_0)-\ell\big).\label{eq:resrho}
\end{align}
In particular $\dot\Psi(z)=-(z-A_k)^{-1}\Psi(z)$ and
$\dot\rho_c(z)=-(z-A_k)^{-1}(\rho_c(z)-\ell)$, and both responses extend to $H^1$-valued
functions on $\mathbb C$ that are analytic except for poles, the poles lying in
$\sigma(A_k)$ with order at $z_*$ at most the maximal Jordan chain length of $A_k$ at
$z_*$.
\end{lemma}

\begin{lemma}[Rational chain pairings]\label{lem:chains}
Let $z_*\in\sigma(A_k)$ and let $\zeta_0,\dots,\zeta_{\kappa-1}\in H^4\cap H^2_0$ be a
Jordan chain \eqref{eq:jordan} of $A_k^*$ at $\bar z_*$. Define the channel
traces $\alpha_p\coloneqq\alpha[\zeta_p]$, $\beta_p\coloneqq\beta[\zeta_p]$ from
\eqref{eq:alphabeta}, and the first-stage functional
$\beta^c_p\coloneqq\beta_p-c\,\alpha_p$. Then, for all $z\in\rho(A_k)$,
\begin{equation}\label{eq:xy}
-\ipL{\rho_c(z)}{\Lap\zeta_p}=\sum_{i=0}^p\frac{\beta^c_{p-i}}{(z-z_*)^{i+1}},\qquad
-\ipL{\Psi(z)}{\Lap\zeta_p}=\sum_{i=0}^p\frac{-\alpha_{p-i}}{(z-z_*)^{i+1}} .
\end{equation}
\end{lemma}

For a vector $v$ in the root subspace $\operatorname{ran}P_{z_*}$ (which lies in
$\mathcal D\subset H^1_0$), the pairing $-\ipL{v}{\Lap\zeta}$ equals $\ip{v}{\zeta}$;
choosing chain bases $\{\phi^{(r)}_a\}$ of $A_k$ at $z_*$ and $\{\zeta^{(r)}_b\}$ of
$A_k^*$ at $\bar z_*$ in the standard biorthogonal normalization
\begin{equation}\label{eq:biorth}
\ip{\phi^{(r)}_a}{\zeta^{(s)}_b}=\delta_{rs}\,\delta_{a+b,\kappa_r-1}
\qquad(0\le a\le\kappa_r-1,\ 0\le b\le\kappa_s-1).
\end{equation}
Such bases exist: the root subspace of $A_k$ at $z_*$ pairs nondegenerately under
$\ip{\cdot}{\cdot}$ with that of $A_k^*$ at $\bar z_*$ (the pairing annihilates every
other root subspace, and the decomposition \eqref{eq:decomp} leaves no residual
annihilator), the dual Jordan structures have equal block sizes, and dual chains
chosen recursively, block by block, make each Gram matrix the anti-diagonal identity.
With this normalization, the coefficient of $\phi^{(r)}_a$ in the principal part of the meromorphic function
$z\mapsto\rho_c(z)$ at $z_*$ is the principal part of
$x^{(r)}_{\kappa_r-1-a}(z)$, and likewise for $\Psi$. In particular, the leading
$(z-z_*)^{-\kappa_{\max}}$ coefficient of $\rho_c$ is
\begin{equation}\label{eq:topcoeff}
\sum_{r:\ \kappa_r=\kappa_{\max}}\beta^c[\zeta^{(r)}_0]\;\phi^{(r)}_0 .
\end{equation}

For the high modes, Lemma \ref{lem:chains} with $\kappa=1$ and $\zeta=\tep m$ gives the
scalar identities, valid for all $z\in\rho(A_k)$ and $m\ge n_1$,
\begin{equation}\label{eq:tailpair}
\ip{\Psi(z)}{\tep m}=\frac{b_m}{z-\lp m},\qquad
-\ipL{\rho_c(z)}{\Lap\tep m}=\frac{\beta^c[\tep m]}{z-\lp m},
\end{equation}
where
\begin{equation}\label{eq:bmdef}
b_m\coloneqq-\alpha[\tep m]=-\frac1{Re}\overline{\tep m{}_{,yy}(1)},\qquad
\beta_m\coloneqq\beta[\tep m]=\frac1{Re}\overline{\tep m{}_{,yyy}(1)} .
\end{equation}

\begin{proposition}[Wall coefficients of the high modes]\label{prop:bmc}
For $m\ge n_1$,
\begin{equation}\label{eq:bm}
b_m=-\frac{(-1)^m\sqrt2\,\sigma_m}{Re}+O(1) ,
\end{equation}
\begin{equation}\label{eq:betam}
\beta_m=(-1)^m\,\frac{\sqrt2\,\sigma_m}{Re}\cdot\frac{2\pi k}{\sinh2\pi k}\Big(\cosh2\pi k-(-1)^m\Big)+O(1),
\end{equation}
so there are $m_0\ge n_1$ and $0<c_1<c_2$, depending on $(k,Re)$, with
$c_1m/Re\le|b_m|,|\beta_m|\le c_2m/Re$ and $b_m,\beta_m\ne0$ for $m\ge m_0$. The two
wall velocity components thus carry coefficients of the same order $m/Re$, and the
first-stage coefficient is $\beta^c[\tep m]=\beta_m-c\,\alpha[\tep m]=\beta_m+c\,b_m$, so
$|\beta^c[\tep m]|\le C(1+|c|)m/Re$.
\end{proposition}

\begin{proof}[Proof of Proposition \ref{prop:bmc}, \eqref{eq:bm}]
Fix $W(y)\coloneqq-y^2(1-y)$: $W(0)=W_y(0)=W(1)=0$, $W_y(1)=1$. Apply \eqref{eq:trace}
with $\zeta=\tep m$ and with $\zeta=e_m$, using
$\mathcal L_k^*\tep m=\overline{\lp m}\Lap\tep m$, $\Lap^2e_m=Re\,\lambda_m\Lap e_m$, and
$\ipL{W}{\Lap\zeta}=\ipL{\Lap W}{\zeta}$:
\begin{align}
\frac1{Re}\overline{\tep m{}_{,yy}(1)}&=\ipL{\mathcal L_kW}{\tep m}-\lp m\ipL{\Lap W}{\tep m},\label{eq:twotraces}\\
\frac1{Re}\overline{e_{m,yy}(1)}&=\tfrac1{Re}\ipL{\Lap^2W}{e_m}-\lambda_m\ipL{\Lap W}{e_m}.\label{eq:twotracesb}
\end{align}
Subtract, with $\delta_m\coloneqq\tep m-e_m$, $\|\delta_m\|_k=O(d_m^{-1})$ by
\eqref{eq:kato}:
\begin{multline}\label{eq:diffbm}
\frac1{Re}\overline{\big(\tep m{}_{,yy}-e_{m,yy}\big)(1)}
=\ipL{\mathcal L_kW-\tfrac1{Re}\Lap^2W}{\tep m}+\tfrac1{Re}\ipL{\Lap^2W}{\delta_m}\\
-(\lp m-\lambda_m)\ipL{\Lap W}{e_m}-\lp m\ipL{\Lap W}{\delta_m}.
\end{multline}
The first term is $O(1)$ ($\mathcal L_kW-\tfrac1{Re}\Lap^2W$ is a fixed $L^2$ function
and $\|\tep m\|\le C$), the second is $O(m^{-1})$ (as $\|\delta_m\|\le C\|\delta_m\|_k$),
the third is $O(1)\cdot O(m^{-1})$ by \eqref{eq:kato} and Proposition \ref{prop:stokes-all}(iii).
For the fourth, from $A_k^*\tep m=\overline{\lp m}\tep m$ and $S_ke_m=\lambda_me_m$,
\[
(S_k-\lambda_m)\delta_m=r_m,\qquad
r_m\coloneqq(\overline{\lp m}-\lambda_m)\tep m-B_k^*\tep m,\qquad \|r_m\|_k\le C .
\]
Taking the $\Hk$ inner product with $e_j$, $j\ne m$, and using self-adjointness,
$\ip{\delta_m}{e_j}=\ip{r_m}{e_j}/(\lambda_j-\lambda_m)$. Hence, with
$\gamma_m\coloneqq\ip{\delta_m}{e_m}=O(d_m^{-1})$,
\[
\ipL{\Lap W}{\delta_m}=\overline{\gamma_m}\,\ipL{\Lap W}{e_m}+\sum_{j\ne m}\frac{\overline{\ip{r_m}{e_j}}}{\lambda_j-\lambda_m}\ipL{\Lap W}{e_j},
\]
and by Cauchy--Schwarz, $\ipL{\Lap W}{e_j}=O(j^{-1})$ and \eqref{eq:gapjm},
\[
\Big|\sum_{j\ne m}\cdots\Big|\le\|r_m\|_k\Big(\sum_{j\ne m}\frac{C\,Re^2}{j^2(j^2-m^2)^2}\Big)^{1/2}=O(m^{-2}),
\]
since with $j=m+i$ the terms are $O(m^{-4}i^{-2})$, those with $j\le m/2$ sum to
$O(m^{-4})$, and those with $j\ge2m$ to $O(m^{-5})$. Thus
$\lp m\ipL{\Lap W}{\delta_m}=O(m^2/Re)\big(O(d_m^{-1}m^{-1})+O(m^{-2})\big)=O(1)$. All
terms of \eqref{eq:diffbm} are $O(1)$, and \eqref{eq:efest} gives \eqref{eq:bm}.
\end{proof}

\begin{proof}[Proof of Proposition \ref{prop:bmc}, \eqref{eq:betam}]
Identical to the proof of \eqref{eq:bm} with the test function
$W(y)\coloneqq y^2(3-2y)$, which satisfies $W(0)=W_y(0)=0$, $W(1)=1$, $W_y(1)=0$, so that
\eqref{eq:trace} gives
$-\frac1{Re}\overline{\zeta_{yyy}(1)}=\ipL{(\mathcal L_k-\lambda\Lap)W}{\zeta}$ for the
respective eigenfunctions; every error estimate goes through verbatim since only
$W\in C^\infty$, $W(0)=W_y(0)=0$ entered them. The main term is
$\tfrac1{Re}e_{m,yyy}(1)$, evaluated in \eqref{eq:efest3}.
\end{proof}

\begin{lemma}[Trace injectivity]\label{lem:inj}
Let $z_*\in\sigma(A_k)$ and $\zeta\in\ker(A_k^*-\bar z_*)$, $\zeta\ne0$. Then
$(\alpha[\zeta],\beta[\zeta])\ne(0,0)$. Consequently the map
$\zeta\mapsto(\alpha[\zeta],\beta[\zeta])$ is injective on each eigenspace of $A_k^*$,
and every eigenvalue of $A_k^*$ (hence of $A_k$) has geometric multiplicity at most $2$.
The same holds for $A_k$ with the traces $\phi\mapsto(\phi_{yy}(1),\phi_{yyy}(1))$.
\end{lemma}

For the tangential trace alone the nonvanishing is Munteanu's lemma
\cite[Lemma 2.2]{Munteanu2011}, established there by the same route.

\subsection{Pre-feedback and pre-compensated operator}\label{sec:precomp}

\subsubsection{State space and mass operator}

\begin{condition}\label{cond:S}
$|\varepsilon_1|\le1$.
\end{condition}

Since $0<\ipL{\ell}{g_1}<1$, Condition \ref{cond:S} gives
$1-\varepsilon_1\ipL{\ell}{g_1}\ge1-\ipL{\ell}{g_1}>0$, a margin depending on $k$ and
$\theta_1$ only; it is in force from here on.

\begin{definition}\label{def:Hkp}
Let $H^1_{(0}\coloneqq\{V\in H^1(0,1):V(0)=0\}$ and
\begin{equation}\label{eq:Hkp}
\Hkp\coloneqq\big\{V\in H^1_{(0}:\ V(1)=\varepsilon_1\ipL{V}{g_1}\big\},
\end{equation}
equipped with $\ip{\cdot}{\cdot}$ from \eqref{eq:Hk}.
\end{definition}

\begin{lemma}[Boundary isomorphism]\label{lem:Xi}
Under Condition \ref{cond:S}, the map
\begin{equation}\label{eq:Xi}
\Xi:\Hk\to\Hkp,\qquad
\Xi V\coloneqq V+\frac{\varepsilon_1\ipL{V}{g_1}}{1-\varepsilon_1\ipL{\ell}{g_1}}\,\ell ,
\end{equation}
is an isomorphism with inverse $\Xi^{-1}W=W-W(1)\ell$, and
$\|\Xi V-V\|_k\le C|\varepsilon_1|\,\|V\|_k$ with $C=C(k,\theta_1)$.
\end{lemma}

\begin{definition}[Modified mass inverse]\label{def:mass}
For $F\in H^{-1}$ let $\Lapm^{-1}F$ be the unique $U\in\Hkp$ with
$\Lap U=F$ in $H^{-1}$; explicitly
\begin{equation}\label{eq:massinv}
\Lapm^{-1}F=\Lap^{-1}F+\frac{\varepsilon_1\ipL{\Lap^{-1}F}{g_1}}{1-\varepsilon_1\ipL{\ell}{g_1}}\,\ell
=\Xi\,\Lap^{-1}F .
\end{equation}
\end{definition}

\subsubsection{Operator, resolvent, and secular matrix}

\begin{definition}[Pre-compensated operator]\label{def:Ap}
On $\Hkp$ define $A_k'$ with domain
\begin{equation}\label{eq:domAp}
\mathcal D(A_k')\coloneqq\Big\{V\in H^3:\ V(0)=V_y(0)=0,\ \
V(1)=\varepsilon_1\ipL{V}{g_1},\ \
V_y(1)=c\,\varepsilon_1\ipL{V}{g_1}+\varepsilon_2\ipL{V}{g_2}\Big\}
\end{equation}
by $A_k'V\coloneqq\Lapm^{-1}\mathcal L_kV$, i.e. $A_k'V$ is the unique
element of $\Hkp$ with $\Lap(A_k'V)=\mathcal L_kV$ in $H^{-1}$.
\end{definition}

\begin{remark}
The definition encodes the pre-feedback consistently with the mass operator of
\eqref{eq:OSk}: along a trajectory of \eqref{eq:OSk} under the static laws
\eqref{eq:prefeed} with $F_k=0$, the time derivative of the constraint
$V(1)=\varepsilon_1\ipL{V}{g_1}$ shows that $V_t$ again lies in $\Hkp$, and
$\Lap V_t=\mathcal L_kV$ then determines $V_t=A_k'V$.
\end{remark}

The eigenvalues of $A_k'$ are located through the determinant of a $2\times2$ matrix
built from the two feedback functionals: the secular determinant, so called after the
classical usage in perturbation theory, where the zeros of such a determinant are the
perturbed eigenvalues.

\begin{definition}[Secular matrix]\label{def:secular}
For $z\in\rho(A_k)$, with the responses of Definition \ref{def:Psi}, set
\begin{equation}\label{eq:secular}
\mathfrak D(z)\coloneqq
\begin{pmatrix}
1-\varepsilon_1\ipL{\rho_c(z)}{g_1} & -\varepsilon_1\ipL{\Psi(z)}{g_1}\\[2pt]
-\varepsilon_2\ipL{\rho_c(z)}{g_2} & 1-\varepsilon_2\ipL{\Psi(z)}{g_2}
\end{pmatrix}.
\end{equation}
\end{definition}

\begin{proposition}[Resolvent of $A_k'$]\label{prop:resAp}
Let $z\in\rho(A_k)$ with $\det\mathfrak D(z)\ne0$. Then $z\in\rho(A_k')$ and, for
$F\in\Hkp$,
\begin{equation}\label{eq:resAp}
(z-A_k')^{-1}F=W_{\mathrm p}+p\,\rho_c(z)+\tilde q\,\Psi(z),\qquad
W_{\mathrm p}\coloneqq(z-A_k)^{-1}\big(F-F(1)\ell\big),
\end{equation}
\begin{equation}\label{eq:pq}
\begin{pmatrix}p\\ \tilde q\end{pmatrix}
=\mathfrak D(z)^{-1}
\begin{pmatrix}\varepsilon_1\ipL{W_{\mathrm p}}{g_1}\\ \varepsilon_2\ipL{W_{\mathrm p}}{g_2}\end{pmatrix}.
\end{equation}
Conversely, if $z\in\rho(A_k)$ and $\det\mathfrak D(z)=0$ then $z\in\sigma(A_k')$, with
eigenvector $p\,\rho_c(z)+\tilde q\,\Psi(z)$ for any nonzero
$(p,\tilde q)\in\ker\mathfrak D(z)$. If, in addition, $\det\mathfrak D(z_0)\ne0$ at some $z_0\in\rho(A_k)$, then $A_k'$ is
closed with compact resolvent and discrete spectrum.
\end{proposition}

\begin{proof}
For $W$ of the form \eqref{eq:resAp}, $W\in H^3$, $W(0)=W_y(0)=0$, and the traces are
$W(1)=p$, $W_y(1)=cp+\tilde q$; thus $W\in\mathcal D(A_k')$ iff
$p=\varepsilon_1\ipL{W}{g_1}$ and $\tilde q=\varepsilon_2\ipL{W}{g_2}$, which is exactly
the linear system $\mathfrak D(z)(p,\tilde q)^{\mathsf T}
=(\varepsilon_1\ipL{W_{\mathrm p}}{g_1},\varepsilon_2\ipL{W_{\mathrm p}}{g_2})^{\mathsf T}$.
For such $W$: $\Lap(zW-A_k'W)=z\Lap W-\mathcal L_kW=z\Lap W_{\mathrm p}-\mathcal L_kW_{\mathrm p}
=\Lap\big((z-A_k)W_{\mathrm p}\big)=\Lap\big(F-F(1)\ell\big)=\Lap F$, using
$\mathcal L_k(\rho_c,\Psi)=z\Lap(\rho_c,\Psi)$ and $\Lap\ell=0$; since
$zW-A_k'W\in\Hkp$ and $F\in\Hkp$, uniqueness in Definition \ref{def:mass} gives
$(z-A_k')W=F$. Injectivity of $z-A_k'$ under $\det\mathfrak D\ne0$: if
$(z-A_k')W=0$ with $W\in\mathcal D(A_k')$, decompose
$W=W_0+W(1)\rho_c(z)+\big(W_y(1)-cW(1)\big)\Psi(z)$ with $W_0\in H^3\cap H^2_0$; then
$\Lap\big((z-A_k)W_0\big)=\Lap(zW-A_k'W)=0$ tested against $H^1_0$, so
$(z-A_k)W_0=0$ and $W_0=0$; the domain constraints then read
$\mathfrak D(z)(W(1),W_y(1)-cW(1))^{\mathsf T}=0$, forcing $W=0$. The same computation
with $\det\mathfrak D(z)=0$ produces the stated eigenvector (it is nonzero because
$\rho_c,\Psi$ have independent traces at $y=1$). For the last statement: at such
$z_0$, \eqref{eq:resAp} exhibits $(z_0-A_k')^{-1}$ as a bounded operator on $\Hkp$
whose range lies in $H^3$ (each component does, with norm controlled through Lemma
\ref{lem:ell}); hence the resolvent is compact into $\Hkp$, $A_k'$ is closed with
compact resolvent, and its spectrum is discrete.
\end{proof}

\begin{proposition}[Factorization of the secular determinant]\label{prop:factor}
Let $D_1(z)\coloneqq1-\varepsilon_1\ipL{\rho_c(z)}{g_1}$ and, on the open set where
$D_1\ne0$,
\begin{equation}\label{eq:Psi1}
\Psi^{(1)}(z)\coloneqq\Psi(z)+\frac{\varepsilon_1\ipL{\Psi(z)}{g_1}}{D_1(z)}\,\rho_c(z),\qquad
D_2(z)\coloneqq1-\varepsilon_2\ipL{\Psi^{(1)}(z)}{g_2}.
\end{equation}
Then
\begin{equation}\label{eq:factor}
\det\mathfrak D(z)=D_1(z)\,D_2(z).
\end{equation}
Moreover $\Psi^{(1)}$ is the tangential frequency response of the intermediate operator
$A_k^{(1)}$ (the operator of Definition \ref{def:Ap} with $\varepsilon_2=0$): it solves
$\mathcal L_k\Psi^{(1)}=z\Lap\Psi^{(1)}$ with
\begin{equation}\label{eq:Psi1bc}
\Psi^{(1)}(1)=\varepsilon_1\ipL{\Psi^{(1)}}{g_1},\qquad
\Psi^{(1)}_y(1)-c\,\varepsilon_1\ipL{\Psi^{(1)}}{g_1}=1 ,
\end{equation}
and $D_2$ is the scalar secular function of the second (tangential) closure relative to
$A_k^{(1)}$.
\end{proposition}

\begin{proof}
Expanding the determinant of \eqref{eq:secular}:
$\det\mathfrak D=D_1\big(1-\varepsilon_2\ipL{\Psi}{g_2}\big)-\varepsilon_1\varepsilon_2\ipL{\Psi}{g_1}\ipL{\rho_c}{g_2}
=D_1-\varepsilon_2\big[D_1\ipL{\Psi}{g_2}+\varepsilon_1\ipL{\Psi}{g_1}\ipL{\rho_c}{g_2}\big]
=D_1\big[1-\varepsilon_2\ipL{\Psi^{(1)}}{g_2}\big]$. For \eqref{eq:Psi1bc}:
$\Psi^{(1)}(1)=\varepsilon_1\ipL{\Psi}{g_1}/D_1$ since $\Psi(1)=0$, $\rho_c(1)=1$, while
$\varepsilon_1\ipL{\Psi^{(1)}}{g_1}=\varepsilon_1\ipL{\Psi}{g_1}\big[1+\varepsilon_1\ipL{\rho_c}{g_1}/D_1\big]
=\varepsilon_1\ipL{\Psi}{g_1}/D_1$; and
$\Psi^{(1)}_y(1)=1+c\,\varepsilon_1\ipL{\Psi}{g_1}/D_1$.
\end{proof}

\begin{lemma}[Multiplicity counting]\label{lem:counting}
Let $\Gamma\subset\rho(A_k)$ be a simple closed contour with
$\det\mathfrak D\ne0$ on $\Gamma$. Then the total algebraic multiplicities of the
spectra enclosed by $\Gamma$ satisfy
\begin{equation}\label{eq:counting}
\nu(\Gamma;A_k')=\nu(\Gamma;A_k)+\frac1{2\pi i}\oint_\Gamma\frac{(\det\mathfrak D)'(z)\,\dd z}{\det\mathfrak D(z)} .
\end{equation}
In particular, for an isolated $z_0$, $\nu(z_0;A_k')=\nu(z_0;A_k)+\ord_{z_0}\det\mathfrak D$,
where $\ord$ counts zeros minus pole order of the meromorphic function
$\det\mathfrak D$.
\end{lemma}

\subsection{Simplicity of the low eigenvalues}\label{sec:split}

Call \emph{low} the finitely many distinct eigenvalues $z_{*1},\dots,z_{*L}$ of $A_k$
other than $\{\lp m\}_{m\ge m_0}$. By Lemma \ref{lem:inj} the adjoint $A_k^*$ has at
most two Jordan chains \eqref{eq:jordan} at each such $\bar z_*$; call $z_*$ a one- or two-chain eigenvalue accordingly.
At a low eigenvalue $z_*$ fix chain bases as in \eqref{eq:biorth}, with chain lengths
$\kappa_1\ge\kappa_2\ge1$ (the second present only in the two-chain case,
$\nu_*=\kappa_1+\kappa_2$), $\zeta^{(r)}_0$ the eigenvector heading the $r$-th adjoint
chain, and $\phi^{(r)}_0$ the corresponding eigenvector of $A_k$. Write $\alpha^{(r)}_p=\alpha[\zeta^{(r)}_p]$,
$\beta^{c,(r)}_p=\beta^c[\zeta^{(r)}_p]$, and
\begin{equation}\label{eq:Rc}
R^c(z_*)\coloneqq\sum_{r:\ \kappa_r=\kappa_1}\beta^c[\zeta^{(r)}_0]\,\ipL{\phi^{(r)}_0}{g_1},
\end{equation}
the leading Laurent coefficient of $\ipL{\rho_c(z)}{g_1}$ at $z_*$ by
\eqref{eq:topcoeff}.

\begin{condition}[Ratio $c$]\label{cond:c}
$c\notin\mathcal C$, where $\mathcal C$ is the finite set of real numbers for which
$\beta^c[\zeta^{(r)}_0]=\beta[\zeta^{(r)}_0]-c\,\alpha[\zeta^{(r)}_0]=0$ for some low
$z_*$ and some $r$.
\end{condition}

Each chain head excludes at most one value of $c$ (if $\alpha[\zeta^{(r)}_0]=0$ then
$\beta[\zeta^{(r)}_0]\ne0$ by Lemma \ref{lem:inj} and no value is excluded), so
$\#\mathcal C\le2L$.

\begin{condition}[First gain shape]\label{cond:th1}
$\theta_1\notin\Theta_1$, where $\Theta_1$ is the union over low $z_*$ of the zero sets
of $\theta\mapsto R^c(z_*)=\ipL{v_{z_*}}{g_\theta}$,
$v_{z_*}\coloneqq\sum_{\kappa_r=\kappa_1}\beta^c[\zeta^{(r)}_0]\phi^{(r)}_0\ne0$, together
with the zero set of $\theta\mapsto\ipL{\Psi(z_0)}{g_\theta}$ for one fixed reference
$z_0\in\rho(A_k)$.
\end{condition}

\begin{lemma}[Discreteness of the exclusion sets]\label{lem:disc}
For $v\in L^2(0,1)$, $v\ne0$, the function
$\theta\mapsto\ipL{v}{g_\theta}=\int_0^1v(y)e^{-\theta(1-y)}\dd y$ is entire and not
identically zero; its zero set is discrete. Moreover
$\ipL{\Psi(z_0)}{g_\theta}=-\theta^{-2}\big(1+O(\theta^{-1})\big)$ as
$\theta\to+\infty$, so the last member of $\Theta_1$ is discrete as well. Hence
$\Theta_1$, and likewise $\Theta_2$ below, are discrete subsets of $(0,\infty)$.
\end{lemma}

\begin{proposition}[First stage: splitting by the $\varepsilon_1$ feedback]\label{lem:step1}
Assume Conditions \ref{cond:S}, \ref{cond:c}, \ref{cond:th1} and set $\varepsilon_2=0$.
Fix disjoint closed disks $\mathbb D_*$ of radius $r_*$ around the low $z_*$, contained
in the region of Proposition \ref{prop:kato} carried by $P_{\mathrm{low}}$ and the disks
$\{|z-\lambda_m|\le d_m/4\}$, $n_1\le m<m_0$. There is $\bar\varepsilon_1>0$ such that
for $0<|\varepsilon_1|\le\bar\varepsilon_1$, at each low $z_*$:
\begin{enumerate}
\item[(i)] $D_1$ has a pole of order exactly $\kappa_1$ at $z_*$, and exactly
$\kappa_1$ zeros in $\mathbb D_*$, all simple, at the $\kappa_1$ roots of
\begin{equation}\label{eq:fan}
(\mu'-z_*)^{\kappa_1}=\varepsilon_1R^c(z_*)\big(1+O(|\mu'-z_*|)\big),
\end{equation}
which are distinct for small $\varepsilon_1$ and spaced by
$|\varepsilon_1R^c|^{1/\kappa_1}$ around $z_*$; each is an algebraically simple eigenvalue of
$A_k^{(1)}$ with eigenvector $\rho_c(\mu')$, and
$D_1'(\mu')=\kappa_1(\mu'-z_*)^{-1}\big(1+O(|\mu'-z_*|)\big)\ne0$.
\item[(ii)] In the single-chain case, $\sigma(A_k^{(1)})\cap\mathbb D_*$ consists
exactly of these $\kappa_1$ roots. In the two-chain case it consists of those
roots together with $z_*$ itself, which is an eigenvalue of $A_k^{(1)}$ of algebraic
multiplicity $\kappa_2$ and geometric multiplicity one: a single Jordan chain.
\item[(iii)] The tangential coefficient of each such root, i.e. the residue coordinate
of the response $\Psi^{(1)}$ at $\mu'$, is
$b^{(1)}(\mu')=\varepsilon_1\ipL{\Psi(\mu')}{g_1}/D_1'(\mu')$, and
$b^{(1)}(\mu')\ne0$ for $0<|\varepsilon_1|\le\bar\varepsilon_1$.
\end{enumerate}
\end{proposition}

\begin{proof}
(i) Write $u=z-z_*$, $P(z)\coloneqq\ipL{\rho_c(z)}{g_1}=R^cu^{-\kappa_1}(1+O(u))$ by
\eqref{eq:topcoeff} and Condition \ref{cond:th1} ($R^c\ne0$; the pole order cannot
exceed $\kappa_1$ by Lemma \ref{lem:resform}). The zeros of $D_1=1-\varepsilon_1P$ in
$\mathbb D_*$ are the solutions of $u^{\kappa_1}=\varepsilon_1R^c(1+\varphi(u))$,
$\varphi(0)=0$ analytic: on $|u|=(2|\varepsilon_1R^c|)^{1/\kappa_1}$ Rouch\'e against
$u^{\kappa_1}-\varepsilon_1R^c$ gives $\kappa_1$ zeros, which for small $\varepsilon_1$
are the $\kappa_1$ analytic-in-$\varepsilon_1^{1/\kappa_1}$ branches
$u_j=\omega^j(\varepsilon_1R^c)^{1/\kappa_1}(1+O(|\varepsilon_1|^{1/\kappa_1}))$,
$\omega=e^{2\pi i/\kappa_1}$, pairwise distinct; on the rest of $\mathbb D_*$,
$|\varepsilon_1P|<1$, so there are no other zeros. Differentiating
$P=u^{-\kappa_1}/\varepsilon_1\cdot(1+O(u))$ at a zero gives the stated $D_1'$. Each
zero $\mu'\in\rho(A_k)$ is an eigenvalue of $A_k^{(1)}$ with eigenvector $\rho_c(\mu')$
by Proposition \ref{prop:resAp} (with $\varepsilon_2=0$ the secular matrix is
triangular, with kernel $(1,0)^{\mathsf T}$); its algebraic multiplicity is
$\ord_{\mu'}D_1=1$ by Lemma \ref{lem:counting}.

(ii) By Lemma \ref{lem:counting} on $\partial\mathbb D_*$ and on a small circle around
$z_*$: $\nu(\mathbb D_*;A^{(1)}_k)=\nu_*+(\kappa_1-\kappa_1)=\nu_*$ and
$\nu(z_*;A^{(1)}_k)=\nu_*-\kappa_1$, which is $0$ (single chain) or $\kappa_2$ (two
chains). Geometric multiplicity in the two-chain case: an eigenfunction of $A^{(1)}_k$
at $z_*$ is a solution of $\mathcal L_kV=z_*\Lap V$, $V(0)=V_y(0)=0$ --- a
two-dimensional space, which in the two-chain case is exactly
$\spn\{\phi^{(1)}_0,\phi^{(2)}_0\}\subset H^2_0$ --- subject to
$V(1)=\varepsilon_1\ipL{V}{g_1}$ and $V_y(1)=cV(1)$. Since $V(1)=0$ on this span, a
two-dimensional kernel would force $\ipL{\phi^{(r)}_0}{g_1}=0$ for $r=1,2$,
contradicting $R^c\ne0$ (in the case $\kappa_1>\kappa_2$,
$R^c=\beta^c[\zeta^{(1)}_0]\ipL{\phi^{(1)}_0}{g_1}$, so $\ipL{\phi^{(1)}_0}{g_1}\ne0$
and the same argument applies with $\phi^{(1)}_0$). Hence one chain of length
$\kappa_2$.

(iii) $\Psi^{(1)}=\Psi+[\varepsilon_1\ipL{\Psi}{g_1}/D_1]\rho_c$ has a simple pole at
the simple zero $\mu'$ of $D_1$ with residue
$\varepsilon_1\ipL{\Psi(\mu')}{g_1}\rho_c(\mu')/D_1'(\mu')$, and the residue coordinate
along the eigenvector $\rho_c(\mu')$ is the stated $b^{(1)}$. Nonvanishing: the branch
$\mu'(\varepsilon_1)$ is a nonconstant Puiseux series; the meromorphic function
$z\mapsto\ipL{\Psi(z)}{g_1}$ is not identically zero (Condition \ref{cond:th1} at
$z_0$), so its zeros are isolated and the composition
$\varepsilon_1\mapsto\ipL{\Psi(\mu'(\varepsilon_1))}{g_1}$ is a not-identically-zero
Puiseux series; hence it is nonzero for all sufficiently small $\varepsilon_1\ne0$, and
so is $b^{(1)}$.
\end{proof}

At a two-chain eigenvalue the first stage splits one chain and leaves the other, which
we call retained. The next lemma identifies the principal part of the second-step
secular function there, \emph{exactly} and independently
of $\varepsilon_1$.

\begin{lemma}[Leading coefficient at the retained chain]\label{lem:t2}
Let $z_*$ be a two-chain low eigenvalue, Conditions
\ref{cond:S}--\ref{cond:th1} in force, $0<|\varepsilon_1|\le\bar\varepsilon_1$. Then the
principal part of $\ipL{\Psi^{(1)}(z)}{g_2}$ at $z_*$ does not depend on
$\varepsilon_1$: it equals the principal part of the meromorphic function
\begin{equation}\label{eq:Nfun}
\frac{N(z)}{\ipL{\rho_c(z)}{g_1}},\qquad
N\coloneqq\ipL{\rho_c}{g_1}\ipL{\Psi}{g_2}-\ipL{\rho_c}{g_2}\ipL{\Psi}{g_1},
\end{equation}
its pole order is at most $\kappa_2$, and its leading coefficient is
\begin{equation}\label{eq:t2}
t_2(z_*)\coloneqq\lim_{z\to z_*}(z-z_*)^{\kappa_2}\ipL{\Psi^{(1)}(z)}{g_2}
=\frac{\Big(\alpha[\zeta^{(1)}_0]\beta[\zeta^{(2)}_0]-\alpha[\zeta^{(2)}_0]\beta[\zeta^{(1)}_0]\Big)\,
\det\big[\ipL{\phi^{(r)}_0}{g_i}\big]_{r,i=1,2}}{R^c(z_*)} .
\end{equation}
The first factor in the numerator is nonzero by Lemma \ref{lem:inj} (injectivity of the
trace pair on the two-dimensional adjoint eigenspace), and it is independent of $c$.
\end{lemma}

\begin{condition}[Second gain shape]\label{cond:th2}
$\theta_2\notin\Theta_2$, where $\Theta_2$ is the union over the two-chain low $z_*$ of
the zero sets of
$\theta\mapsto\ipL{w_{z_*}}{g_\theta}$,
$w_{z_*}\coloneqq\ipL{\phi^{(1)}_0}{g_1}\phi^{(2)}_0-\ipL{\phi^{(2)}_0}{g_1}\phi^{(1)}_0$.
\end{condition}

Under Condition \ref{cond:th1}, $w_{z_*}\ne0$ (the coefficients are not both zero, as
in Proposition \ref{lem:step1}(ii)), so $\Theta_2$ is discrete by Lemma \ref{lem:disc}; and
$t_2(z_*)\ne0$ for $\theta_2\notin\Theta_2$, since
$\det[\ipL{\phi^{(r)}_0}{g_i}]=\ipL{w_{z_*}}{g_{\theta_2}}$.

\begin{proposition}[Second stage: splitting the retained chains]\label{lem:step2}
Assume Conditions \ref{cond:S}--\ref{cond:th2} and $0<|\varepsilon_1|\le\bar\varepsilon_1$.
There is $\bar\varepsilon_2=\bar\varepsilon_2(\varepsilon_1)>0$ such that for
$0<|\varepsilon_2|\le\bar\varepsilon_2$ the spectrum of $A_k'$ in the union of the low
disks consists of algebraically simple, pairwise distinct eigenvalues; specifically, in
each $\mathbb D_*$:
\begin{enumerate}
\item[(i)] (Retained chains.) At each two-chain $z_*$, the retained chain of length
$\kappa_2$ splits into $\kappa_2$ simple eigenvalues of $A_k'$ at the roots of
$(\mu''-z_*)^{\kappa_2}=\varepsilon_2t_2(z_*)(1+o(1))$, each a simple zero of $D_2$,
with eigenvector $\Psi^{(1)}(\mu'')$ and residue coordinate $b'=1/D_2'(\mu'')\ne0$.
\item[(ii)] (Simple modes.) Each simple eigenvalue $\mu$ of $A_k^{(1)}$ in the low
region --- under Condition \ref{cond:th1} every low eigenvalue of $A_k$ is moved by
step 1, so these are exactly the roots \eqref{eq:fan} of Proposition \ref{lem:step1} --- either moves to the unique nearby simple zero
$\mu''=\mu+\varepsilon_2\varrho_\mu+O(\varepsilon_2^2)$ of $D_2$, where $\varrho_\mu$ is
the residue of $\ipL{\Psi^{(1)}}{g_2}$ at $\mu$, with $b'=1/D_2'(\mu'')\ne0$; or, if
$\varrho_\mu=0$, stays at $\mu$ with $b'=b^{(1)}(\mu)/D_2(\mu)\ne0$.
\item[(iii)] No other spectrum of $A_k'$ appears in the low region, and the total count
in each $\mathbb D_*$ is $\nu_*$.
\end{enumerate}
\end{proposition}

\begin{proof}
By Lemma \ref{lem:t2}, $D_2=1-\varepsilon_2\ipL{\Psi^{(1)}}{g_2}$ has at $z_*$ a pole of
order exactly $\kappa_2$ with leading coefficient $-\varepsilon_2t_2\ne0$, at each
visible simple eigenvalue $\mu$ of $A^{(1)}_k$ a simple pole of residue
$-\varepsilon_2\varrho_\mu$, and is $1+O(\varepsilon_2)$ on the rest of the low region
minus fixed neighborhoods of these points (here the $\varepsilon_1$-dependent bound
$\sup|\ipL{\Psi^{(1)}}{g_2}|$ on that set enters $\bar\varepsilon_2(\varepsilon_1)$).
All neighborhoods in this argument are chosen after $\varepsilon_1$ is fixed and may
depend on it: the roots \eqref{eq:fan} and the retained $z_*$ have a strictly positive
minimum separation, of order $|\varepsilon_1R^c|^{1/\kappa_1}$, and
$\bar\varepsilon_2(\varepsilon_1)$ is reduced so that the second-step root clusters, of radius
$O(|\varepsilon_2t_2|^{1/\kappa_2})$, and the displaced simple eigenvalues,
$O(\varepsilon_2\varrho_\mu)$, remain inside mutually disjoint such neighborhoods.
The zeros of $D_2$: near $z_*$, the Puiseux argument of Proposition \ref{lem:step1}(i) with
$(\kappa_1,R^c,\varepsilon_1)$ replaced by $(\kappa_2,t_2,\varepsilon_2)$ gives
$\kappa_2$ simple zeros at the stated roots, distinct, with
$D_2'(\mu'')=\kappa_2(\mu''-z_*)^{-1}(1+o(1))\ne0$; near a visible simple $\mu$, the
equation $1=\varepsilon_2[\varrho_\mu/(z-\mu)+O(1)]$ gives one simple zero at
$\mu+\varepsilon_2\varrho_\mu(1+O(\varepsilon_2))$; and there are no other zeros. Now
count with Lemma \ref{lem:counting} and $\det\mathfrak D=D_1D_2$: on a small circle
around $z_*$, $\nu(z_*;A_k')=\nu_*-\kappa_1-\kappa_2+0=0$; around each step-2 root,
$\nu=0+0+1=1$; around each step-1 root $\mu$ (a simple zero of $D_1$): if visible,
$D_2$ contributes pole $-1$ there and its nearby zero $+1$ in the slightly larger
circle, so the circle around $\mu$ itself gives $\nu(\mu;A_k')=0+1-1=0$ and the nearby
zero of $D_2$ carries $\nu=1$; if invisible, $\nu(\mu;A_k')=0+1+0=1$ and the eigenvalue
stays at $\mu$. The eigenvectors and residue coordinates: at a simple zero $\mu''$ of
$D_2$ with $D_1(\mu'')\ne0$, the kernel of $\mathfrak D(\mu'')$ is spanned by the vector
producing $\Psi^{(1)}(\mu'')$ (Propositions \ref{prop:resAp}, \ref{prop:factor}), and
the tangential frequency response of $A_k'$ is
\begin{equation}\label{eq:Psi2}
\Psi^{(2)}(z)=\frac{\Psi^{(1)}(z)}{D_2(z)}
\end{equation}
(verified as in Proposition \ref{prop:factor}: $\Psi^{(2)}$ solves
$\mathcal L_k\Psi^{(2)}=z\Lap\Psi^{(2)}$ with
$\Psi^{(2)}(1)=\varepsilon_1\ipL{\Psi^{(2)}}{g_1}$ and
$\Psi^{(2)}_y(1)-c\varepsilon_1\ipL{\Psi^{(2)}}{g_1}-\varepsilon_2\ipL{\Psi^{(2)}}{g_2}=1$),
whose residue at $\mu''$ is $\Psi^{(1)}(\mu'')/D_2'(\mu'')$, giving $b'=1/D_2'(\mu'')$
in the normalization $\chi'=\Psi^{(1)}(\mu'')$; at a retained invisible $\mu$,
$\Psi^{(2)}=\Psi^{(1)}/D_2$ has residue $b^{(1)}(\mu)\rho_c(\mu)/D_2(\mu)$, giving
$b'=b^{(1)}(\mu)/D_2(\mu)\ne0$, the nonvanishing of $b^{(1)}$ at every such root being
Proposition \ref{lem:step1}(iii). Distinctness of all
eigenvalues across the low region holds for $\varepsilon_2$ small because their limits
as $\varepsilon_2\to0$ (the roots at $z_*$ and the step-1 configuration) are distinct,
except for the pairs (step-1 root, its step-2 perturbation), which coincide in the
limit but are the same eigenvalue.
\end{proof}

\subsection{Localization of the high eigenvalues}\label{sec:high}

\begin{lemma}[High-mode estimates]\label{lem:tails}
Enlarge $m_0$ if necessary. There is $C=C(k,Re,c,\theta_1,\theta_2)$ such that for
$m\ge m_0$ and $z$ on the circle $\Gamma_m=\{|z-\lp m|=d_m/4\}$, likewise (with
$\tau_m$ replaced by $\tau_j$) on the origin-centered circles
$C_j\coloneqq\{|z|=r_j\}$, $r_j\coloneqq\tfrac12(|\lambda_j|+|\lambda_{j+1}|)$,
$j\ge m_0$, and also for real $z\ge\Lambda_0$ large:
\begin{enumerate}
\item[(i)] $\big|\ipL{\Psi(z)}{g_i}\big|\le C\,\tau_m$, with
$\tau_m\coloneqq(1+\log m)/m$ (respectively $\le C/\sqrt z$ for real $z$ large);
\item[(ii)] $\big|\ipL{\rho_c(z)}{g_i}-\ipL{\ell}{g_i}\big|\le C\,\tau_m$
(respectively $\to0$);
\item[(iii)] $\big|\ipL{(z-A_k)^{-1}G}{g_i}\big|\le C\,m^{-2}\|G\|_k$ for all
$G\in\Hk$;
\item[(iv)] $\|\Psi(z)\|_k\le C$ and $\|\rho_c(z)-\ell\|_k\le C$.
\end{enumerate}
\end{lemma}

\begin{proposition}[High modes of $A_k'$]\label{prop:high}
Assume Conditions \ref{cond:S}--\ref{cond:th2}. There are $\bar\varepsilon_1,\bar\varepsilon_2>0$,
$\bar\varepsilon_1$ as in Lemma \ref{lem:step1} and depending on $c,\theta_1$ only,
$\bar\varepsilon_2$ depending in addition on $\theta_2$ and on $\varepsilon_1$, such that for
$0<|\varepsilon_1|\le\bar\varepsilon_1$, $0<|\varepsilon_2|\le\bar\varepsilon_2$ and
every $m\ge m_0$:
\begin{enumerate}
\item[(i)] $A_k'$ has exactly one eigenvalue $\mu'_m$ in $\{|z-\lp m|<d_m/4\}$, it is
simple, and $|\mu'_m-\lp m|\le C(|\varepsilon_1|+|\varepsilon_2|)/Re$;
\item[(ii)] the spectral projections satisfy
$\|\hat P_m-P_m\|\le C(|\varepsilon_1|+|\varepsilon_2|)/m$, where $\hat P_m$ is the
Riesz projection of $\hat A=\Xi^{-1}A_k'\Xi$; consequently the eigenvector
$\chi'_m\in\Hkp$ of $A_k'$, suitably normalized, satisfies
$\|\chi'_m-\ep m\|_k\le C/m$;
\item[(iii)] the residue coordinate $b'_m$ of $\Psi^{(2)}$ at $\mu'_m$ (in the
normalization of (ii)) satisfies $|b'_m-b_m|\le\tfrac12|b_m|$; in particular
$b'_m\ne0$ and $c_1m/(2Re)\le|b'_m|\le2c_2m/Re$;
\item[(iv)] $\sigma(A_k')$ consists exactly of the low
eigenvalues of Lemmas \ref{lem:step1}--\ref{lem:step2}, of total algebraic multiplicity
$m_0-1$, and the points $\mu'_m$, $m\ge m_0$; in particular
$\{\Rea z>\|B_k\|+1\}\subset\rho(A_k')$.
\end{enumerate}
\end{proposition}

\begin{proof}
Every bound $C(|\varepsilon_1|+|\varepsilon_2|)$ below abbreviates
$C_1|\varepsilon_1|+C_2|\varepsilon_2|$, with $C_1$ carrying the $g_1$ pairings and
depending on $k,Re,c,\theta_1$ alone and $C_2$ the $g_2$ pairings; each smallness
requirement is imposed on the two terms separately, so $\bar\varepsilon_1$ is fixed
before $\theta_2$ is chosen.

(i) Let $s_1\coloneqq1-\varepsilon_1\ipL{\ell}{g_1}\in[\tfrac12,\tfrac32]$ for
$\varepsilon_1$ small. By Lemma \ref{lem:tails}, on $\Gamma_m$,
$|D_1-s_1|\le C|\varepsilon_1|\tau_m$ and, using
$\Psi^{(1)}=\Psi+[\varepsilon_1\ipL{\Psi}{g_1}/D_1]\rho_c$,
$|D_2-1|\le C|\varepsilon_2|\tau_m(1+|\varepsilon_1|)$; hence
$|\det\mathfrak D-s_1|\le C(|\varepsilon_1|+|\varepsilon_2|)\tau_m$ on $\Gamma_m$.
Inside the disk, $\det\mathfrak D$ is meromorphic with the single possible pole
$\lp m$; the residue matrix of $\mathfrak D$ at $\lp m$ is rank one (it is
$-(\varepsilon_1\ipL{\ep m}{g_1},\varepsilon_2\ipL{\ep m}{g_2})^{\mathsf T}\otimes(\beta^c[\tep m],b_m)$,
from the residues $\beta^c[\tep m]\ep m$ of $\rho_c$ and $b_m\ep m$ of $\Psi$), so
$\det\mathfrak D$ has at most a simple pole, of residue $\varrho_m$ with
$|\varrho_m|\le C(|\varepsilon_1|+|\varepsilon_2|)/Re$. Apply Rouch\'e to
$f(z)=(z-\lp m)\det\mathfrak D(z)$, analytic in the disk, against $s_1(z-\lp m)$:
\[
|f-s_1(z-\lp m)|\le\tfrac{d_m}4\,C(|\varepsilon_1|+|\varepsilon_2|)\tau_m+|\varrho_m|(1+C\tau_m)
<s_1\tfrac{d_m}4\quad\text{on }\Gamma_m
\]
for $|\varepsilon_1|+|\varepsilon_2|$ small (uniformly in $m$, since
$|\varrho_m|/d_m\le C\varepsilon/m$). So $f$ has exactly one zero; by Lemma
\ref{lem:counting} the disk contains exactly one eigenvalue of $A_k'$, simple, at the
zero of $\det\mathfrak D$ if $\varrho_m\ne0$ (namely
$\mu'_m=\lp m-\varrho_m/s_1+O(\varepsilon\tau_m\varrho_m)$) or at $\lp m$ itself if
$\varrho_m=0$; in both cases $|\mu'_m-\lp m|\le C(|\varepsilon_1|+|\varepsilon_2|)/Re$.

(ii) By \eqref{eq:resdiff} and Lemma \ref{lem:tails}(iii),(iv), the resolvent difference
on $\Gamma_m$ has norm at most
$\|\mathfrak D^{-1}\|\,C(|\varepsilon_1|+|\varepsilon_2|)m^{-2}\cdot C$, that is, at
most $C'\varepsilon\,Re\,m^{-2}$; so $\|\hat P_m-P_m\|\le\tfrac{d_m}4\cdot C'\varepsilon\,Re\,m^{-2}\le C\varepsilon/m<1$.
With $\hat\chi_m\coloneqq\hat P_m\ep m/\ip{\hat P_m\ep m}{\tep m}$ and
$\chi'_m\coloneqq\Xi\hat\chi_m$: $\|\hat\chi_m-\ep m\|_k\le C\varepsilon/m$ and
$\|\Xi\hat\chi_m-\hat\chi_m\|_k\le C|\varepsilon_1|\,|\ipL{\hat\chi_m}{g_1}|\le C\varepsilon/m$,
since $|\ipL{\ep m}{g_1}|\le C/m$.

(iii) Since $\tep m\in H^2_0$, Lemma \ref{lem:pairconv} gives
$L_m[v]\coloneqq-\ipL{v}{\Lap\tep m}=\ip{v}{\tep m}$ for \emph{every} $v\in H^1$ with
$v(0)=0$, boundary traces of $v$ notwithstanding. $L_m[\Psi^{(2)}(z)]$ is analytic on
$\rho(A_k')$, and by \eqref{eq:tailpair} and \eqref{eq:Psi2},
\[
L_m\big[\Psi^{(2)}(z)\big]=\frac{b_m+F(z)\beta^c[\tep m]}{(z-\lp m)\,D_2(z)},\qquad
F=\frac{\varepsilon_1\ipL{\Psi}{g_1}}{D_1} .
\]
Integrate over $\Gamma_m$: the left side gives $b'_mL_m[\chi'_m]$ (the only singularity of
$L_m[\Psi^{(2)}]$ inside $\Gamma_m$ is the eigenvalue $\mu'_m$: $\Psi^{(2)}$ is
analytic on $\rho(A_k')$, and by (i) the interior of $\Gamma_m$ meets $\sigma(A_k')$
only at $\mu'_m$). On $\Gamma_m$, $|F|\le C|\varepsilon_1|\tau_m$,
$|D_2-1|\le C|\varepsilon_2|\tau_m$, $|\beta^c[\tep m]|\le Cm/Re$ by Proposition
\ref{prop:bmc}, so the right side is
$b_m\big(1+O(\varepsilon\tau_m)\big)+O(\varepsilon\tau_m\,m/Re)$, i.e.
$b_m+O(\varepsilon(1+\log m)/Re)$. And
$L_m[\chi'_m]=\ip{\chi'_m}{\tep m}=1+O(\varepsilon/m)$, the normalization of (ii) up to
the $\Xi$-correction. Hence $|b'_m-b_m|\le C\varepsilon(1+\log m)/Re\le\tfrac12|b_m|$ for $m\ge m_0$ and
$\varepsilon$ small, by \eqref{eq:bm}.

(iv) Let $\Gamma_{\mathrm{big}}\coloneqq\{|z|=r_{\mathrm{big}}\}$ be a fixed circle
enclosing the low region inside the region of Proposition \ref{prop:kato}; on it the
entries of $\mathfrak D-I$ are $O(|\varepsilon_1|+|\varepsilon_2|)$ uniformly, so for
$\varepsilon$ small $\det\mathfrak D\ne0$ on $\Gamma_{\mathrm{big}}$ with winding
number zero, and Lemma \ref{lem:counting} gives
$\nu(\Gamma_{\mathrm{big}};A_k')=\nu(\Gamma_{\mathrm{big}};A_k)=m_0-1$: the low
eigenvalues of Lemmas \ref{lem:step1}--\ref{lem:step2} carry the entire low
multiplicity. Next, on the circle $C_j$ of Lemma \ref{lem:tails},
$|\det\mathfrak D-s_1|\le C(|\varepsilon_1|+|\varepsilon_2|)\tau_j<s_1$, so
$\det\mathfrak D$ is zero- and pole-free on $C_j$ with winding number zero, and Lemma
\ref{lem:counting} gives
$\nu(C_j;A_k')=\nu(C_j;A_k)$; and $\nu(C_j;A_k)=j$: for $m\le j$,
$|\lp m|\le|\lambda_m|+\|B_k\|<|\lambda_m|+d_j/2\le r_j$, while for $m\ge j+1$,
$|\lp m|\ge|\lambda_m|-\|B_k\|>r_j$, so $C_j$ encloses exactly the low block and the
high modes $m_0,\dots,j$, of total multiplicity $(m_0-1)+(j-m_0+1)=j$. The eigenvalues
already identified inside $C_j$ --- the low ones and $\mu'_m$ for $m_0\le m\le j$ ---
number exactly $j$ counted with multiplicity, so there is no other spectrum inside
$C_j$; since $r_j\to\infty$, every point of $\mathbb C$ lies inside some $C_j$, so
there is no other spectrum at all. The half-plane statement follows by shrinking $\bar\varepsilon_1,\bar\varepsilon_2$
once, the constraint on $\bar\varepsilon_1$ involving first-stage constants only: the low eigenvalues converge, as $\varepsilon_1,\varepsilon_2\to0$, to the low
plant eigenvalues (Lemmas \ref{lem:step1}--\ref{lem:step2}; the displacements are
Puiseux, $O(|\varepsilon_1|^{1/\kappa_1})+O(|\varepsilon_2|^{1/\kappa_2})$, so take
them of modulus $<1$), the plant spectrum has $\Rea z\le\|B_k\|$ by Proposition
\ref{prop:ops}(ii), and the high-mode displacements are $O(\varepsilon/Re)$ by (i).
\end{proof}

\subsection{Proof of Theorem \ref{thm:spectral}}\label{sec:spectral}

Theorem \ref{thm:spectral}, stated in Section \ref{sec:main1}, is assembled here. The
Fredholm construction of Section \ref{sec:part-fredholm} needs more than the theorem
states: how far the pre-feedback moves each plant eigenvalue, how close the
pre-compensated eigenvectors stay to the plant's, and how fast the tangential
coefficients grow. Those three quantities are recorded next.

\begin{proposition}[Pre-compensated spectrum: location and asymptotics]\label{prop:spectral-quant}
Under the hypotheses of Theorem \ref{thm:spectral}, with $\{\lp m\}_{m\ge m_0}$ the
plant eigenvalues of Proposition \ref{prop:kato}, $d_m$ the Stokes gaps of Proposition
\ref{prop:stokes-all}(ii), and $\ep m$ the high plant eigenvectors:
\begin{enumerate}
\item[(i)] the total multiplicity of $\sigma(A_k')$ in the low region is $m_0-1$, and
for $m\ge m_0$ the eigenvalue $\mu'_m$ is the unique one in $\{|z-\lp m|<d_m/4\}$,
with $|\mu'_m-\lp m|\le C(|\varepsilon_1|+|\varepsilon_2|)/Re$;
\item[(ii)] $\|\chi'_m-\ep m\|_k\le C/m$ for $m\ge m_0$;
\item[(iii)] $c_1m/(2Re)\le|b'_m|\le2c_2m/Re$ for $m\ge m_0$.
\end{enumerate}
\end{proposition}

\begin{proof}[Proof of Theorem \ref{thm:spectral}]
Combine Propositions \ref{lem:step1}, \ref{lem:step2} (low region) and Proposition
\ref{prop:high} (high modes and resolvent set), with the counting Lemma
\ref{lem:counting} on the circle $\Gamma_{\mathrm{big}}=\{|z|=r_{\mathrm{big}}\}$
enclosing the low region ($\det\mathfrak D\ne0$ there, with winding number zero since
$\det\mathfrak D\to s_1\ne0$ uniformly as $\varepsilon\to0$ on the fixed circle) for the
low count $m_0-1$. For the Riesz basis, compare the spectral subspaces of
$\hat A=\Xi^{-1}A_k'\Xi$ with the \emph{plant} decomposition, not the Stokes one. By
Corollary \ref{cor:decomp} and Lemma \ref{lem:bm}, there is a bounded invertible
$\Theta$ on $\Hk$ carrying the disjoint family
$\{Q^A_{\mathrm{low}}\}\cup\{P_m\}_{m\ge m_0}$,
$Q^A_{\mathrm{low}}\coloneqq P_{\mathrm{low}}+\sum_{n_1\le m<m_0}P_m$, to mutually
orthogonal projections summing to $I$ ($\Theta$ maps the Riesz basis of Corollary
\ref{cor:decomp} to the underlying orthonormal basis, so
$\Theta Q^A_{\mathrm{low}}\Theta^{-1}$ and $\Theta P_m\Theta^{-1}$ are orthogonal).
Apply Lemma \ref{lem:bm} to the conjugated family
$\{\Theta\hat P_{\mathrm{low}}\Theta^{-1}\}\cup\{\Theta\hat P_m\Theta^{-1}\}_{m\ge m_0}$
against these:
$\|\Theta(\hat P_{\mathrm{low}}-Q^A_{\mathrm{low}})\Theta^{-1}\|\le\kappa_\Theta\,C\varepsilon<1$
for $\varepsilon$ small ($\hat P_{\mathrm{low}}-Q^A_{\mathrm{low}}$ is the contour
integral over $\Gamma_{\mathrm{big}}$ of the $O(\varepsilon)$ resolvent difference
\eqref{eq:resdiff}), and
$\sum_{m\ge m_0}\|\Theta(\hat P_m-P_m)\Theta^{-1}\|^2\le\kappa_\Theta^2C^2\varepsilon^2\sum_mm^{-2}<\infty$
by Proposition \ref{prop:high}(ii), with $\kappa_\Theta\coloneqq\|\Theta\|\|\Theta^{-1}\|$.
The lemma yields a Riesz basis, which $\Theta^{-1}$ carries back. Within the finite-dimensional range of
$\hat P_{\mathrm{low}}$ pass to the eigenvector basis (all low eigenvalues simple), a
finite change of basis; then transport by the isomorphism $\Xi$ of Lemma
\ref{lem:Xi}. Item (iii) collects
Proposition \ref{lem:step1}(iii), Proposition \ref{lem:step2}(i)--(ii), and Proposition
\ref{prop:high}(iii). \end{proof}

\begin{proof}[Proof of Proposition \ref{prop:spectral-quant}]
Each item is the corresponding clause of the results just combined: (i) from the
counting above and Proposition \ref{prop:high}(i), (ii) from Proposition
\ref{prop:high}(ii) with $\|\ep m-\chi'_m\|_k\le C/m$, and (iii) from Proposition
\ref{prop:high}(iii).
\end{proof}

\begin{remark}
Theorem \ref{thm:spectral}(iii) is the Fattorini--Hautus condition for the pair
($A_k'$, tangential wall velocity): by the trace identity, $b'_n$ is (a nonzero multiple of)
the tangential trace functional evaluated on the adjoint eigenvector of $A_k'$ at
$\overline{\mu'_n}$. Items (i)--(iii) are the hypotheses of the abstract one-input
Fredholm theory \cite{HayatLoko26}, here obtained for every plant rather than assumed.
The Riesz basis of Theorem \ref{thm:riesz} is the step at which that theory is usually
obstructed; \cite{GHXZ25} reaches it for skew-adjoint operators by a compactness and
duality argument in place of the quadratically close criterion.
\end{remark}

\section{Two wall inputs suffice}\label{sec:necessity}

Theorem \ref{thm:2ctrl}, stated in Section \ref{sec:main1}, is proved here: the two
wall velocity components control the linearized flow at every wavenumber and Reynolds
number, with nothing assumed about the plant. With it come Propositions
\ref{prop:oneinput} and \ref{prop:ceiling}, each conditional on a spectral structure
not exhibited for this flow, the second of which is why the design of Section
\ref{sec:model} uses two stages. A defect of this kind
has forced a change of setup before: the water tank of \cite{CoronHXZ22} is
uncontrollable about its uniform steady states, and Fredholm backstepping is reached
there through a rank-one modification of the control operator, which appears in the
closed loop as a dynamic extension.

\begin{proof}[Proof of Theorem \ref{thm:2ctrl}]
First, solutions exist classically: lift the boundary data by
$\widehat V(t,y)\coloneqq V_c(t)\,y^2(3-2y)-\tau(t)\,y^2(1-y)$, which carries the four
required traces; the remainder solves the no-slip problem with a $C_c^\infty$-in-time
$\Hk$-valued forcing, given by the variation-of-constants formula for the analytic
semigroup of $A_k$, so $V$ is classical and $V(T)\in\Hk$. Let $z_0\in\Hk$ annihilate
the reachable set at time $T$. For $\delta\in(0,T)$, the controls supported in
$(0,T-\delta)$ reach at time $T$ exactly $e^{\delta A_k}$ applied to the
time-$(T-\delta)$ reachable set, so $z_\delta\coloneqq e^{\delta A_k^*}z_0$ annihilates
the reachable set at time $T-\delta$, and
$z_\delta\in\bigcap_m\mathcal D((A_k^*)^m)$ by analyticity of the semigroup. Let
$\zeta(t)\coloneqq e^{(T-\delta-t)A_k^*}z_\delta=e^{(T-t)A_k^*}z_0$ for $t\le T-\delta$;
then $\zeta(t)\in H^4\cap H^2_0$ for $t\le T-\delta$, and $s\mapsto e^{sA_k^*}z_0$ is
analytic on $(0,\infty)$. Along a classical solution,
using \eqref{eq:trace} and $\mathcal L_k^*\zeta=\Lap A_k^*\zeta$,
\[
\frac{\dd}{\dd t}\big[-\ipL{V(t)}{\Lap\zeta(t)}\big]
=-\ipL{\mathcal L_kV}{\zeta}+\ipL{V}{\Lap A_k^*\zeta}
=-\tau(t)\,\alpha[\zeta(t)]+V_c(t)\,\beta[\zeta(t)] ,
\]
where the mass pairing has no boundary terms because $\zeta(t)\in H^2_0$. Integrating
from $0$ to $T-\delta$ and using $V(0)=0$, $V(T-\delta)\in H^1_0$,
$\zeta(T-\delta)=z_\delta\in H^2$:
\[
\ip{V(T-\delta)}{z_\delta}=\int_0^{T-\delta}\big(V_c\,\beta[\zeta]-\tau\,\alpha[\zeta]\big)\dd t .
\]
Annihilation for all controls forces $\alpha[\zeta(t)]=\beta[\zeta(t)]=0$ on
$(0,T-\delta)$, i.e. the traces of $e^{sA_k^*}z_0$ vanish for $s\in(\delta,T)$; since
$\delta\in(0,T)$ was arbitrary and $s\mapsto e^{sA_k^*}z_0$ is analytic, they vanish
for all $s>0$. Decompose by Corollary \ref{cor:decomp} applied
to $A_k^*$ (same proof): $z_0=\sum_\lambda P^*_\lambda z_0$ over the distinct eigenvalues
$\bar\lambda$ of $A_k^*$, where for the high modes the projections are rank one. Then, with
$N_\lambda\coloneqq(A_k^*-\bar\lambda)$ nilpotent on $\operatorname{ran}P^*_\lambda$,
\[
\alpha[\zeta(T-s)]=\sum_\lambda e^{\lambda s}\sum_{j\ge0}\frac{s^j}{j!}\,\alpha\big[N_\lambda^jP^*_\lambda z_0\big],
\]
absolutely convergent for each $s>0$: the terms with $m\ge m_0$ are bounded by
$Ce^{s\Rea\bar\lambda^{\mathrm p}_m}\,m\,|c_m|/Re$ with $(c_m)\in\ell^2$ and
$\Rea\lp m\le-cm^2/Re$, and $\alpha[\tep m]=O(m/Re)$ by Proposition \ref{prop:bmc}.
Order the distinct real parts decreasingly. If some $P^*_\lambda z_0\ne0$, let $R_1$ be
the largest real part carrying a nonzero projection and $G_1$ the (finite) group of
eigenvalues with $\Rea\bar\lambda=R_1$. Then
$e^{-R_1s}\alpha[\zeta(T-s)]=\sum_{\lambda\in G_1}e^{i(\Ima\lambda)s}p_\lambda(s)+O(e^{-gs})$
for some $g>0$, with $p_\lambda$ polynomials; vanishing for all $s>0$ forces every
$p_\lambda\equiv0$: with $d$ the maximal degree in the group, multiply by
$s^{-d}e^{-i(\Ima\bar\lambda_0)s}$ and Ces\`aro-average in $s$ to isolate the
degree-$d$ coefficient at the frequency $\Ima\bar\lambda_0$; repeat over the finitely
many frequencies, then descend inductively in degree $d-1,\dots,0$. Hence
$\alpha[N^j_\lambda P^*_\lambda z_0]=0$ for all $\lambda,j$, and likewise for $\beta$.
Fix $\lambda$ with $w\coloneqq P^*_\lambda z_0\ne0$ and let $j_*$ be maximal with
$N^{j_*}_\lambda w\ne0$. Then $N^{j_*}_\lambda w\in\ker(A_k^*-\bar\lambda)$ with both
traces zero, so $N_\lambda^{j_*}w=0$ by Lemma \ref{lem:inj}: contradiction. Hence
$P^*_\lambda z_0=0$ for every $\lambda$, and $z_0=0$ by Corollary \ref{cor:decomp} for
$A_k^*$.
\end{proof}

\begin{proof}[Proof of Proposition \ref{prop:oneinput}]
$E\coloneqq\ker(z_*-A_k)$ is two-dimensional and lies in $H^4\cap H^2_0$: every
$\phi\in E$ has $\phi(1)=\phi_y(1)=0$. The restriction $F|_E$ is a linear functional on a
two-dimensional space, so it has a nonzero kernel element $\phi$. Then $\phi$ satisfies
all boundary conditions of the closed loop ($\phi_y(1)=0=F[\phi]$) and
$\mathcal L_k\phi=z_*\Lap\phi$, so $e^{z_*t}\phi$ solves it. The normal case is
identical with $G|_E$. For the last statement, $\ker(A_k^*-\bar z_*)$ is
two-dimensional as well, the geometric multiplicities of $A_k$ and $A_k^*$ at
conjugate eigenvalues agreeing for an isolated eigenvalue of finite algebraic
multiplicity; the tangential functional $\alpha$ of \eqref{eq:alphabeta} restricted to
it is a linear functional on a two-dimensional space, hence has a nonzero kernel
element $\zeta$. Thus the Fattorini--Hautus condition fails at $z_*$ for the tangential
component, so no controller acting through it --- of any state dimension --- moves
$z_*$; the normal component is identical with $\beta$.
\end{proof}

\begin{proposition}[Multiplicity retained by one rank-one stage]\label{prop:ceiling}
Let $z_*\in\sigma(A_k)$ carry at least two Jordan chains of length $\ge2$; write
$\kappa_{\max}\ge\kappa_2\ge2$ for the two longest chain lengths and $\nu_*$ for the
algebraic multiplicity. Then for every real weight $g\in L^2$, every $c\in\mathbb R$ and
every $\varepsilon_1$ with $\varepsilon_1\ipL{\ell}{g}\ne1$, the operator $A_k^{(1)}$ of the first stage alone, that is
\eqref{eq:prefeed} with $\varepsilon_2=0$ and $F_k=0$, satisfies
\begin{equation}\label{eq:ceiling}
\nu(z_*;A_k^{(1)})\ \ge\ \nu_*-\kappa_{\max}\ \ge\ \kappa_2\ \ge\ 2 .
\end{equation}
The same bound holds for a rank-one stage acting through the tangential wall velocity
alone. Consequently, at such an eigenvalue one rank-one stage cannot produce an
operator with algebraically simple spectrum, whichever wall velocity components it
acts through, so a second stage is necessary.
\end{proposition}

\begin{proof}
With $\varepsilon_2=0$, $\det\mathfrak D=D_1=1-\varepsilon_1\ipL{\rho_c(z)}{g}$. By
Lemma \ref{lem:resform}, the pole order of $z\mapsto\rho_c(z)$ at $z_*$ is at most
$\kappa_{\max}$, hence so is the pole order of $D_1$, and therefore
$\ord_{z_*}D_1\ge-\kappa_{\max}$. Lemma \ref{lem:counting} on a small circle around
$z_*$ gives $\nu(z_*;A_k^{(1)})\ge\nu_*-\kappa_{\max}=\sum_r\kappa_r-\kappa_{\max}\ge\kappa_2$.
For the tangential wall velocity the secular function is $1-\varepsilon\ipL{\Psi(z)}{g}$, with
the same pole-order bound. The final statements follow because an operator with an
eigenvalue of algebraic multiplicity $\ge2$ does not have simple spectrum.
\end{proof}

\section{Fredholm design and closed-loop stability}\label{sec:part-fredholm}

With Theorem \ref{thm:spectral} in hand, this part is short: the shifted Stokes
target, the transformation, and the closed loop.

\subsection{Fredholm design}\label{sec:design}

Standing assumptions from here on: Conditions \ref{cond:S}, \ref{cond:c},
\ref{cond:th1}, \ref{cond:th2} and $0<|\varepsilon_1|\le\bar\varepsilon_1$,
$0<|\varepsilon_2|\le\bar\varepsilon_2$, so that Theorem \ref{thm:spectral} is in force.

\begin{definition}[Target operator]\label{def:target}
Let the shift $q$ satisfy
\begin{equation}\label{eq:qdef}
q\ \ge\ \max\{\omega,\ \|B_k\|+2\}\qquad\text{and}\qquad
\lambda_n-q\in\rho(A_k')\ \text{for every }n\ge1 ,
\end{equation}
and define, on $\Hk$,
\begin{equation}\label{eq:Stilde}
\widetilde S_k\coloneqq S_k-qI,\qquad \mathcal D(\widetilde S_k)=\mathcal D ,
\end{equation}
the shifted Stokes operator, with eigenpairs $(\nu_n,e_n)$,
$\nu_n\coloneqq\lambda_n-q$, and $\nu_{\max}\coloneqq\nu_1=\lambda_1-q\le-\omega$.
\end{definition}

Since $q$ is free, the target decays at any prescribed rate. Its admissible values
omit only a locally finite subset of the half-line: on any compact
interval $[Q_1,Q_2]\subset[\|B_k\|+2,\infty)$, fix $N=N(Q_2)\ge m_0$ with
$d_n>4(Q_2+\|B_k\|+2)$ for $n\ge N$ and with $|\lambda_N|$ exceeding the radius of the
low region by $Q_2+1$. Then for every $q\in[Q_1,Q_2]$ and every $n\ge N$ the point
$\lambda_n-q$ lies inside a Stokes gap, at distance $\ge q-\|B_k\|-1\ge1$ from
$\mu'_n$, at distance $\ge d_n-q-\|B_k\|-1\ge q$ from every other $\mu'_m$ with
$m\ge m_0$, and at distance $\ge1$ below the bounded low spectrum: no $n\ge N$
resonates anywhere on the interval. For each of the finitely many $n<N$, the
resonance $\lambda_n-q=\mu'_m$ admits only finitely many $m$ while $q\in[Q_1,Q_2]$,
since $|\mu'_m|\to\infty$: finitely many excluded $q$ per compact interval. For the
fixed admissible $q$, write $n_q\coloneqq N(Q_2)$ for any $Q_2\ge q$.

The target lives on $\Hk=H^1_0$ with no-slip boundary conditions at both walls; the transformation
below maps the pre-compensated space $\Hkp$ onto $\Hk$. This is a difference from
the one-space Fredholm designs \cite{CoronLu2015,HayatLoko26,BKS26}: the nonlocal
boundary condition of the pre-feedback is carried entirely by the transformation, and
the target is an exactly Stokes-type system.

\begin{lemma}[Expansion of the tangential response]\label{lem:expand2}
$\Psi^{(2)}(z)\in\Hkp$ for $z\in\rho(A_k')$, it satisfies the resolvent formula
$\Psi^{(2)}(z)=\Psi^{(2)}(z_0)-(z-z_0)(z-A_k')^{-1}\Psi^{(2)}(z_0)$, and in the Riesz
basis of Theorem \ref{thm:spectral},
\begin{equation}\label{eq:expand2}
\Psi^{(2)}(z)=\sum_{m\ge1}\frac{b'_m}{z-\mu'_m}\,\chi'_m ,
\end{equation}
convergent in $\Hkp$, where $b'_m$ are the residue coordinates of Theorem
\ref{thm:spectral}(iii).
\end{lemma}

\begin{proposition}[Quadratic closeness]\label{prop:qc2}
For every $n\ge1$ set
\begin{equation}\label{eq:psin2}
\beta'_n\coloneqq\frac{b'_n}{\nu_n-\mu'_n},\qquad
\psi'_n\coloneqq\frac{\Psi^{(2)}(\nu_n)}{\beta'_n}.
\end{equation}
Then, with constants depending on the design and on $q$,
$c_3n/Re\le|\beta'_n|\le C_3n/Re$ for $n\ge m_0$, and
\begin{equation}\label{eq:qc2}
\|\psi'_n-\chi'_n\|_k\le\frac Cn,\qquad
\sum_{n\ge1}\|\psi'_n-\Xi e_n\|_k^2<\infty .
\end{equation}
\end{proposition}

\begin{proof}
By Definition \ref{def:target}, $1\le q-\|B_k\|-1\le|\nu_n-\mu'_n|\le q+\|B_k\|+1$
for $n\ge m_0$, and $|b'_n|\asymp n/Re$ by Theorem \ref{thm:spectral}(iii), giving the
bounds on $\beta'_n$. By \eqref{eq:expand2},
$\Psi^{(2)}(\nu_n)-\beta'_n\chi'_n=\sum_{m\ne n}\frac{b'_m}{\nu_n-\mu'_m}\chi'_m$,
whose $\Hkp$ norm squared is
$\le C\sum_{m\ne n}m^2/(Re^2|\nu_n-\mu'_m|^2)\le C_q$: for $n\ge n_q$ every
denominator with $m\ne n$ is $\ge q$ by Definition \ref{def:target} and
$\ge\tfrac12c_0'|m^2-n^2|/Re$ once $c_0'|m^2-n^2|/Re\ge2(q+\|B_k\|+1)$, so the split
of Proposition \ref{prop:bmc} bounds the sum, and the finitely many $n<n_q$ contribute
finitely. Divide by $\beta'_n$. Then
$\|\psi'_n-\Xi e_n\|_k\le\|\psi'_n-\chi'_n\|_k+\|\chi'_n-\ep n\|_k+\|\ep n-e_n\|_k+\|e_n-\Xi e_n\|_k$,
each $O(1/n)$: the second by Theorem \ref{thm:spectral}(ii), the third by
\eqref{eq:kato}, the fourth since
$\|\Xi e_n-e_n\|_k\le C|\varepsilon_1||\ipL{e_n}{g_1}|\le C/n$.
\end{proof}

\begin{proposition}[Completeness]\label{prop:complete2}
$\{\psi'_n\}_{n\ge1}$ is complete in $\Hkp$.
\end{proposition}

\begin{proof}
Let $\zeta\in\Hkp$ with $\ip{\psi'_n}{\zeta}=0$ for all $n$; put
$z_m\coloneqq\ip{\chi'_m}{\zeta}$, so $(z_m)\in\ell^2$, and
\[
G(\lambda)\coloneqq\sum_{m\ge1}\frac{b'_mz_m}{\lambda-\mu'_m},
\]
meromorphic with simple poles at $\mu'_m$, residues $b'_mz_m$, and $G(\nu_n)=0$ for all
$n$ by \eqref{eq:expand2}. Fix real $\mu_0>\|B_k\|+1$, in $\rho(A_k')$ by Proposition
\ref{prop:high}(iv), and set
$P(\lambda)\coloneqq\prod_m\frac{\lambda-\mu'_m}{\mu_0-\mu'_m}$,
$Q(\lambda)\coloneqq\prod_n\frac{\lambda-\nu_n}{\mu_0-\nu_n}$, convergent since
$\sum_m|\mu_0-\mu'_m|^{-1}<\infty$ ($|\mu'_m|\asymp m^2/Re$). The $\nu_n$ are distinct
and $H\coloneqq GP/Q$ is entire. On the circles $|\lambda|=r_j\coloneqq\tfrac12(|\lambda_j|+|\lambda_{j+1}|)$,
$j$ large enough that $d_j>8(q+\|B_k\|+2)$: the moduli $|\nu_n|=q+|\lambda_n|$ are
the Stokes moduli shifted by $q$, so the midpoint circles keep
$\dist(\lambda,\{\nu_n\})\ge d_j/8$; and, by Theorem
\ref{thm:spectral}(i) and \eqref{eq:kato}, $\dist(\lambda,\{\mu'_m\})\ge d_j/8$; hence
\[
|G(\lambda)|\le\|z\|_{\ell^2}\Big(\sum_m\frac{|b'_m|^2}{|\lambda-\mu'_m|^2}\Big)^{1/2}\le C,\qquad
\Big|\frac{P}{Q}\Big|\le C\exp\Big(C\sum_m\frac{|\nu_m-\mu'_m|}{|\lambda-\nu_m|}\Big)\le C' ,
\]
using $|\nu_m-\mu'_m|\le q+\|B_k\|+1$ for $m\ge m_0$, the finitely many low
differences bounded by a constant of the design, and
$\sum_m(|m^2-j^2|+cj)^{-1}=O(j^{-1}\log j)$. So $H$ is bounded on the circle family,
hence on $\mathbb C$ by the maximum principle, hence constant. Along real
$t\to+\infty$: $|G(t)|\to0$ by dominated convergence
($|t-\mu'_m|\ge t+cm^2/Re-C$), and $P(t)/Q(t)$ tends to the nonzero constant
$\prod_m\frac{\mu_0-\nu_m}{\mu_0-\mu'_m}$ since
$\sum_m|\nu_m-\mu'_m|/(t+|\nu_m|)=O(t^{-1/2})$. Hence $H\equiv0$, so $G\equiv0$, all
residues $b'_mz_m$ vanish, and $b'_m\ne0$ (Theorem \ref{thm:spectral}(iii)) gives
$z_m=0$ for all $m$, i.e. $\zeta=0$.
\end{proof}

\begin{theorem}[Riesz basis; Hilbert--Schmidt structure between the two spaces]\label{thm:riesz}
$\{\psi'_n\}_{n\ge1}$ is a Riesz basis of $\Hkp$. The operator $R:\Hk\to\Hkp$ defined by
$Re_n\coloneqq\psi'_n$ is bounded and boundedly invertible, and $R-\Xi$ is
Hilbert--Schmidt (as maps into $H^1_{(0}$). Consequently $T\coloneqq R^{-1}:\Hkp\to\Hk$
is bounded and boundedly invertible with $T\psi'_n=e_n$, and $T-\Xi^{-1}$ is
Hilbert--Schmidt.
\end{theorem}

\begin{proof}
$K\coloneqq\Xi^{-1}(R-\Xi)$ is Hilbert--Schmidt on $\Hk$: $Ke_n=\Xi^{-1}\psi'_n-e_n$ and
$\|\Xi^{-1}\psi'_n-e_n\|_k\le\|\Xi^{-1}\|\,\|\psi'_n-\Xi e_n\|_k$ is square-summable by
\eqref{eq:qc2}. So $R=\Xi(I+K)$ with $I+K$ Fredholm of index zero; $I+K$ is injective
iff $\{\psi'_n\}$ is $\omega$-independent and surjective iff $\{\psi'_n\}$ is complete
in $\Hkp$; index zero makes these equivalent, and completeness is Proposition
\ref{prop:complete2}. Hence $R$ is invertible, $\{\psi'_n\}=R(\{e_n\})$ is a Riesz
basis, and $T-\Xi^{-1}=-T(R-\Xi)\Xi^{-1}$ is Hilbert--Schmidt.
\end{proof}

\subsection{Closed-loop stability}\label{sec:cl}

\begin{definition}[Feedback and closed-loop realizations]\label{def:cl}
With
$\gamma_n\coloneqq\psi'_{n,y}(1)-c\varepsilon_1\ipL{\psi'_n}{g_1}-\varepsilon_2\ipL{\psi'_n}{g_2}$,
define
\begin{equation}\label{eq:F2}
F_k[V]\coloneqq\sum_{n\ge1}a_n\gamma_n\qquad\text{for }V=\sum_na_n\psi'_n\in\Hkp ,
\end{equation}
the abstract closed loop
\begin{equation}\label{eq:Acl2}
A_{\mathrm{cl}}\coloneqq T^{-1}\widetilde S_kT,\qquad
\mathcal D(A_{\mathrm{cl}})\coloneqq\{V\in\Hkp:\ TV\in\mathcal D\},
\end{equation}
and the PDE realizations
\begin{gather}
\mathcal D_{\max}'\coloneqq\{V\in H^3:\ V(0)=V_y(0)=0,\ V(1)=\varepsilon_1\ipL{V}{g_1}\},\qquad
A_{\max}V\coloneqq\Lapm^{-1}\mathcal L_kV,\label{eq:Amax2}\\
\mathcal D(\mathcal A_F)\coloneqq\{V\in\mathcal D_{\max}':\
V_y(1)=c\varepsilon_1\ipL{V}{g_1}+\varepsilon_2\ipL{V}{g_2}+F_k[V]\},\qquad
\mathcal A_F\coloneqq A_{\max}\big|_{\mathcal D(\mathcal A_F)} .\label{eq:AF2}
\end{gather}
\end{definition}

The operator $\mathcal A_F$ is the Orr--Sommerfeld evolution \eqref{eq:OSk} under the
full feedback \eqref{eq:prefeed}: the static pre-feedback in both boundary conditions,
and the Fredholm term $F_k$ in the tangential one.

\begin{lemma}[Feedback coefficients]\label{lem:gamman}
The coefficients $\gamma_n$ of Definition \ref{def:cl} satisfy
\begin{equation}\label{eq:gamman}
\gamma_n=1/\beta'_n\ \ \text{for every }n,\qquad (\gamma_n)_n\in\ell^2 .
\end{equation}
\end{lemma}

\begin{theorem}[Rapid stabilization at wavenumber $k$, no plant hypothesis]\label{thm:main}
Fix $k\ne0$, $Re>0$, admissible design parameters as in Theorem \ref{thm:spectral}, and
a target as in Definition \ref{def:target}. Then:
\begin{enumerate}
\item[(i)] $F_k$ is a bounded linear functional on $\Hkp$.
\item[(ii)] $A_{\mathrm{cl}}=\mathcal A_F$: the abstract closed loop is exactly the
Orr--Sommerfeld operator with the boundary feedback \eqref{eq:prefeed}, i.e.
\begin{equation}\label{eq:Uc2}
V_c=\varepsilon_1\ipL{V}{g_1},\qquad
U_c=\frac{i}{2\pi k}\Big[c\varepsilon_1\ipL{V}{g_1}+\varepsilon_2\ipL{V}{g_2}+F_k[V]\Big].
\end{equation}
\item[(iii)] $\mathcal A_F$ generates the analytic semigroup $T^{-1}e^{\widetilde S_kt}T$
on $\Hkp$, and
\begin{equation}\label{eq:decay2}
\|e^{\mathcal A_Ft}\|_{\Hkp\to\Hkp}\le\|T\|\,\|T^{-1}\|\,e^{\nu_{\max}t},\qquad
\nu_{\max}<0 .
\end{equation}
By \eqref{eq:uV} this is exponential decay of the $L^2$ energy
$\|u(k,\cdot)\|^2+\|V(k,\cdot)\|^2$ of the velocity pair. Given any $\omega>0$, the
choice $q\ge\omega$ in Definition \ref{def:target} achieves
$\nu_{\max}\le-\omega$: rapid stabilization at an arbitrary decay rate.
\item[(iv)] The target variable $\alpha\coloneqq TV\in\Hk$ satisfies the no-slip
shifted Stokes system $\alpha_t=\widetilde S_k\alpha$, with $\alpha(t)\in\mathcal D$ for
$t>0$, and for $t=0$ when $V(0)\in\mathcal D(A_{\mathrm{cl}})$.
\end{enumerate}
\end{theorem}

\begin{proof}
(i) $(\gamma_n)\in\ell^2$ by \eqref{eq:gamman} and $|\beta'_n|\ge c_3n/Re$; the
coordinates of the Riesz basis satisfy $\|(a_n)\|_{\ell^2}\le C\|V\|_k$; Cauchy--Schwarz.

(ii) \emph{Step 1: $A_{\mathrm{cl}}\subset\mathcal A_F$.} Finite combinations
$\mathcal C$ of $\{\psi'_n\}$ form a core for $A_{\mathrm{cl}}$ (finite combinations of
$\{e_n\}$ are a core for $\widetilde S_k$, and $T$ is an isomorphism intertwining the
graphs). Each $\psi'_n\in H^4$ satisfies
$\mathcal L_k\psi'_n=\nu_n\Lap\psi'_n$ and the three boundary conditions of
$\mathcal D(\mathcal A_F)$ (Definition \ref{def:cl}), and
$A_{\mathrm{cl}}\psi'_n=T^{-1}\widetilde S_ke_n=\nu_n\psi'_n=\Lapm^{-1}\mathcal L_k\psi'_n$,
since $\nu_n\psi'_n\in\Hkp$ solves $\Lap U=\mathcal L_k\psi'_n$. Hence on $\mathcal C$,
$\Lap(A_{\mathrm{cl}}V)=\mathcal L_kV$ in $H^{-1}$, so
$\|\Lap^2V\|_{H^{-1}}\le Re\big(\|A_{\mathrm{cl}}V\|_k+C\|V\|_k\big)$ (the pairing
$\langle\Lap W,\varphi\rangle=-\ipL{W_y}{\varphi_y}-4\pi^2k^2\ipL{W}{\varphi}$ is
bounded by $\|W\|_k$ for any $W\in H^1$), and Lemma \ref{lem:ell} gives the graph
estimate $\|V\|_{H^3}\le C(\|A_{\mathrm{cl}}V\|_k+\|V\|_k)$ on $\mathcal C$. For
$V\in\mathcal D(A_{\mathrm{cl}})$ take $V^{(j)}\in\mathcal C$ converging in graph norm;
then $V^{(j)}\to V$ in $H^3$, all traces and the bounded functionals
$\ipL{\cdot}{g_i}$, $F_k$ pass to the limit, and
$\Lapm^{-1}\mathcal L_k:H^3\cap\Hkp\to\Hkp$ is continuous. Hence
$V\in\mathcal D(\mathcal A_F)$ and $\mathcal A_FV=A_{\mathrm{cl}}V$.

\emph{Step 2: codimension.} Fix real $\lambda>\|B_k\|+1$, so
$\lambda\in\rho(A_k')$ (Proposition \ref{prop:high}(iv)) and
$\lambda\in\rho(A_{\mathrm{cl}})$ ($\nu_{\max}<0$). The kernel of $\lambda-A_{\max}$ on
$\mathcal D_{\max}'$ is the set of solutions of $\mathcal L_kV=\lambda\Lap V$ with
$V(0)=V_y(0)=0$ (a two-dimensional space) and $V(1)=\varepsilon_1\ipL{V}{g_1}$: at least
one-dimensional; exactly one, since a two-dimensional kernel would contain a nonzero
element additionally satisfying
$V_y(1)-c\varepsilon_1\ipL{V}{g_1}-\varepsilon_2\ipL{V}{g_2}=0$, i.e. an eigenvector of
$A_k'$, contradicting $\lambda\in\rho(A_k')$. It is spanned by $\Psi^{(2)}(\lambda)$.
For $V\in\mathcal D_{\max}'$ put
$W\coloneqq(\lambda-A_{\mathrm{cl}})^{-1}(\lambda-A_{\max})V\in\mathcal D(A_{\mathrm{cl}})$
(note $(\lambda-A_{\max})V\in\Hkp$); then $(\lambda-A_{\max})(V-W)=0$ by Step 1, so
$\mathcal D_{\max}'=\mathcal D(A_{\mathrm{cl}})\oplus\spn\{\Psi^{(2)}(\lambda)\}$.

\emph{Step 3: identification.} Set
$\Gamma[V]\coloneqq V_y(1)-c\varepsilon_1\ipL{V}{g_1}-\varepsilon_2\ipL{V}{g_2}-F_k[V]$
on $\mathcal D_{\max}'$, so $\mathcal D(\mathcal A_F)=\ker\Gamma\supset\mathcal D(A_{\mathrm{cl}})$
by Step 1. By \eqref{eq:expand2},
$\|\Psi^{(2)}(\lambda)\|_k^2\le C\sum_m(m/Re)^2(\lambda+cm^2/Re)^{-2}\to0$ as
$\lambda\to+\infty$ along the reals, so
$\Gamma[\Psi^{(2)}(\lambda)]=1-F_k[\Psi^{(2)}(\lambda)]\to1$; fix $\lambda$ large with
$\Gamma[\Psi^{(2)}(\lambda)]\ne0$. Any $V\in\ker\Gamma$ decomposes by Step 2 as
$V=W+t\Psi^{(2)}(\lambda)$ with $W\in\mathcal D(A_{\mathrm{cl}})$, and
$0=\Gamma[V]=t\,\Gamma[\Psi^{(2)}(\lambda)]$ forces $t=0$. Hence
$\mathcal D(\mathcal A_F)=\mathcal D(A_{\mathrm{cl}})$, and the operators agree there by
Step 1.

(iii) $\widetilde S_k$ is self-adjoint with spectrum in $(-\infty,\nu_{\max}]$;
similarity by the isomorphism $T$ preserves generation, analyticity, and gives
\eqref{eq:decay2}; then (ii) and \eqref{eq:uV}. The rate statement is Definition
\ref{def:target}. (iv) is \eqref{eq:Acl2}.
\end{proof}

\section{Closed loop in physical variables}\label{sec:phys}

This section assembles the modal designs into the physical channel: the closed loop is
a single analytic semigroup on a physical energy space, the pressure is reconstructed,
and Theorem \ref{thm:first} is proved. Wavenumbers are $k\in\mathbb Z/L$ and
$\mathcal B$ is the band of Theorem \ref{thm:first}. The design parameters can be
taken common to the whole band: choose $c$ outside the union of the finite excluded
sets of the finitely many $k\in\mathcal B$, then $\theta_1$ outside the union of the
corresponding discrete sets, a common $\varepsilon_1\ne0$ below the minimum threshold,
then $\theta_2$ and $\varepsilon_2$ likewise, and shifts $q(k)$ as in Definition
\ref{def:target} with $\nu_{\max}(k)=\lambda_1(k)-q(k)\le-\omega$; unions of finitely
many finite, respectively discrete, sets are finite, respectively discrete, so such
parameters exist.

\begin{lemma}[Reality]\label{lem:real}
The exclusion sets and thresholds of Theorem \ref{thm:spectral} can be taken identical
for $k$ and $-k$, and the designs chosen so that
$\psi'^{(-k)}_n=\overline{\psi'^{(k)}_n}$ and
$F_{-k}[\overline V]=\overline{F_k[V]}$. Consequently, for real physical initial data
the controls $V_c(x),U_c(x)$ below are real.
\end{lemma}

\begin{proof}
$\Lap$ and $\ell$ are even in $k$, $\mathcal L_{-k}\overline V=\overline{\mathcal L_kV}$,
and the weights $g_1,g_2$ and parameters $(c,\theta_i,\varepsilon_i)$ are real; hence
$\overline{\Psi^{(k)}(z)}=\Psi^{(-k)}(\bar z)$, $\overline{\rho_c^{(k)}(z)}=\rho_c^{(-k)}(\bar z)$,
$\overline{\mathfrak D^{(k)}(z)}=\mathfrak D^{(-k)}(\bar z)$,
$\sigma(A'_{-k})=\overline{\sigma(A'_k)}$, and all trace data conjugate. The exclusion
sets, defined by vanishing of moduli of conjugate-symmetric quantities, coincide.
Choosing the same real shift $q$ (admissible: $\rho(A'_{-k})=\overline{\rho(A'_k)}$
and $\lambda_n$ real), conjugation intertwines every object of the construction. Since a real field has $V(-k,y)=\overline{V(k,y)}$,
\eqref{eq:Uc2} gives $V_c(-k)=\overline{V_c(k)}$, $U_c(-k)=\overline{U_c(k)}$.
\end{proof}

\begin{proposition}[Mean mode]\label{lem:mean}
At $k=0$, continuity and $V(0)=0$ give $V\equiv0$, and the streamwise mean obeys
$u_t=\tfrac1{Re}u_{yy}$, $u(0)=0$, $u(1)=U_c(0,t)$. Let $c_0>0$ and
\begin{equation}\label{eq:k0}
k_0(y,\eta)\coloneqq-c_0Re\,\eta\,
\frac{I_1\big(\sqrt{c_0Re\,(y^2-\eta^2)}\big)}{\sqrt{c_0Re\,(y^2-\eta^2)}} ,
\end{equation}
$I_1$ the modified Bessel function. Then $w\coloneqq u-\int_0^yk_0(y,\eta)u(\eta)\dd\eta$
is a boundedly invertible Volterra transformation of $L^2(0,1)$, and under the
feedback $U_c(0,t)=\int_0^1k_0(1,\eta)u(\eta,t)\dd\eta$ the mean mode is mapped onto
$w_t=\tfrac1{Re}w_{yy}-c_0w$, $w(0)=w(1)=0$. Hence the closed-loop mean mode generates
an analytic semigroup on $L^2(0,1)$ with
$\|u(t)\|\le C_0e^{-(\pi^2/Re+c_0)t}\|u(0)\|$, $C_0$ the condition number of the
transformation; choosing $c_0\ge\omega$ gives decay at least at the rate $\omega$.
\end{proposition}

\begin{proof}
The kernel equations for the transformation are
$\tfrac1{Re}(k_{0,yy}-k_{0,\eta\eta})=c_0\,k_0$ with $k_0(y,0)=0$ and
$k_0(y,y)=-\tfrac{c_0Re}2y$, whose solution is \eqref{eq:k0}
\cite[Ch.~4]{SGK2009}; $k_0$ is continuous on the triangle, so the transformation is
bounded with bounded Volterra inverse. The target decays at
$\pi^2/Re+c_0$ in $L^2$, and the estimate transfers through the transformation.
\end{proof}

\begin{lemma}[Unactuated modes: dissipativity and sector]\label{lem:sector}
For $|k|\ge K$, compute the numerical range $W(A_k)$ of the block $A_k$ in
$\ip{\cdot}{\cdot}$, equivalently in the energy metric
$\ip{\cdot}{\cdot}/(4\pi^2k^2)$, which leaves $W$ unchanged. Then:
since $K\ge8Re$,
$\Rea W(A_k)\le-2\pi^2k^2/Re\le-\omega$ and $|\Ima W(A_k)|\le|\Rea W(A_k)|$, so
$W(A_k)\subset\Sigma\coloneqq\{z:\Rea z\le-\omega,\ |\Ima z|\le|\Rea z|+1\}$.
Consequently the direct sum $A_\infty$ of the unactuated blocks generates an
analytic semigroup on the $\ell^2$ sum of the modal energy spaces, with
$\|e^{tA_\infty}\|\le e^{-\omega t}$, per mode
$\|e^{tA_k}\|\le e^{-2\pi^2k^2t/Re}$, and $\|A_ke^{sA_k}\|\le C\,e^{C_1s}/s$ for
$s>0$, with $C=C(Re)$ independent of $k$ and $C_1\coloneqq1$.
\end{lemma}

\begin{proof}
For
$V\in\mathcal D$,
$\ip{A_kV}{V}=-\ipL{\mathcal L_kV}{V}=-\tfrac1{Re}\|\Lap V\|^2-8\pi ki\,q+16\pi ki\|V\|^2$,
$q\coloneqq\ipL{y(y-1)\Lap V}{V}$. By Proposition \ref{prop:ops}(i),
$\|\Lap V\|^2\ge4\pi^2k^2\|V\|_k^2$, and, $-\Lap$ having least eigenvalue
$p\coloneqq\pi^2+4\pi^2k^2$ on $H^1_0$, $\|V\|\le\|\Lap V\|/p$. Since
$|y(y-1)|\le\tfrac14$, $|q|\le\tfrac14\|\Lap V\|\,\|V\|\le\|\Lap V\|^2/(4p)$, and
$\|V\|^2\le\|\Lap V\|^2/p^2$. Hence
\begin{gather}
\Rea\ip{A_kV}{V}\le-\|\Lap V\|^2\Big(\frac1{Re}-\frac{2\pi|k|}{p}\Big)\le-\frac{\|\Lap V\|^2}{2Re}
\le-\frac{2\pi^2k^2}{Re}\,\|V\|_k^2,\label{eq:sectorRe}\\
|\Ima\ip{A_kV}{V}|\le\|\Lap V\|^2\Big(\frac{2\pi|k|}{p}+\frac{16\pi|k|}{p^2}\Big),\label{eq:sectorIm}
\end{gather}
and since $K\ge8Re$, so that $Re\le|k|/8$, and since $k^2/p\le1/(4\pi^2)$ and
$k^2/p^2\le1/(4\pi^4)$,
\begin{equation}\label{eq:sectorc}
Re\Big(\frac{2\pi|k|}{p}+\frac{16\pi|k|}{p^2}\Big)
\le\frac{\pi k^2}{4p}+\frac{2\pi k^2}{p^2}\le\frac1{16\pi}+\frac1{2\pi^3}<\frac12 ,
\end{equation}
which gives the second inequality in \eqref{eq:sectorRe} and
$|\Ima|\le|\Rea|$, for every $Re>0$; the bound $2\pi^2k^2/Re\ge\omega$ holds because
$K\ge\pi^{-1}\sqrt{\omega Re/2}$. Each block is closed with compact resolvent and
$\{\Rea z>\|B_k\|\}\subset\rho(A_k)$ (Proposition \ref{prop:ops}(ii)), so the component of
$\mathbb C\setminus\overline{W(A_k)}$ containing the far positive real axis meets
$\rho(A_k)$; the numerical-range bound then gives, on that component,
$\|(z-A_k)^{-1}\|\le1/\dist(z,W(A_k))\le1/\dist(z,\Sigma)$, uniformly in $k$. Since
$|\Ima z|\le|\Rea z|+C_1=\Rea(C_1-z)$ on $\Sigma$,
$\Sigma\subset\{z:|\arg(C_1-z)|\le\pi/4\}$: the bound is a sectorial estimate with
vertex $C_1$ and half-angle $\pi/4$. The direct sum has resolvent the direct sum, with
the same bound, hence generates an analytic semigroup, and the contour representation
$A_ke^{sA_k}=\tfrac1{2\pi i}\oint ze^{sz}(z-A_k)^{-1}\dd z$ over the shifted sector
boundary gives the uniform derivative bound. The contraction bounds follow from
Lumer--Phillips with \eqref{eq:sectorRe}.
\end{proof}

By \eqref{eq:uV}, the map $J_kV\coloneqq\big(V_y/(-2\pi ki),V\big)$ carries $\Hk$,
respectively $\Hkp$, with the norm $\|\cdot\|_k/(2\pi|k|)$, isometrically onto the
$k$-th physical energy subspace; the modal blocks below are transported by these maps.
At $k=0$ the feedback condition of Proposition \ref{lem:mean} lies in the domain of the
mean block, $L^2(0,1)$ carrying no trace.

\begin{definition}[Physical state space and generator]\label{def:X}
Let $X$ be the space \eqref{eq:Xspace} of Theorem \ref{thm:first}, with
$\|(u,V)\|_X^2\coloneqq\|u\|_{L^2}^2+\|V\|_{L^2}^2$, and let $A_X$ be the orthogonal
modal direct sum of the closed-loop blocks $\mathcal A_F(k)$ for
$k\in\mathcal B\setminus\{0\}$ (Theorem \ref{thm:main}), the closed-loop mean block of
Proposition \ref{lem:mean}, and the unactuated blocks of Lemma \ref{lem:sector}, with
$\mathcal D(A_X)$ the set of $w\in X$ whose modes lie in the block domains with
square-summable block images.
\end{definition}

\begin{lemma}[Uniform high-mode graph and pressure estimate]\label{lem:pgraph}
There is $C=C(Re)$ such that for every $|k|\ge K$ and $V\in\mathcal D$,
\begin{equation}\label{eq:kgraph}
\|V\|_{H^3}\le C(1+|k|)^3\big(\|A_kV\|_k+\|V\|_k\big),
\end{equation}
and, with $u\coloneqq V_y/(-2\pi ki)$, $u_t\coloneqq(A_kV)_y/(-2\pi ki)$ and
$2\pi ki\,p\coloneqq\tfrac1{Re}\Lap u+8\pi ki\,y(y-1)u+4(2y-1)V-u_t$,
\begin{equation}\label{eq:pk}
\|p\|_{L^2(0,1)}\le C(1+|k|)^3\big(\|A_kV\|_e+\|V\|_e\big),
\end{equation}
where $\|\cdot\|_e=\|\cdot\|_k/(2\pi|k|)$ is the modal energy norm.
\end{lemma}

By a solution of the closed loop with initial field $w_0\in X$ is meant
$w\in C([0,\infty);X)\cap C^1((0,\infty);X)$ with $w(t)\in\mathcal D(A_X)$ and
$\dot w=A_Xw$ for $t>0$, and $w(0)=w_0$; existence and uniqueness below are within this
class.

\begin{theorem}[Closed-loop linearized Navier--Stokes, periodic channel]\label{thm:phys}
Under a common-band design as above:
\begin{enumerate}
\item[(i)] $A_X$ generates an analytic semigroup on $X$, with
\begin{equation}\label{eq:phys}
\|e^{tA_X}\|_{X}\le C^{1/2}e^{-\omega t},\qquad t\ge0,
\end{equation}
$C$ the largest squared condition number of the transformations over the band, the
mean mode included; every wavenumber decays at least at the rate $\omega$ in its modal
energy, and $e^{tA_X}$ commutes with complex conjugation of the field.
\item[(ii)] For $w_0\in X$ and $t>0$, every mode is a classical solution; defining,
for $k\ne0$,
\begin{equation}\label{eq:pressure}
2\pi ki\,p(k,\cdot,t)\coloneqq\frac1{Re}\Lap u(k)+8\pi ki\,y(y-1)\,u(k)+4(2y-1)V(k)-u_t(k),
\end{equation}
and $p(0,\cdot,t)\coloneqq0$, the triple $(u,V,p)$ satisfies
\eqref{eq:linu}--\eqref{eq:cont} mode by mode, with the wall conditions \eqref{eq:bc},
the feedback \eqref{eq:Uc2} on $\mathcal B$ and $U_c=V_c=0$ off it, and
$p(\cdot,t)\in L^2(\Omega_L)$.
\end{enumerate}
\end{theorem}

\begin{proof}
(i) The modal subspaces are orthogonal and $A_X$-invariant; the resolvent of $A_X$ is
the direct sum of the block resolvents, bounded by the maximum of finitely many
sectorial bounds ($\mathcal A_F(k)=T_k^{-1}\widetilde S_kT_k$, Theorem
\ref{thm:main}(iii)) and the uniform bound of Lemma \ref{lem:sector}; hence $A_X$ is
sectorial. The decay: Parseval decouples the wavenumbers;
$\|e^{t\mathcal A_F(k)}\|\le\|T_k\|\|T_k^{-1}\|e^{-\omega t}$ on the band (Theorem
\ref{thm:main}(iii) with \eqref{eq:uV}), $\|e^{tA_0}\|\le C_0e^{-\omega t}$ for the mean mode (Proposition \ref{lem:mean}), and
$\|e^{tA_\infty}\|\le e^{-\omega t}$ off the band (Lemma \ref{lem:sector}); combining
gives \eqref{eq:phys}.
Reality: the conjugation symmetry of Lemma \ref{lem:real} intertwines the $\pm k$
blocks.

(ii) For $t>0$, analyticity gives $w(t)\in\mathcal D(A_X^r)$ for every $r$, and each
mode bootstraps through the smooth-coefficient fourth-order ODE to any spatial
regularity; in particular every mode is classical. The streamwise equation \eqref{eq:linu} holds by the definition
\eqref{eq:pressure} of $p$. For the normal equation: with
$\omega\coloneqq u_y(k)-2\pi ki\,V(k)$, continuity gives $\Lap V(k)=-2\pi ki\,\omega$,
and \eqref{eq:OSk} is exactly the vorticity equation
$\omega_t=\tfrac1{Re}\Lap\omega+8\pi ki\,y(y-1)\omega+8V$. On the other hand,
$\partial_y$ of \eqref{eq:linu} minus $2\pi ki$ times \eqref{eq:linV} eliminates $p$
and, using
$\partial_y\big[8\pi ki\,y(y-1)u\big]+4\partial_y\big[(2y-1)V\big]
=8\pi ki\,y(y-1)u_y+8V$ (the terms $8\pi ki(2y-1)u$ and $4(2y-1)V_y$ cancel by
continuity), yields the same vorticity identity minus $2\pi ki$ times the residual of
\eqref{eq:linV}. Hence that residual vanishes, and \eqref{eq:linV} holds; continuity
and the wall traces hold by the block domains, the controlled traces being
\eqref{eq:Uc2}. At $k=0$: $V\equiv0$ and \eqref{eq:linu} is the mean block of
Proposition \ref{lem:mean} under its feedback, while \eqref{eq:linV} reads $p_y=0$, so
the mean pressure is constant in $y$; it is determined only up to an additive
constant, and $p(0,\cdot,t)=0$ is one admissible choice.
Summability of the pressure: fix $t>0$ and set $s\coloneqq\min\{t,2\}/2$. For
$|k|\ge K$, write $V_k(t)=e^{sA_k}e^{(t-s)A_k}V_k(0)$; Lemma \ref{lem:sector} gives
$\|A_ke^{sA_k}\|\le Ce^{C_1s}/s$ with $C$ independent of $k$ and
$\|e^{(t-s)A_k}\|\le e^{-2\pi^2k^2(t-s)/Re}$ for $|k|\ge K$, so
\[
\|A_kV_k(t)\|_e+\|V_k(t)\|_e\le C_t\,e^{-\pi^2k^2t/Re}\,\|V_k(0)\|_e ,
\]
with $C_t$ locally bounded in $t>0$ (for $t\le2$, $t-s=t/2$; for $t>2$, $s=1$ and
$2(t-1)\ge t$). Lemma \ref{lem:pgraph} then gives
$\|p(k,\cdot,t)\|\le C_t(1+|k|)^3e^{-\pi^2k^2t/Re}\|w_{0,k}\|_e$, and
$\sup_k(1+|k|)^3e^{-\pi^2k^2t/Re}<\infty$ for each $t>0$, so
$\sum_k\|p(k,\cdot,t)\|^2\le C_t'\|w_0\|_X^2$; the finitely many modes with
$|k|<K$ contribute finitely, each being classical with $p(k,\cdot,t)\in L^2(0,1)$.
\end{proof}

With Theorem \ref{thm:phys}, the closed loop of the periodic channel is a well-posed
linear evolution in the primitive variables on the compatibility space $X$; initial
fields outside $X$ --- whose controlled modes violate the feedback trace at $t=0$ ---
are intrinsic to static feedback through the wall value and are not treated.

The band grows with $\omega$ like $\sqrt{\omega Re}$, and with it the
number of terms in \eqref{eq:synth}; the gains themselves do not grow with $|k|$, the
wavenumbers outside $\mathcal B$ being left to viscosity.

\begin{proof}[Proof of Theorem \ref{thm:first}]
Take a common-band design as above, with per-wavenumber targets of Definition
\ref{def:target} at rate $\nu_{\max}(k)\le-\omega$. By Theorem \ref{thm:phys}(i), $A_X$ generates an analytic semigroup on $X$ with
$\|e^{tA_X}\|_X\le C^{1/2}e^{-\omega t}$, so for every initial field in $X$ the closed
loop has the unique solution $e^{tA_X}w_0$, which satisfies \eqref{eq:decay-first} and
is real for real data; the pressure is Theorem \ref{thm:phys}(ii). The physical
feedback: expanding $V$ in Fourier series, the modal laws \eqref{eq:Uc2} on
$\mathcal B\setminus\{0\}$ assemble into \eqref{eq:Vcphys}--\eqref{eq:Ucphys} with the
displayed $\chi$ and the synthesis \eqref{eq:synth}; the conjugate symmetry of Lemma
\ref{lem:real} makes $\chi$ and $\mathcal G$ real, and the mean-mode law of Proposition
\ref{lem:mean}, which acts on $u(0,\cdot,t)=\tfrac1L\int_0^Lu(t,\xi,\cdot)\dd\xi$ with
the real kernel \eqref{eq:k0}, is the second term of \eqref{eq:Ucphys}. The kernels: for each
$k\in\mathcal B$, $\mathcal K^{(k)}$ is the integral kernel of $I-T_k$ on $\Hkp$: the
boundary part $(I-\Xi^{-1})V=\varepsilon_1\ipL{V}{g_1}\ell$ contributes
$\varepsilon_1\ell(y)g_1(\eta)$, which fixes
$\mathcal K(1,\eta)=\varepsilon_1g_1(\eta)$ and $\mathcal K(0,\eta)=0$ in
\eqref{eq:kpde-c}, while $\Xi^{-1}-T_k$ is Hilbert--Schmidt (Theorem \ref{thm:riesz})
and contributes a kernel vanishing at both walls in $y$. The wall-trace functionals in
\eqref{eq:kfeed} are the bounded functionals $\varepsilon_1\ipL{\cdot}{g_1}$ and
$c\varepsilon_1\ipL{\cdot}{g_1}+\varepsilon_2\ipL{\cdot}{g_2}+F_k$, in the pairing
sense of Remark \ref{rem:diag}; on $\mathcal D(\mathcal A_F)$, where
$T_kV\in\mathcal D\subset H^2_0$, they coincide with the classical traces, the gain
identity \eqref{eq:kfeed} with the decomposition \eqref{eq:prefeed} is the boundary
condition of $\mathcal D(\mathcal A_F)$ read together with $\alpha_y(1)=0$, and
$\mathcal K_y(0,\eta)=0$ expresses $(T_kV)_y(0)=0$. The system \eqref{eq:kpde} is the
formal transcription of the identity $T_kA_{\mathrm{cl}}=\widetilde S_kT_k$ of
\eqref{eq:Acl2} and Theorem \ref{thm:main}(ii), written in kernels.
\end{proof}

\section{Numerical illustration}\label{sec:numerics}

The purpose of this section is not fluid-dynamical. It is to show that the gain
constructed in Sections \ref{sec:part-spectral} and \ref{sec:part-fredholm} is not an
abstraction: it is computed from the plant data by the recipe of Theorem
\ref{thm:main}, and when it is applied at the wall of a nonlinear Navier--Stokes
simulation it does to the full flow what the theorem says the feedback does to the
linearization. The design is carried out on the
linearized plant; the test is run on the nonlinear one, at finite amplitude.

\begin{figure*}[p]
\centering
\includegraphics[width=\textwidth]{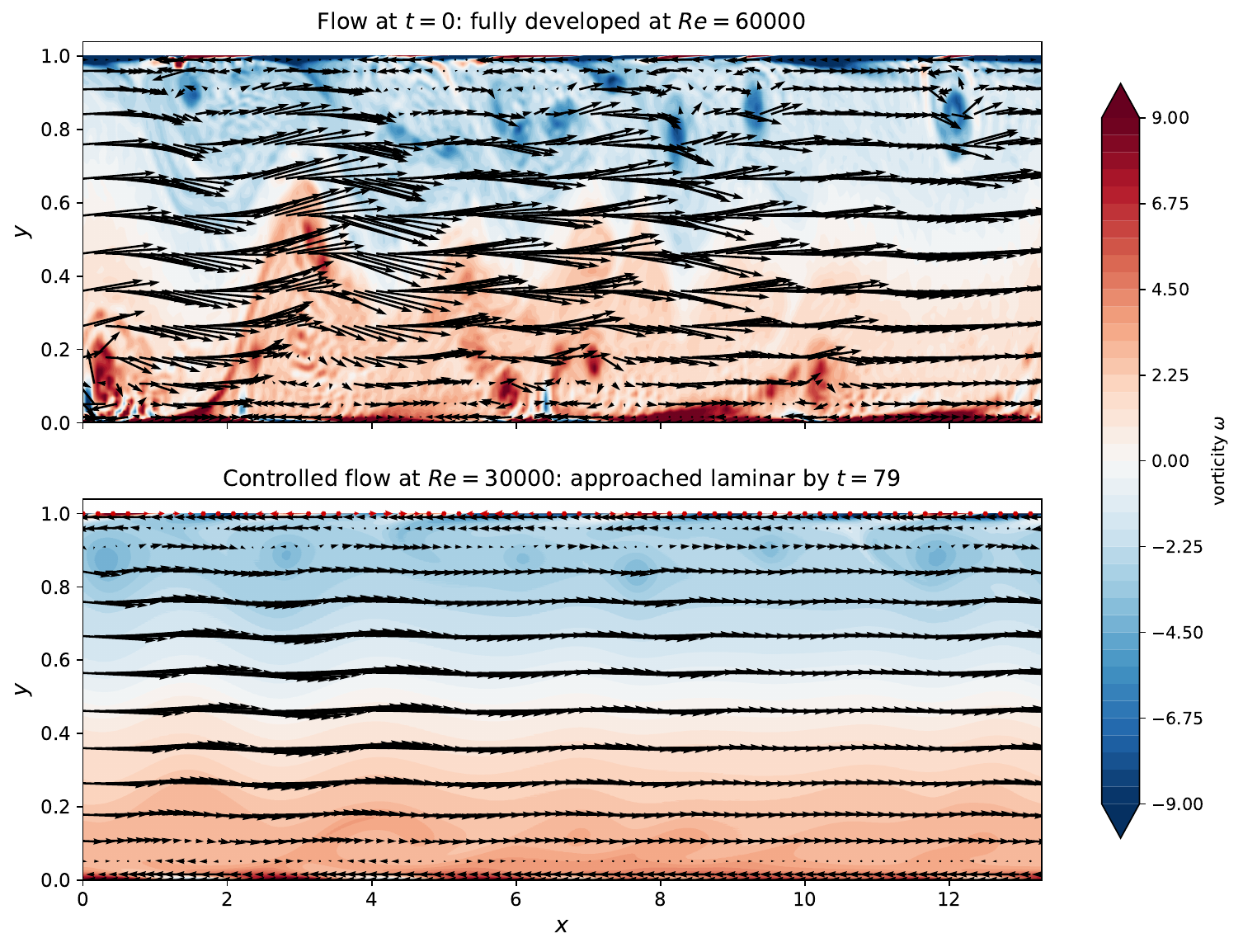}
\caption{Nonlinear channel flow at $Re=U^eH/\nu=30000$ ($15000$ on the half-height
convention, $2.6$ times critical), $L=13.3$, at which one wavenumber of the plant is
unstable, $k=0.3$. Color is vorticity, black arrows are velocity in the frame moving
at speed $0.20$, red arrows at the upper wall are the actuation, drawn at the same
scale as the velocity. Above: the initial condition, a fully developed flow carried
over from $Re=60000$, with $E'=2.52\times10^{-2}$. Below: the flow at $t=79.1$ under
the feedback of Theorem \ref{thm:first}, at $E'=2.38\times10^{-3}$.}
\label{fig:dns}
\end{figure*}

\subsection{Plant and numerical method}\label{sec:num-method}

The nonlinear two-dimensional Navier--Stokes equations are solved in the periodic
channel $(x,y)\in[0,L)\times[0,1]$ with $L=13.3$ and $Re=30000$ in the variables of
Section \ref{sec:model}, that is $Re=U^eH/\nu$ with $H$ the channel height and $U^e$
the centerline velocity of the parabolic profile; on the half-height convention of the
hydrodynamic stability literature this is $15000$, against a critical value of $5772$.
The formulation is vorticity--streamfunction, $\omega=u_y-V_x=\Delta\psi$ with
$u=\psi_y$ and $V=-\psi_x$, so that the pressure is eliminated and the wall conditions
of Section \ref{sec:model} are carried by $\psi$:
\begin{equation}\label{eq:dnsvort}
\omega_t+u\omega_x+V\omega_y=\frac1{Re}\Delta\omega .
\end{equation}
The lower wall is no-slip and impermeable; at the upper wall $u(t,x,1)=U_c(t,x)$ and
$V(t,x,1)=V_c(t,x)$ are the two actuated velocity components. The flux is held fixed
at its laminar value $2/3$.

Discretization is Fourier in $x$ with $N_x=256$ modes and dealiasing by the two-thirds
rule, and Chebyshev collocation in $y$ with $N_y=576$ points. Time stepping is
Crank--Nicolson on the viscous term and third-order Adams--Bashforth on the convective
term, with the step held below $2\times10^{-3}$ and reduced so that the streamwise
Courant number stays under $0.05$. The vorticity has no boundary condition of its own:
its wall values are the two unknowns of an influence-matrix problem, fixed at each step
so that the streamfunction recovered from $\omega$ carries the prescribed wall
velocities. The Helmholtz solves for $\omega$ and $\psi$ are performed in the
eigenbasis of the Chebyshev second-derivative matrix with homogeneous data, computed
once and serving every wavenumber; the eigenvalues are real and the eigenvector matrix
has condition number $3.3$, so the solves are accurate to $10^{-13}$ across the
wavenumber range.

Two checks establish that the plant is the intended one. The Orr--Sommerfeld eigenvalue
computed by the same Chebyshev discretization reproduces the benchmark of
\cite{Orszag71} at $Re=10000$, $\alpha=1$ to all eight published digits. The growth
rate of a small disturbance in the nonlinear code agrees with that eigenvalue to
$0.1\%$, and the parabolic profile is a steady solution of the discretization to
$10^{-11}$ over the duration of a run.

\subsection{Design choices}\label{sec:num-design}

In a channel of length $L=13.3$ the admissible wavenumbers are $k=m/L$ with $m$ an
integer, the corresponding mode being $e^{2\pi imx/L}$, and the Orr--Sommerfeld
operator at $Re=30000$ is unstable at exactly one of them, $m=4$, that is $k=0.3$, with $\lambda=+0.01147-0.40557i$ and phase speed $0.215$. The band
actuated is $m=1,\dots,5$, which is wider than the unstable set because the band of
Theorem \ref{thm:first} selects the slow wavenumbers and not the unstable ones: the
four long waves in the band are stable but decay far more slowly than the rate asked
of the closed loop. The mean mode carries no actuation in this run.

The design parameters are $c=0.7$, $\theta_1=1.3$, $\theta_2=3.1$ and
$\varepsilon_1=\varepsilon_2=0.2$. Theorem \ref{thm:first} places every controlled
wavenumber at one common shift $q$, which is what a decay rate uniform in $k$ requires.
In the nonlinear run the shift is instead graduated in the wavenumber: the unstable
wavenumber is assigned the closed-loop eigenvalue $-\sigma$ with $\sigma=0.15$, and
each stable wavenumber has its leading eigenvalue deepened by the factor $1+\beta$ with
$\beta=3$. Grading the shift is a specialization within the design, since the
construction is carried out one wavenumber at a time and the estimate of Theorem
\ref{thm:first} holds with $\omega$ the smallest of the assigned rates. Its purpose is
frugality of the actuation: a common shift large enough for the unstable wavenumber
drives the already well-damped long waves much further left than the nonlinear flow
requires, and enlarges the wall velocity throughout.

The pre-feedback does not stabilize: at $k=0.3$ it moves the unstable eigenvalue from
$+0.01147$ to the $\lambda'$ of Table \ref{tab:gains}, still in the right half plane.
The stabilization is the tangential input's.

\begin{table}[t]
\centering
\caption{The five actuated wavenumbers. Here $\lambda'$ is the leading eigenvalue of
the pre-compensated operator $A_k'$, $q_k$ the shift assigned to it, and $F_k$ the
Fredholm gain of \eqref{eq:prefeed}.}
\label{tab:gains}
\begin{tabular}{cccccc}
\hline
$m$ & $k$ & $\lambda'$ & $q_k$ & $\lambda'-q_k$ & $\max|F_k|$\\
\hline
1 & $0.075$ & $-0.0273$ & $0.0819$ & $-0.1092$ & $1.50$\\
2 & $0.150$ & $-0.0395$ & $0.1184$ & $-0.1579$ & $2.40$\\
3 & $0.225$ & $-0.0189$ & $0.0566$ & $-0.0755$ & $1.65$\\
4 & $0.300$ & $+0.0056$ & $0.1556$ & $-0.1500$ & $7.41$\\
5 & $0.375$ & $-0.0158$ & $0.0474$ & $-0.0632$ & $4.57$\\
\hline
\end{tabular}
\end{table}

The gain is assembled exactly as in the proof of Theorem \ref{thm:main}: the
pre-compensated operator $A_k'$ is discretized on the same Chebyshev grid with the two
static conditions \eqref{eq:prefeed} written into the boundary rows of the pencil; its
eigenvalues $\mu_n'$ and the residue coefficients $b_n'$ are obtained from the right
and left eigenvectors; the frequency responses $\Psi^{(2)}(\nu_n)$ are obtained by
solving the pencil with unit tangential wall velocity; and $F_k=\gamma^{\!\top}T_k^{-1}$
with $\gamma_n=1/\beta_n'$ as in Lemma \ref{lem:gamman}. Modes that are not moved have
$\nu_n=\mu_n'$ and therefore $\gamma_n=0$, so the construction leaves them untouched.
The closed-loop spectrum of the discretized pencil under the gains of Table
\ref{tab:gains} lies on the prescribed targets to ten digits. At these parameters the
transformation is ill conditioned, $\|T_k\|\|T_k^{-1}\|\approx10^{7}$, which is the
non-normality of the plant rather than an artifact of the design; the constants left
unquantified by Theorem \ref{thm:main} are therefore not benign, and their growth with
the number of moved eigenvalues is what limits the construction in double precision.
The corresponding quantity for a Volterra design of this plant is the kernel itself,
whose growth in $Re$ is the subject of the estimate in \cite[Prop.~1]{VK07}.

\subsection{Result}\label{sec:num-result}

\begin{figure}[t]
\centering
\includegraphics[width=0.95\textwidth]{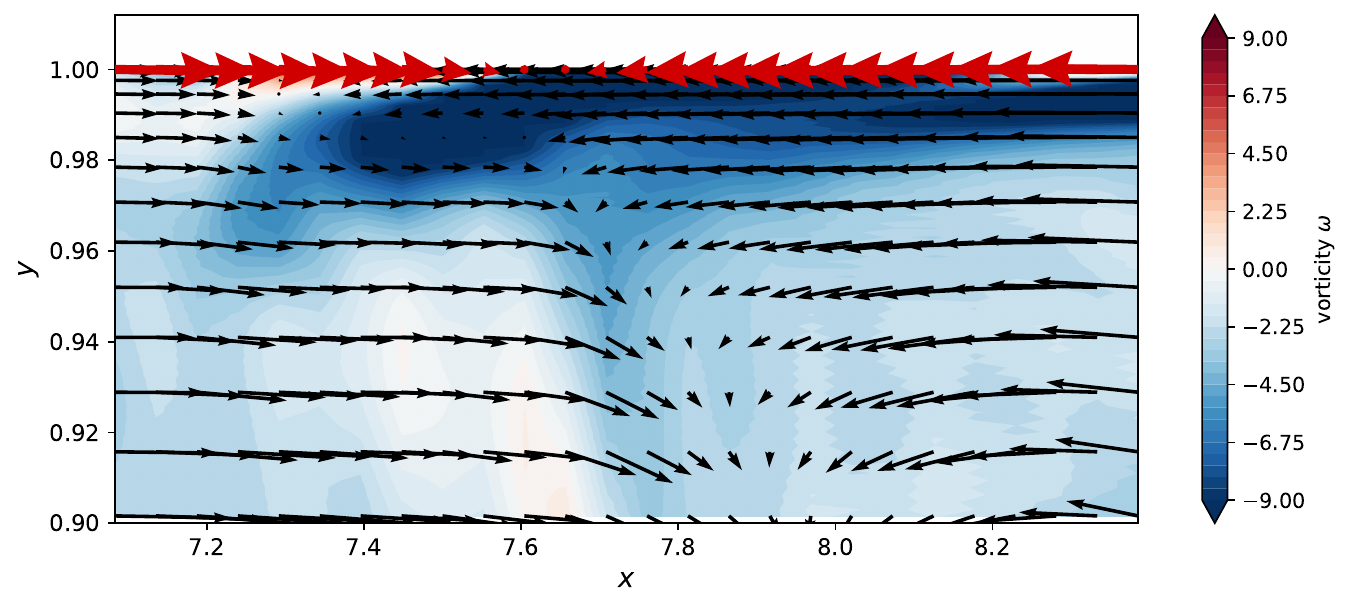}
\caption{A tenth of the channel length at the actuated wall, at $t=24$, early in the
transient. Red arrows are the wall actuation $(U_c,V_c)$, black arrows the velocity in
the frame moving at speed $0.20$, at one common scale. The two wall jets meet near
$x=7.7$, and the fluid they deliver leaves the wall there, along the pale separator in
the vorticity.}
\label{fig:wall}
\end{figure}

The initial condition is the upper panel of Figure \ref{fig:dns}: a fully developed
flow at $Re=60000$, at which two wavenumbers are unstable, grown from the unstable
eigenmode until it saturates and then carried over to the $Re=30000$ plant. It is used
in place of the saturated state of the $Re=30000$ plant itself in order to make the
test harder while keeping the initial datum a physical one --- a solution of the
nonlinear equations rather than a synthetic disturbance, and one of larger amplitude
and richer structure than the plant being controlled would produce on its own. Its
disturbance energy is $E'\coloneqq\langle(u-U^e)^2+V^2\rangle=2.52\times10^{-2}$.

The feedback is applied from the initial instant. The simulation achieved exponential
stabilization and was terminated once the energy had fallen under a tenth of its
initial value, at $t\approx80$. The lower panel of Figure \ref{fig:dns} is the flow at
$t=79.1$, at $E'=2.38\times10^{-3}$; that snapshot is shown because the fluid dynamics
are still alive in it, and yet attenuated to near extinction of the perturbations. The
tangential wall velocity peaks at $0.54$ at the initial instant, when the control is
confronted with the carried-over flow, and is down to $0.090$ at $t=79.1$; the normal
wall velocity is $0.0031$ there, smaller by a factor of thirty, as the pre-feedback
strength $\varepsilon_1$ dictates.

Figure \ref{fig:wall} is a segment of a tenth of the channel length at the actuated
wall, taken early in the transient where the actuation is largest. The tangential wall
velocity changes sign twice across the segment; the two wall jets that result collide
near $x=7.7$, and the fluid they deliver is carried away from the wall in the wedge of
converging velocity vectors below that point, which reaches to $y=0.92$. The
actuation and the velocity just inside the wall are drawn at a single scale and are of
the same length. Over this window the wall velocity $|(U_c,V_c)|$ has maximum $0.315$
and mean $0.196$, against a maximum of $0.310$ and a mean $0.188$ for the perturbation
velocity $|(u-U^e,V)|$ at $y=0.95$; that is, the actuation drives the fluid at the wall
no harder than the fluid a twentieth of the channel height inside is already moving.

\section{Conclusions}\label{sec:status}

The linearized channel is stabilized at an arbitrary rate by static feedback of both
wall velocity components, with no hypothesis on the Orr--Sommerfeld spectrum: multiple
eigenvalues and Jordan chains of any admissible length are admitted, and every
restriction falls on the design parameters. The Fredholm transformation reaches a
Navier--Stokes plant, the first in that line to be non-normal and to carry a mass
operator in its evolution.

Against \cite{VK07}, the same plant with the same actuators: on the normal wall
velocity the two laws ask for about the same thing, a one-pole filter per wavenumber
there and a static average of arbitrarily small gain here, both subject to the same
$t=0$ wall-compatibility condition. What separates the designs is the
transformation. The Volterra transformation requires spatial causality, which the
first stage of \cite{VK07} imposes; the Fredholm transformation accepts the plant's
nonlocal structure, while its gain is constructed from spectral data
instead of in closed form, through an analysis with thresholds and exclusion sets that
have no Volterra counterpart. The arbitrary decay rate discriminates nothing: a
damped Volterra target yields it too, and the damping already appears in \cite{VTC08},
where increasing it speeds the controlled modes of this same channel. For anyone designing a controller for this
flow, \cite{VK07} remains an instrument with certain advantages: its kernels are
explicit, its
stability holds in $H^1$ and $H^2$, and its gain follows from the plant data by a
computation rather than from a spectral construction. The two constructions place
their conditioning differently. The Volterra gain is the kernel trace itself, and it
carries the growth in $Re$ of the kernel that Proposition 1 of \cite{VK07} bounds; the
Fredholm gain is bounded, and the growth is confined to the transformation that
produces it, which is inverted once. What is added here is the reach
of the method and the results of Sections~\ref{sec:part-spectral} and
\ref{sec:necessity} on the plant itself: two-input controllability of the linearized
flow, and controllability of every mode from the tangential wall velocity for every
plant.

In the nonlinear run of Section \ref{sec:numerics} the feedback acts through twenty
real inputs on a discretized state of about 150,000 variables with quadratic
nonlinearity, one input per seventy-four hundred states, a ratio at which nonlinear
feedback design is seldom exercised.

The test at $Re=30000$ demonstrates that applying the feedback is equivalent to slowing
down the flow twentyfold. The decay rates assigned in Table \ref{tab:gains} are those
the uncontrolled case has at a Reynolds number of about $1500$.

The Volterra design passed from this channel to three dimensions, and to conducting
fluids, without a change of method \cite{VSK09}, the magnetic field there being
removable to leave the 3-D Navier--Stokes channel. The same passage is expected for
the design given here: the third velocity component adds a second wavenumber and a
larger modal system, and the analysis of Sections \ref{sec:part-spectral} and
\ref{sec:part-fredholm} is carried out one wavenumber at a time.


\appendix
\section{Proofs of supporting results}\label{app:proofs}

\subsection{Energy space, plant operators, and Stokes spectrum}

\begin{proof}[Proof of Lemma \ref{lem:ell}]
$\Lap^2V=V_{yyyy}-8\pi^2k^2V_{yy}+16\pi^4k^4V$, and $\|V_{yy}\|_{H^{-1}}\le\|V_y\|$,
$\|V\|_{H^{-1}}\le\|V\|$, so $\|V_{yyyy}\|_{H^{-1}}\le\|\Lap^2V\|_{H^{-1}}+C\|V\|_{H^1}$.
For $g\in L^2$, $c_k\|g-\bar g\|\le\|g_y\|_{H^{-1}}\le C_k\|g-\bar g\|$ with
$\bar g=\int_0^1g$: in the unweighted dual norm this is an equality, since
$\{\varphi_y:\varphi\in H^1_0\}$ is the mean-zero subspace of $L^2$ and
$\|\varphi_y\|=\|\varphi\|_{H^1_0}$, and the $k$-weighted and unweighted $H^{-1}$
norms are equivalent. Applying this to $g=V_{yyy}$:
$\|V_{yyy}\|\le C_k'\|V_{yyyy}\|_{H^{-1}}+|V_{yy}(1)-V_{yy}(0)|$, and
$|V_{yy}(1)-V_{yy}(0)|\le2\|V_{yy}\|_{L^\infty}\le\epsilon\|V_{yyy}\|+C_\epsilon\|V_{yy}\|$
by the one-dimensional Agmon inequality; similarly
$\|V_{yy}\|\le\epsilon\|V_{yyy}\|+C_\epsilon\|V_y\|$. Absorbing gives \eqref{eq:ell}.
\end{proof}

\begin{proof}[Proof of Proposition \ref{prop:ops}(i)]
Consider the form $\mathfrak s[V,W]\coloneqq-\tfrac1{Re}\ipL{\Lap V}{\Lap W}$ on $H^2_0$,
closed, symmetric, nonpositive in $\Hk$, with $H^2_0$ dense in $\Hk$. For $V,W\in H^2_0$,
integrating by parts once,
$\ipL{\Lap V}{\Lap W}=\ipL{\Lap V}{W_{yy}}-4\pi^2k^2\ipL{\Lap V}{W}$, and
$\ipL{\Lap V}{W_{yy}}=-\ipL{(\Lap V)_y}{W_y}$ when $\Lap V\in H^1$, with no boundary term
since $W_y(0)=W_y(1)=0$. Thus $\mathfrak s[V,\cdot]$ is $\|\cdot\|_k$-bounded iff
$\Lap^2V\in H^{-1}$ iff $V\in H^3$, and then
$\mathfrak s[V,W]=-\tfrac1{Re}\ipL{\Lap^2V}{W}=\ip{S_kV}{W}$. So the form operator is
$S_k$ on $\mathcal D$; compact resolvent from
$\mathcal D\hookrightarrow\Hk$ compact; $H^4\cap H^2_0$ dense in the graph norm by
\eqref{eq:ell}. Eigenfunctions solve a constant-coefficient ODE, hence smooth. For
$V\in H^2_0$,
$\|\Lap V\|^2-4\pi^2k^2\|V\|_k^2=\|V_{yy}\|^2+4\pi^2k^2\|V_y\|^2>0$ unless $V=0$ (the
cross term integrates by parts without boundary contribution), so every eigenvalue
satisfies $\lambda<-4\pi^2k^2/Re$.
\end{proof}

\begin{proof}[Proof of Proposition \ref{prop:ops}(ii)]
For $V\in H^1_0$, $\Lap V\in H^{-1}$ with $\|\Lap V\|_{H^{-1}}=\|V\|_k$. Multiplication by
$y(y-1)$ is bounded on $H^1_0$, hence by duality on $H^{-1}$, and
$\Lap^{-1}:H^{-1}\to H^1_0$ is an isometry; $-16\pi ki\,\Lap^{-1}$ is bounded
$L^2\to H^1_0$. So $B_k$ is bounded, and its $\Hk$-adjoint is explicit: for
$V\in\mathcal D$ and $\zeta\in H^2\cap H^2_0$,
$\ip{B_kV}{\zeta}=-8\pi ki\ipL{\Lap V}{y(y-1)\zeta}+16\pi ki\ipL{V}{\zeta}
=-\ipL{\Lap V}{B_k^*\zeta}$ with
\begin{equation}\label{eq:Bstar}
B_k^*\zeta=-8\pi ki\,y(y-1)\zeta+\Lap^{-1}\big(16\pi ki\,\zeta\big),
\end{equation}
consistent with \eqref{eq:Lkstar}: $\Lap(S_k+B_k^*)\zeta=\mathcal L_k^*\zeta$ on
$H^4\cap H^2_0$. The spectral statements are those for bounded
perturbations of a self-adjoint operator with compact resolvent; the numerical range
bound gives the half-plane. If $(A_k-z_*)\phi_0=0$ then
$\Lap^2\phi_0=Re\,\Lap(z_*\phi_0-B_k\phi_0)\in L^2$, so $\phi_0\in H^4$; for a Jordan
chain $(A_k-z_*)\phi_j=\phi_{j-1}$ the same bootstrap applies inductively. For
$A_k^*$: by \eqref{eq:Bstar}, $B_k^*\zeta\in H^3$ whenever $\zeta\in H^3\cap H^2_0$,
so the chain equation gives
$\Lap^2\zeta_j=Re\,\Lap(\bar z_*\zeta_j+\zeta_{j-1}-B_k^*\zeta_j)\in H^1\subset L^2$,
and $\zeta_j\in H^4$ inductively.
\end{proof}

\begin{proof}[Proof of Lemma \ref{lem:trace}]
For $W\in H^4$: integrating $\ipL{\Lap^2W}{\zeta}$ by parts twice with
$\zeta=\zeta_y=0$ at the walls gives $\ipL{\Lap W}{\Lap\zeta}$; then
\[
\ipL{\Lap W}{\Lap\zeta}
=\Big[W_y\overline{\Lap\zeta}-W\,\overline{(\Lap\zeta)_y}\Big]_0^1+\ipL{W}{\Lap^2\zeta}
=W_y(1)\overline{\zeta_{yy}(1)}-W(1)\overline{\zeta_{yyy}(1)}+\ipL{W}{\Lap^2\zeta},
\]
using $W(0)=W_y(0)=0$, $\Lap\zeta(1)=\zeta_{yy}(1)$ and
$(\Lap\zeta)_y(1)=\zeta_{yyy}(1)$ (from $\zeta(1)=\zeta_y(1)=0$). For the advection term,
$\ipL{8\pi ki\,y(y-1)\Lap W}{\zeta}=8\pi ki\ipL{\Lap W}{y(y-1)\zeta}$, and two
integrations by parts produce the bracket
$[W_y\,\overline{y(y-1)\zeta}-W\,\overline{(y(y-1)\zeta)_y}]_0^1$, which vanishes since
$y(y-1)$ and $(y(y-1)\zeta)_y=(2y-1)\zeta+y(y-1)\zeta_y$ vanish at both walls (using
$\zeta\in H^2_0$). The constant term contributes its conjugate. Extend to $W\in H^3$ by
density; both sides are $H^3$-continuous in $W$. The last statement follows by pairing
with $\Lap^{-1}$: for $\zeta\in H^4\cap H^2_0$ and $W\in\mathcal D$ (so $W(1)=W_y(1)=0$),
$\ip{A_kW}{\zeta}=-\ipL{\mathcal L_kW}{\zeta}=-\ipL{W}{\mathcal L^*_k\zeta}
=\ip{W}{\Lap^{-1}\mathcal L^*_k\zeta}$, the last equality valid since
$\Lap^{-1}\mathcal L^*_k\zeta\in H^1_0$ and $W\in H^1_0$; the core property follows from
\eqref{eq:ell} as in Proposition~\ref{prop:ops}(i).
\end{proof}

\begin{proof}[Proof of Lemma \ref{lem:pairconv}]
One integration by parts gives
$\ip{W}{\zeta}=\big[W\overline{\zeta_y}\big]_0^1-\ipL{W}{\Lap\zeta}$, and the bracket
vanishes because $\zeta_y(0)=\zeta_y(1)=0$.
\end{proof}

\begin{proof}[Proof of Proposition \ref{prop:stokes-all}(i)]
$S_ke=\lambda e$ with $e\in\mathcal D$ is $\Lap(\Lap-Re\,\lambda)e=0$ with the
no-slip conditions $e(0)=e_y(0)=e(1)=e_y(1)=0$ at both walls. Solutions are spanned by $e^{\pm2\pi ky}$ and
$e^{\pm\beta y}$, $\beta^2=4\pi^2k^2+Re\,\lambda=-\sigma^2$, so $\beta=i\sigma$. Imposing
$e(0)=e_y(0)=0$ leaves the span of $f,g$. The conditions $e(1)=e_y(1)=0$ have a
nontrivial solution iff $f(1)g_y(1)-f_y(1)g(1)=0$. Direct differentiation gives
$g_y(1)=2\pi kf(1)$, so the condition is $2\pi kf(1)^2-f_y(1)g(1)=0$. With
$f(1)=\cosh2\pi k-\cos\sigma$, $f_y(1)=2\pi k\sinh2\pi k+\sigma\sin\sigma$,
$g(1)=\sinh2\pi k-\tfrac{2\pi k}{\sigma}\sin\sigma$, expansion and
$\cosh^2-\sinh^2=1$, $\cos^2+\sin^2=1$ yield \eqref{eq:char}.
\end{proof}

\begin{proof}[Proof of Proposition \ref{prop:stokes-all}(ii)]
Set
$\Phi(\sigma)\coloneqq\frac{\sigma^2-4\pi^2k^2}{\sigma}\sin\sigma-\frac{4\pi k}{\sinh2\pi k}\big(1-\cosh(2\pi k)\cos\sigma\big)$
and $\Phi_0(\sigma)\coloneqq\sigma\sin\sigma$. On $\Gamma_n=\{|\sigma-n\pi|=\pi/2\}$,
write $\sigma=n\pi+z$: $|\sin\sigma|=|\sin z|$, and
$|\sin z|^2=\sin^2(\Rea z)+\sinh^2(\Ima z)\ge(2\Rea z/\pi)^2+(\Ima z)^2=1$ for
$|z|=\pi/2$; also $|\sin z|,|\cos z|\le\cosh(\pi/2)$. Hence on $\Gamma_n$,
$|\Phi_0(\sigma)|\ge(n+\tfrac12)\pi$ while
$|\Phi(\sigma)-\Phi_0(\sigma)|\le C_k$, bounded in $n$. For $n\ge n_0$ with
$(n_0+\tfrac12)\pi>C_k$, Rouch\'e's theorem gives exactly one zero of $\Phi$ inside
$\Gamma_n$, simple, and real since $\Phi(\bar\sigma)=\overline{\Phi(\sigma)}$. The disks
cover $\sigma\ge(n_0+\frac12)\pi$ except the points $(n+\frac32)\pi$, where
$|\sin\sigma|=1$ and $\Phi\ne0$ by the same bound. Since every Stokes eigenvalue
corresponds to a real $\sigma>0$ (Proposition \ref{prop:ops}(i)), these roots are all the Stokes
eigenvalues with $\sigma\ge(n_0+\frac12)\pi$. With $\sigma_n=(n+1)\pi+\epsilon_n$,
$|\epsilon_n|<\pi/2$, \eqref{eq:char} gives
$\big|\frac{\sigma_n^2-4\pi^2k^2}{\sigma_n}\sin\epsilon_n\big|\le C_k$, hence
$|\sin\epsilon_n|\le C_k'/n$, and $|\sin\epsilon|\ge\frac2\pi|\epsilon|$ on
$|\epsilon|\le\pi/2$ gives $\epsilon_n=O(n^{-1})$; expanding
$\sin\epsilon_n=\epsilon_n+O(n^{-3})$, $\cos\epsilon_n=1+O(n^{-2})$ in \eqref{eq:char}
gives \eqref{eq:sigman}. For \eqref{eq:gap}:
$\lambda_n-\lambda_{n+1}=(\sigma_{n+1}-\sigma_n)(\sigma_{n+1}+\sigma_n)/Re$ with
$\sigma_{n+1}-\sigma_n=\pi+\epsilon_{n+1}-\epsilon_n\ge\pi-Cn^{-1}\ge\pi/2$ (enlarging
$n_0$) and $\sigma_{n+1}+\sigma_n\ge(2n+1)\pi$ for $n\ge n_0$.
For the count below $(n_0+\frac12)\pi$, put $x=y-\tfrac12$, $a\coloneqq2\pi k$; the
clamped conditions split into $\sigma\tan(\sigma/2)=-a\tanh(a/2)$ for eigenfunctions
even in $x$ and $\sigma\cot(\sigma/2)=a\coth(a/2)$ for odd ones, the poles of $\tan$
and $\cot$ forcing both coefficients to vanish. The left sides have derivatives
$(\sigma+\sin\sigma)/(2\cos^2(\sigma/2))>0$ and
$(\sin\sigma-\sigma)/(2\sin^2(\sigma/2))<0$, and run from $-\infty$ to $0$ on
$\big((2r-1)\pi,2r\pi\big)$ and from $+\infty$ to $0$ on $\big(2r\pi,(2r+1)\pi\big)$,
giving one root in each; on $(0,\pi)$ the even sides have opposite signs and the odd
left side is below $2<a\coth(a/2)$. Hence exactly one root in each
$\big(n\pi,(n+1)\pi\big)$, $n\ge1$, and none below.
For \eqref{eq:gapjm}: for $j\ge n_0$ it follows from
$\sigma_j\in(j\pi,(j+1)\pi)$; for $j<n_0$ and $n\ge n_0$,
$|\lambda_j-\lambda_n|\ge|\lambda_n|-|\lambda_{n_0}|\ge c\,n^2/Re\ge c\,|j^2-n^2|/(n_0^2Re)$
once $n\ge2n_0$, and the finitely many remaining pairs are handled by decreasing $c_0'$.
\end{proof}

\begin{proof}[Proof of Proposition \ref{prop:stokes-all}(iii)]
From $a_nf(1)+b_ng(1)=0$ and \eqref{eq:sigman},
$b_n/a_n=-(\cosh2\pi k-(-1)^{n+1})/\sinh2\pi k+O(n^{-1})$, bounded. Then
$e_{n,y}=a_n(2\pi k\sinh2\pi ky+\sigma_n\sin\sigma_ny)+b_n2\pi k(\cosh2\pi ky-\cos\sigma_ny)$,
so $\|e_{n,y}\|^2=a_n^2\sigma_n^2\|\sin\sigma_ny\|^2(1+O(n^{-1}))=\tfrac12a_n^2\sigma_n^2(1+O(n^{-1}))$
and $4\pi^2k^2\|e_n\|^2=O(a_n^2)$, giving the first two claims. Next
$e_{n,yy}(1)=a_n(4\pi^2k^2\cosh2\pi k+\sigma_n^2\cos\sigma_n)+b_n(4\pi^2k^2\sinh2\pi k+2\pi k\sigma_n\sin\sigma_n)$
is dominated by $a_n\sigma_n^2\cos\sigma_n=a_n\sigma_n^2(-1)^{n+1}(1+O(n^{-2}))$. For
\eqref{eq:efest3}, from $f_{yyy}=8\pi^3k^3\sinh(2\pi ky)-\sigma^3\sin(\sigma y)$ and
$g_{yyy}=8\pi^3k^3\cosh(2\pi ky)+2\pi k\sigma^2\cos(\sigma y)$,
\[
e_{n,yyy}(1)=a_n\big(8\pi^3k^3\sinh2\pi k-\sigma_n^3\sin\sigma_n\big)
+b_n\big(8\pi^3k^3\cosh2\pi k+2\pi k\,\sigma_n^2\cos\sigma_n\big).
\]
Here $\sin\sigma_n=(-1)^{n+1}\epsilon_n(1+O(n^{-2}))$ with $\epsilon_n$ from
\eqref{eq:sigman}, so
\[
-a_n\sigma_n^3\sin\sigma_n=2\sqrt2\,\sigma_n\,\frac{2\pi k}{\sinh2\pi k}\,(-1)^{n+1}\big(\cosh2\pi k-(-1)^{n+1}\big)\big(1+O(n^{-1})\big),
\]
while
$b_n\,2\pi k\sigma_n^2\cos\sigma_n=-\sqrt2\,\sigma_n\frac{2\pi k}{\sinh2\pi k}(-1)^{n+1}\big(\cosh2\pi k-(-1)^{n+1}\big)(1+O(n^{-1}))$;
the hyperbolic terms are $O(1)$. Adding the two contributions gives \eqref{eq:efest3}.
Since $\cosh2\pi k>1$, the right side of \eqref{eq:efest3} is nonzero for $n\ge n_0$
(enlarging $n_0$ if needed). Finally $|\ipL{h}{e_n}|\le\|h\|\,\|e_n\|$.
\end{proof}

\begin{proof}[Proof of Proposition \ref{prop:stokes-all}(iv)]
Suppose $S_ke=\lambda e$, $e\ne0$, and $e(1)=e_y(1)=e_{yy}(1)=0$. Expanding in the
basis centered at $y=1$,
$e=A_1\cos\sigma(y-1)+B_1\sin\sigma(y-1)+C_1\cosh2\pi k(y-1)+D_1\sinh2\pi k(y-1)$, the
conditions at $y=1$ give $A_1+C_1=0$, $B_1\sigma+D_1\,2\pi k=0$ and
$-A_1\sigma^2+C_1\,4\pi^2k^2=0$; the first and third force
$C_1(4\pi^2k^2+\sigma^2)=0$, so $A_1=C_1=0$ and
$e=B_1[\sin\sigma(y-1)-\tfrac{\sigma}{2\pi k}\sinh2\pi k(y-1)]$. The no-slip condition
at $y=0$ gives $e_y(0)=B_1\sigma[\cos\sigma-\cosh2\pi k]=0$; since
$\cosh2\pi k>1\ge\cos\sigma$ for $k\ne0$ and real $\sigma$ (Proposition
\ref{prop:ops}), $B_1=0$ and $e\equiv0$. For simplicity: $S_k$ is self-adjoint, so
algebraic and geometric multiplicities agree, and a two-dimensional eigenspace would
contain a nonzero solution of the linear condition $e_{yy}(1)=0$.
\end{proof}

\begin{proof}[Proof of Lemma \ref{lem:bm}]
Define $K$ by $Ke_{n,i}\coloneqq(Q_n-Q_n^0)e_{n,i}$, so $(I+K)e_{n,i}=Q_ne_{n,i}$ and
$\|K\|_{HS}^2\le\sum_n\operatorname{rank}Q_n^0\,\|Q_n-Q_n^0\|^2<\infty$; thus $I+K$ is
Fredholm of index zero. If $\sum_{n,i}c_{n,i}Q_ne_{n,i}=0$ with $(c_{n,i})\in\ell^2$,
applying the bounded $Q_m$ termwise gives $Q_m\sum_ic_{m,i}e_{m,i}=0$; and $Q_m$ is
injective on $\operatorname{ran}Q_m^0$, since $x=Q_m^0x$, $Q_mx=0$ imply
$\|x\|=\|(Q_m^0-Q_m)x\|\le\|Q_m-Q_m^0\|\|x\|$. So all $c_{m,i}=0$, $I+K$ is injective,
hence invertible, and $\{Q_ne_{n,i}\}$ is the image of an orthonormal basis under a
bounded invertible operator.
\end{proof}

\subsection{Frequency responses and pre-compensated operator}

\begin{proof}[Proof of Lemma \ref{lem:resform}]
$W\coloneqq\Psi(z)-\Psi(z_0)$ has all four boundary values zero, so
$W\in H^4\cap H^2_0\subset\mathcal D$, and
$(\mathcal L_k-z\Lap)W=(z-z_0)\Lap\Psi(z_0)$. Apply $\Lap^{-1}$: since $W\in H^1_0$,
$\Lap^{-1}\Lap W=W$; since $\Psi(z_0)(1)=0$, $\Lap^{-1}\Lap\Psi(z_0)=\Psi(z_0)$. Thus
$(A_k-z)W=(z-z_0)\Psi(z_0)$, which is \eqref{eq:resPsi}. For $\rho_c$: now
$\rho_c(z_0)(1)=1$, so $\Lap^{-1}\Lap\rho_c(z_0)=\rho_c(z_0)-\ell$, the subtraction of
the harmonic part matching the traces; the same computation gives \eqref{eq:resrho}.
The meromorphic continuation and the pole-order bound follow because the $z$-dependence
sits entirely in $(z-A_k)^{-1}$ applied to a fixed vector of $\Hk$
(note $\rho_c(z_0)-\ell\in H^1_0$), and the principal part of the resolvent at $z_*$ is
$\sum_{j=1}^{\kappa_{\max}}(z-z_*)^{-j}N^{j-1}P_{z_*}$ with $N$ the nilpotent part on
$\operatorname{ran}P_{z_*}$.
\end{proof}

\begin{proof}[Proof of Lemma \ref{lem:chains}]
The chain relations at the differential level read
$\mathcal L_k^*\zeta_p=\Lap(\bar z_*\zeta_p+\zeta_{p-1})$ (with $\zeta_{-1}\coloneqq0$),
by applying $\Lap$ to $A_k^*\zeta_p=\Lap^{-1}\mathcal L_k^*\zeta_p$. Set
$x_p(z)\coloneqq-\ipL{\rho_c(z)}{\Lap\zeta_p}$. Apply \eqref{eq:trace} with
$W=\rho_c(z)$, $\zeta=\zeta_p$, using $\mathcal L_k\rho_c=z\Lap\rho_c$,
$\ipL{\Lap\rho_c}{\zeta_p}=\ipL{\rho_c}{\Lap\zeta_p}$ (no boundary terms since
$\zeta_p\in H^2_0$), $\rho_{c,y}(1)=c$, $\rho_c(1)=1$:
\[
-z\,x_p=c\,\alpha_p-\beta_p+\ipL{\rho_c}{\Lap(\bar z_*\zeta_p+\zeta_{p-1})}
=c\,\alpha_p-\beta_p-z_*x_p-x_{p-1},
\]
that is, $(z-z_*)x_p=\beta^c_p+x_{p-1}$, and induction from $x_{-1}=0$ gives the first
formula of \eqref{eq:xy}. For $\Psi$, the traces are $\Psi_y(1)=1$, $\Psi(1)=0$, so the
recursion is $(z-z_*)y_p=-\alpha_p+y_{p-1}$ with
$y_p\coloneqq-\ipL{\Psi(z)}{\Lap\zeta_p}$.
\end{proof}

\begin{proof}[Proof of Lemma \ref{lem:inj}]
$\zeta\in H^4\cap H^2_0$ solves the fourth-order ODE
$\mathcal L_k^*\zeta=\bar z_*\Lap\zeta$ with smooth coefficients and leading coefficient
$1/Re$. If also $\zeta_{yy}(1)=\zeta_{yyy}(1)=0$, then together with
$\zeta(1)=\zeta_y(1)=0$ all four Cauchy data at $y=1$ vanish, so $\zeta\equiv0$ by
uniqueness for linear ODEs. Injectivity: a nontrivial kernel element of the trace map
would be such a $\zeta$; and a linear injection into $\mathbb C^2$ bounds the dimension
by $2$. The statement for $A_k$ is identical with $\mathcal L_k\phi=z_*\Lap\phi$.
\end{proof}

\begin{proof}[Proof of Lemma \ref{lem:Xi}]
$\Xi V\in H^1_{(0}$ and
$(\Xi V)(1)=\varepsilon_1\ipL{V}{g_1}/(1-\varepsilon_1\ipL{\ell}{g_1})$; a direct check
gives $(\Xi V)(1)=\varepsilon_1\ipL{\Xi V}{g_1}$, so $\Xi V\in\Hkp$. For $W\in\Hkp$,
$V\coloneqq W-W(1)\ell$ satisfies $V(1)=0$ and
$\varepsilon_1\ipL{V}{g_1}=\varepsilon_1\ipL{W}{g_1}-W(1)\varepsilon_1\ipL{\ell}{g_1}
=W(1)(1-\varepsilon_1\ipL{\ell}{g_1})$, whence $\Xi V=V+W(1)\ell=W$. Boundedness of both
maps and the smallness estimate are immediate from
$|\ipL{V}{g_1}|\le\|g_1\|\|V\|\le C\|V\|_k$.
\end{proof}

\begin{proof}[Proof of Lemma \ref{lem:counting}]
Conjugate to $\Hk$: $\hat A\coloneqq\Xi^{-1}A_k'\Xi$, which has the same spectrum and
multiplicities as $A_k'$. For $G\in\Hk$: $(\Xi G)(1)\ell$ cancels inside
\eqref{eq:resAp}, giving
\begin{equation}\label{eq:resdiff}
(z-\hat A)^{-1}-(z-A_k)^{-1}
=\big(\rho_c(z)-\ell\big)\otimes\mathfrak a(z)+\Psi(z)\otimes\mathfrak b(z),
\end{equation}
where, with $R(z)\coloneqq(z-A_k)^{-1}$,
$\mathfrak a=[\mathfrak D^{-1}]_{11}\,\varepsilon_1\ipL{R(z)\,\cdot}{g_1}+[\mathfrak D^{-1}]_{12}\,\varepsilon_2\ipL{R(z)\,\cdot}{g_2}$
and
$\mathfrak b=[\mathfrak D^{-1}]_{21}\,\varepsilon_1\ipL{R(z)\,\cdot}{g_1}+[\mathfrak D^{-1}]_{22}\,\varepsilon_2\ipL{R(z)\,\cdot}{g_2}$
(here $u\otimes\mathfrak f$ denotes $G\mapsto\mathfrak f[G]\,u$, and $\Xi^{-1}\rho_c=\rho_c-\ell$,
$\Xi^{-1}\Psi=\Psi$ since $\Psi(1)=0$). The difference is rank two, hence trace class,
and by Lemma \ref{lem:resform},
$\varepsilon_i\ipL{R(z)(\rho_c(z)-\ell)}{g_i}=-\varepsilon_i\ipL{\dot\rho_c(z)}{g_i}=\mathfrak D_{i1}'(z)$
and $\varepsilon_i\ipL{R(z)\Psi(z)}{g_i}=-\varepsilon_i\ipL{\dot\Psi(z)}{g_i}=\mathfrak D_{i2}'(z)$.
Therefore
\[
\tr\big[(z-\hat A)^{-1}-(z-A_k)^{-1}\big]
=\mathfrak a\big[\rho_c-\ell\big]+\mathfrak b\big[\Psi\big]
=\tr\big(\mathfrak D(z)^{-1}\mathfrak D'(z)\big)
=\frac{(\det\mathfrak D)'(z)}{\det\mathfrak D(z)}
\]
by Jacobi's formula. Integrating over $\Gamma$ and using that the trace of a Riesz
projection is the enclosed algebraic multiplicity,
$\nu(\Gamma;\hat A)-\nu(\Gamma;A_k)=\frac1{2\pi i}\oint_\Gamma\tr\big[(z-\hat A)^{-1}-(z-A_k)^{-1}\big]\dd z$,
which is \eqref{eq:counting} by the argument principle for the meromorphic function
$\det\mathfrak D$.
\end{proof}

\begin{proof}[Proof of Lemma \ref{lem:disc}]
Entirety is clear. If $\ipL{v}{g_\theta}\equiv0$ then all $\theta$-derivatives at
$\theta=0$ vanish: $\int_0^1v(y)(y-1)^j\dd y=0$ for every $j\ge0$, so $v$ is orthogonal
to all polynomials and $v=0$. For $\Psi$: by Watson's lemma with $s=1-y$,
$\Psi(z_0)(1-s)=-s+O(s^2)$ (since $\Psi(1)=0$, $\Psi_y(1)=1$), so
$\int_0^1\Psi(z_0)(1-s)e^{-\theta s}\dd s=-\theta^{-2}+O(\theta^{-3})$.
\end{proof}

\subsection{Splitting, high modes, and closed loop}

\begin{proof}[Proof of Lemma \ref{lem:t2}]
\emph{Independence of $\varepsilon_1$.} With $Q\coloneqq\ipL{\Psi}{g_1}$,
$P\coloneqq\ipL{\rho_c}{g_1}$,
\[
F\coloneqq\frac{\varepsilon_1Q}{1-\varepsilon_1P}=\frac{Q}{\varepsilon_1^{-1}-P}
=-\frac QP\cdot\frac1{1-(\varepsilon_1P)^{-1}} .
\]
$1/P$ has a zero of order $\kappa_1$ at $z_*$, so $(\varepsilon_1P)^{-1}$ has a zero of
order $\kappa_1$; hence the Taylor coefficients of $F$ at $z_*$ of orders
$0,\dots,\kappa_1-1$ coincide with those of $-Q/P$ and are $\varepsilon_1$-free.
(That $F$ is regular at $z_*$ follows since $P$ has a genuine pole there.) In
$\Psi^{(1)}=\Psi+F\rho_c$, the factor $\rho_c$ has pole order $\le\kappa_1$, so the
principal part of $F\rho_c$ at $z_*$ uses only those first $\kappa_1$ Taylor
coefficients of $F$. Therefore the principal part of $\ipL{\Psi^{(1)}}{g_2}$ equals that
of $\ipL{\Psi}{g_2}-(Q/P)\ipL{\rho_c}{g_2}=N/P$, which is \eqref{eq:Nfun}.

\emph{Laurent bookkeeping.} Write $u=z-z_*$ and, per chain $r$, form the polynomials
\[
A^{(r)}(u)=\sum_{p=0}^{\kappa_r-1}\alpha^{(r)}_pu^p,\qquad
B^{(r)}(u)=\sum_{p=0}^{\kappa_r-1}\beta^{c,(r)}_pu^p,\qquad
G_i^{(r)}(u)=\sum_{a=0}^{\kappa_r-1}\ipL{\phi^{(r)}_a}{g_i}u^a .
\]
By Lemma \ref{lem:chains} and the biorthogonal normalization \eqref{eq:biorth}, the
principal parts at $z_*$ are
\[
\ipL{\Psi}{g_i}\ \dot=\ -\sum_ru^{-\kappa_r}\big[A^{(r)}G_i^{(r)}\big]_{<\kappa_r},\qquad
\ipL{\rho_c}{g_i}\ \dot=\ \sum_ru^{-\kappa_r}\big[B^{(r)}G_i^{(r)}\big]_{<\kappa_r},
\]
where $[X]_{<\kappa}$ truncates a power series past $u^{\kappa-1}$ and $\dot=$ means
equality modulo functions regular at $z_*$. Form $N$. The same-chain contribution of chain $r$ is
\begin{equation}\label{eq:samechain}
-u^{-2\kappa_r}\Big(\big[B^{(r)}G_1^{(r)}\big]_{<\kappa_r}\big[A^{(r)}G_2^{(r)}\big]_{<\kappa_r}
-\big[B^{(r)}G_2^{(r)}\big]_{<\kappa_r}\big[A^{(r)}G_1^{(r)}\big]_{<\kappa_r}\Big)
=O\big(u^{-\kappa_r}\big),
\end{equation}
because the untruncated products cancel identically,
$B^{(r)}G_1^{(r)}\cdot A^{(r)}G_2^{(r)}=B^{(r)}G_2^{(r)}\cdot A^{(r)}G_1^{(r)}$, while
each truncation error starts at internal order $\kappa_r$ against a bounded cofactor;
and $-\kappa_r\ge-\kappa_1>-\kappa_1-\kappa_2$, so no same-chain term reaches the
bottom order --- including the equal-chain case $\kappa_1=\kappa_2$, where
\eqref{eq:samechain} gives order $-\kappa$ against a bottom of $-2\kappa$. The
cross-chain products reach the bottom order only through their constant coefficients:
writing $X_0$ for the $u^0$ coefficient of a bracket,
\begin{multline}\label{eq:Nbottom}
N_{-\kappa_1-\kappa_2}
=-B^{(1)}_0G^{(1)}_{1,0}A^{(2)}_0G^{(2)}_{2,0}-B^{(2)}_0G^{(2)}_{1,0}A^{(1)}_0G^{(1)}_{2,0}\\
+B^{(1)}_0G^{(1)}_{2,0}A^{(2)}_0G^{(2)}_{1,0}+B^{(2)}_0G^{(2)}_{2,0}A^{(1)}_0G^{(1)}_{1,0}
=\big(A^{(1)}_0B^{(2)}_0-A^{(2)}_0B^{(1)}_0\big)\det\big[G^{(r)}_{i,0}\big],
\end{multline}
with $A^{(r)}_0=\alpha^{(r)}_0$, $B^{(r)}_0=\beta^{c,(r)}_0$,
$G^{(r)}_{i,0}=\ipL{\phi^{(r)}_0}{g_i}$; the $c$-terms cancel in the first factor. Thus
$N$ has no Laurent coefficients below $-\kappa_1-\kappa_2$. Dividing by
$P=u^{-\kappa_1}(R^c+O(u))$: with $1/P=u^{\kappa_1}\sum_{m\ge0}s_mu^m$, $s_0=1/R^c$,
\begin{equation}\label{eq:division}
[N/P]_j=\sum_{m\ge0}N_{j-m-\kappa_1}\,s_m ,
\end{equation}
so for $j\le-\kappa_2-1$ every index on the right falls below the bottom order and the
coefficient vanishes, giving pole order $\le\kappa_2$; at $j=-\kappa_2$ only $m=0$
survives, $[N/P]_{-\kappa_2}=N_{-\kappa_1-\kappa_2}/R^c$, which is \eqref{eq:t2}.
\end{proof}

\begin{proof}[Proof of Lemma \ref{lem:tails}]
Expand in the unconditional decomposition of Corollary \ref{cor:decomp}. For (i):
$\ipL{\Psi(z)}{g}=\ipL{P_{\mathrm{low}}\Psi(z)}{g}+\sum_{j\ge n_1}\frac{b_j}{z-\lp j}\ipL{\ep j}{g}$
by \eqref{eq:tailpair}. Here $|b_j|\le Cj/Re$, $|\ipL{\ep j}{g}|\le\|\ep j\|\,\|g\|\le C/j$
(Propositions \ref{prop:kato}, \ref{prop:stokes-all}(iii)), and on $\Gamma_m$,
$|z-\lp j|\ge c|j^2-m^2|/(2Re)$ for $j\ne m$ while $|z-\lp m|=d_m/4\ge c_0m/(4Re)$; the
$j=m$ term is $O(1/m)$ and $\sum_{j\ne m}|j^2-m^2|^{-1}=O(\tau_m)$. For the low part,
the identity $(z-z_0)(z-A_k)^{-1}=I+(A_k-z_0)(z-A_k)^{-1}$ and \eqref{eq:resPsi} give
$P_{\mathrm{low}}\Psi(z)=-(z-A_k)^{-1}(A_k-z_0)P_{\mathrm{low}}\Psi(z_0)$, of norm
$\le C/|z|\le CRe/m^2$. (ii) is identical using
$\ip{\rho_c(z)-\ell}{\tep j}=\beta^c[\tep j]/(z-\lp j)$ from \eqref{eq:tailpair},
$\ipL{\ell}{\Lap\tep j}=0$, and
$P_{\mathrm{low}}(\rho_c(z)-\ell)=-(z-A_k)^{-1}(A_k-z_0)P_{\mathrm{low}}(\rho_c(z_0)-\ell)$,
the projection being applied before $A_k$ because $\rho_c(z_0)-\ell\notin\mathcal D$.
For (iii):
$\ipL{(z-A_k)^{-1}G}{g}=\ipL{(z-A_k)^{-1}P_{\mathrm{low}}G}{g}+\sum_j\frac{\ip{G}{\tep j}}{z-\lp j}\ipL{\ep j}{g}$;
Cauchy--Schwarz with $\big(\sum_jj^{-2}|z-\lp j|^{-2}\big)^{1/2}=O(m^{-2})$ (the $j=m$
term is $O(m^{-1}\cdot Re/m)$, the sum near $j=m$ is $O(Re^2m^{-4})$) and the low part
$O(1/|z|)$ give the bound. (iv) is the vector version of the sums in (i)--(ii) with
$\ell^2$ instead of $\ell^1$ coefficients. On $C_j$ the same decompositions apply. The moduli $|\lambda_m|$ increase in $m$, so
for $m\le j$, $|z|-|\lambda_m|\ge r_j-|\lambda_j|=\tfrac12(|\lambda_{j+1}|-|\lambda_j|)\ge d_j/2$,
and for $m\ge j+1$, $|\lambda_m|-|z|\ge|\lambda_{j+1}|-r_j\ge d_j/2$; hence
$|z-\lp m|\ge\big||z|-|\lambda_m|\big|-\|B_k\|\ge d_j/4$ for every $m$. Moreover, for
$m\le j-1$, $\big||z|-|\lambda_m|\big|\ge|\lambda_j|-|\lambda_m|\ge c_0'|j^2-m^2|/Re$,
and for $m\ge j+2$,
$\big||z|-|\lambda_m|\big|\ge|\lambda_m|-|\lambda_{j+1}|\ge c_0'|m^2-(j+1)^2|/Re\ge\tfrac12c_0'|m^2-j^2|/Re$,
both by \eqref{eq:gapjm}; absorbing $\|B_k\|$ after enlarging $m_0$,
$|z-\lp m|\ge c|m^2-j^2|/Re$ for $|m-j|\ge2$, while for $m\in\{j,j+1\}$ the bound
$d_j/4\ge c_0j/(4Re)$ serves. And $|z|\asymp j^2/Re$; every estimate goes through with
$m$ replaced by $j$. For real $z\to+\infty$ all denominators are
$\ge z+cj^2/Re$ and the same computations give the stated decay.
\end{proof}

\begin{proof}[Proof of Lemma \ref{lem:expand2}]
Membership in $\Hkp$ and the domain-difference computation: for
$W_d\coloneqq\Psi^{(2)}(z)-\Psi^{(2)}(z_0)$, all three (linear, nonlocal) boundary
conditions of $\mathcal D(A_k')$ are satisfied by $W_d$, since both responses satisfy the
same inhomogeneous set \eqref{eq:Psi2}; and
$\Lap(A_k'W_d)=\mathcal L_kW_d=z\Lap\Psi^{(2)}(z)-z_0\Lap\Psi^{(2)}(z_0)$ identifies
$A_k'W_d=z\Psi^{(2)}(z)-z_0\Psi^{(2)}(z_0)$ (the right side lies in $\Hkp$), whence
$(A_k'-z)W_d=(z-z_0)\Psi^{(2)}(z_0)$. Let $J'_m$ be the coordinate functionals of
$\{\chi'_m\}$ and $c_m(z)\coloneqq J'_m[\Psi^{(2)}(z)]$. Applying $J'_m$ to the
resolvent formula and using $J'_m[(z-A_k')^{-1}V]=J'_m[V]/(z-\mu'_m)$:
$c_m(z)=c_m(z_0)(\mu'_m-z_0)/(z-\mu'_m)$, a pure simple-pole rational function whose
residue is by definition the residue coordinate $b'_m$. Convergence of
\eqref{eq:expand2}: $(b'_m/(z-\mu'_m))_m\in\ell^2$ by Theorem \ref{thm:spectral}(iii)
and the gaps.
\end{proof}

\begin{proof}[Proof of Lemma \ref{lem:gamman}]
The tangential trace combination of $\Psi^{(2)}$ is identically one \eqref{eq:Psi2}, so
$\gamma_n$ is that combination evaluated at $\nu_n$ and divided by $\beta'_n$, i.e.
$\gamma_n=1/\beta'_n$; summability is $|\beta'_n|\ge c_3n/Re$ from Proposition
\ref{prop:qc2}.
\end{proof}

\begin{proof}[Proof of Lemma \ref{lem:pgraph}]
From $\Lap(A_kV)=\mathcal L_kV$,
$\Lap^2V=Re\big[\Lap(A_kV)-8\pi ki\,y(y-1)\Lap V+16\pi ki\,V\big]$. For $W\in H^1_0$
and $\varphi\in H^1_0$,
$|\ipL{\Lap W}{\varphi}|\le\|W_y\|\|\varphi_y\|+4\pi^2k^2\|W\|\|\varphi\|\le(1+2\pi|k|)\|W\|_k\|\varphi\|_{H^1_0}$,
so $\|\Lap W\|_{H^{-1}}\le(1+2\pi|k|)\|W\|_k$ in the unweighted dual norm;
multiplication by $y(y-1)$ is bounded on $H^1_0$, hence on $H^{-1}$, with a $k$-free
constant; and $\|16\pi ki\,V\|_{H^{-1}}\le C\|V\|_k$. Hence
$\|\Lap^2V\|_{H^{-1}}\le CRe(1+|k|)^2(\|A_kV\|_k+\|V\|_k)$. Retaining the
$k$-dependence in the elliptic estimate of Lemma \ref{lem:ell}:
$\|V_{yyyy}\|_{H^{-1}}\le\|\Lap^2V\|_{H^{-1}}+8\pi^2k^2\|V_y\|+16\pi^4k^4\|V\|
\le\|\Lap^2V\|_{H^{-1}}+C(1+|k|)^3\|V\|_k$, and the one-dimensional chain there has
$k$-free constants, giving \eqref{eq:kgraph}. For \eqref{eq:pk}, divide
\eqref{eq:pressure} by $2\pi ki$ and estimate term by term with
$u=V_y/(-2\pi ki)$: $\|\Lap u\|\le(\|V_{yyy}\|+4\pi^2k^2\|V_y\|)/(2\pi|k|)$,
$\||k|y(y-1)u\|\le C\|V\|_k$, $\|(2y-1)V\|\le C\|V\|$,
$\|u_t\|\le\|A_kV\|_k/(2\pi|k|)$; assembling with \eqref{eq:kgraph} and converting by
$\|\cdot\|_k=2\pi|k|\,\|\cdot\|_e$ gives \eqref{eq:pk}.
\end{proof}

\medskip\paragraph{Acknowledgment.}
The author's problems, ideas, and results were developed with the assistance of Claude and ChatGPT in final theorem formulation, proofs, simulations, and drafting throughout the paper, under the author's correction and complete verification.
This work was funded by AFOSR grant FA9550-23-1-0535 and NSF grant ECCS-2151525.

\end{document}